\newif\iflong
\longtrue
\documentclass{article}
\usepackage[utf8]{inputenc}
\usepackage[margin=1in]{geometry}
\usepackage{amsfonts,amsmath,amssymb,amsthm}
\usepackage{cite,comment,enumerate,hyperref}
\usepackage[colorinlistoftodos,textsize=small,color=red!25!white,obeyFinal]{todonotes}
\usepackage{csquotes}
\usepackage{booktabs}
\usepackage{multirow}

\newtheorem{theorem}{Theorem}
\newtheorem{corollary}[theorem]{Corollary}

\newtheorem{lemma}[theorem]{Lemma}

\theoremstyle{definition}
\newtheorem{definition}{Definition}
\theoremstyle{remark}	

\usepackage{xcolor,xspace}

\usepackage{algorithm}
\usepackage{algorithmic}

\usepackage{marginnote}
\usepackage{etoolbox}
\newtoggle{lmargin}
\newcommand{\alternatingtodo}[2][]{%
    \iftoggle{lmargin}%
    {%
        \todo[#1]{#2}%
        \togglefalse{lmargin}%
    }{%
        {%
            \let\marginpar\marginnote%
            \reversemarginpar%
            \todo[#1]{#2}%
        }%
        \toggletrue{lmargin}%
    }%
    \ignorespaces%
}

\DeclareMathOperator*{\E}{\mathbb{E}}
\DeclareMathOperator{\argmax}{argmax} 
\DeclareMathOperator{\argmin}{argmin} 

\newcommand{\density}{\rho}
\newcommand{\Geom}{\textsf{Geom}}

\newcommand{\AQ}{\mathcal{A}_Q}

\allowdisplaybreaks

\title{Almost Tight Bounds for Differentially Private Densest Subgraph}
\date{}
\author{Michael Dinitz\thanks{Department of Computer Science, Johns Hopkins University, Baltimore, MD.  \texttt{mdinitz@cs.jhu.edu}.  Supported in part by NSF grants CCF-1909111 and CCF-2228995.  Work partially done while a Visiting Researcher at Google Research New York.} 
\and Satyen Kale\thanks{Google Research.  \texttt{satyenkale@google.com}.}
\and Silvio Lattanzi\thanks{Google Research.  \texttt{silviol@google.com}.}
\and Sergei Vassilvitskii\thanks{Google Research.  \texttt{sergeiv@google.com}.}}

\begin{document}

\begin{titlepage}
\maketitle
\begin{abstract}
    We study the Densest Subgraph (DSG) problem under the additional constraint of differential privacy.  DSG is a fundamental theoretical question which plays a central role in graph analytics, and so privacy is a natural requirement.  All known private algorithms for Densest Subgraph lose constant multiplicative factors, despite the existence of non-private exact algorithms.  We show that, perhaps surprisingly, this loss is not necessary: in both the classic differential privacy model and the LEDP model (local edge differential privacy, introduced recently by Dhulipala et al.~[FOCS 2022]), we give $(\epsilon, \delta)$-differentially private algorithms with no multiplicative loss whatsoever.  In other words, the loss is \emph{purely additive}.  Moreover, our additive losses match or improve the best-known previous additive loss (in any version of differential privacy) when $1/\delta$ is polynomial in $n$, and are almost tight: in the centralized setting, our additive loss is $O(\log n /\epsilon)$ while there is a known lower bound of $\Omega(\sqrt{\log n / \epsilon})$.  
    
    We also give a number of extensions.  First, we show how to extend our techniques to both the node-weighted and the directed versions of the problem.  Second, we give a separate algorithm with pure differential privacy (as opposed to approximate DP) but with worse approximation bounds.  And third, we give a new algorithm for privately computing the optimal density which implies a separation between the structural problem of privately computing the densest subgraph and the numeric problem of privately computing the density of the densest subgraph.  
\end{abstract}
\thispagestyle{empty}

\end{titlepage}

\section{Introduction}

Differential privacy has become the \emph{de facto} standard for private data analysis.  It is often applied to data which is inherently numeric, but there is growing interest in providing solutions on more structural inputs, particularly on graphs.  In this setting, much work has focused on private computation of numerical functions of a graph, e.g., the number of triangles.  But there is nothing in the definition of differential privacy which requires numeric output, and indeed, it is often possible to privately compute \emph{objects} which are approximately optimal solutions (sometimes represented explicitly, sometimes implicitly) to combinatorial optimization problems on graphs.  This dates back to~\cite{NRS07,GLMRT10}, and there is now a relatively developed literature on private graph algorithms\cite{gupta2010differentially, farhadi2022differentially, fan2022private, arora2019differentially, dhulipala2022differential, cohen2022near, eliavs2020differentially}.

%A well-known drawback of the classical notions of differential privacy are that they require a trusted curator, i.e., they are \emph{centralized}.  But often it is preferable if users can make their data private \emph{before} sending it to a central curator, since this does not require trusting the curator with sensitive data.  This motivation led to the development of the \emph{local} differential privacy model, which has also been studied extensively, particularly in the context of numeric data~\cite{KLNRS11}.  In the presence of graph data, where the sensitive information is the existence (or not) of edges, this leads to a natural model known as \emph{local edge differential privacy} (LEDP)~\cite{dhulipala2022differential}.  In this model each node has access to its incident edges, but the curator initially knows only the node set $V$.  Then the curator sends queries to the nodes, who respond in ways that preserve the privacy of their incident edges.  This process repeats for a number of rounds, until eventually the curator outputs something. See Appendix~\ref{app:DP} for a formal definition.

In this work we study the \emph{densest subgraph problem} (DSG), one of the most fundamental algorithmic problems in data mining and graph analytics~\cite{chen2010dense, dourisboure2007extraction, kumar2006structure,angel2014dense, gibson2005discovering, saha2010dense, akiba2013fast}, in classic edge differential privacy and in the \emph{local edge differential privacy} (LEDP)~\cite{dhulipala2022differential} models.  In the DSG problem we are asked to find a subset $S$ of nodes which maximizes the number of edges in the subgraph induced by $S$ divided by the size of $S$, i.e., which maximizes half of the average induced degree. 
 This value is known as the \emph{density} of $S$, and is denoted by $\density(S) = |E(S)| / |S|$.  Densest subgraph has been studied extensively in the non-private setting, and is well-known to admit an exact polynomial-time algorithm based on flow~\cite{goldberg1984finding}, as well as a fast greedy $2$-approximation~\cite{Charikar00} and fast streaming, parallel and dynamic algorithms~\cite{angel2014dense, esfandiari2016brief, mcgregor2015densest, BKV12, ghaffari2019improved, epasto2015efficient, sawlani2020near, bhattacharya2015space}. 

Due to its importance in data mining and its multiple applications in diverse domains, including bioinformatics, network science, fraud detection, and social network analysis~\cite{cadena2018graph, hooi2016fraudar, cadena2016dense, rozenshtein2014event, gibson2005discovering}, differentially private versions of DSG have recently been developed~\cite{NV21, dhulipala2022differential, FHS22,DLL23}.  This line of work focuses on \emph{edge}-differential privacy, where the private information is the set of edges (not the set of nodes).  This setting is of practical importance for DSG as well as for other clustering problems.  In fact, in many practical scenarios one is interested in detecting something about the structure of the network without leveraging the private information contained in the edges. For example  in social network analysis one is interested in network statistics without revealing information about the connection between two specific nodes, or in messaging networks where one wants to analyze the networks without revealing the frequency of messages between two nodes. In many of these settings we might also want \emph{local} edge-differential privacy (LEDP), where rather than trusting a central curator, the individual nodes communicate with an untrusted curator in a way that protects their private edge information~\cite{dhulipala2022differential} (see Section~\ref{sec:prelims} for  definitions of these models).  For these reasons, even beyond DSG, several other problems have been studied in the local edge differential privacy setting or the edge differential privacy setting in the data mining and machine learning literature~\cite{bun2021differentially, cohen2022near, liu2022better, mohamed2022differentially, chen2023private}.

Letting $S^*$ denote the optimal solution, we say that an algorithm is an $(\alpha, \beta)$-approximation if it outputs a set $S$ with $\density(S) \geq \frac{1}{\alpha} \density(S^*) - \beta$.  There has been significant recent and concurrent work on private DSG with nontrivial approximation guarantees in most model variants: $\epsilon$-DP~\cite{dhulipala2022differential,FHS22}, $\epsilon$-LEDP~\cite{dhulipala2022differential,DLL23}, and $(\epsilon, \delta)$-DP~\cite{NV21}.  See Table~\ref{tab:results} for a summary of these results.  We note that the most interesting regime for $\delta$ is when it is inverse polynomial in $n$, so we will usually think of $\log(1/\delta)$ as being on the same order as $\log n$.

The most notable feature of all of the known results is that \emph{they all incur multiplicative loss}.  This is the case even though exact algorithms with no loss exist in the non-private setting~\cite{goldberg1984finding}.  The smallest multiplicative loss is $1+\eta$ (from~\cite{dhulipala2022differential}), which incurs an additional $O(\frac{1}{\epsilon} \log^4 n)$ additive loss and is not in the local model.  But while it is known that additive loss is necessary even if there is also multiplicative loss~\cite{FHS22}, there is no known lower bound which indicates that multiplicative loss is actually necessary.  This is the fundamental question that we attack: is it possible to design a private DSG algorithm with no multiplicative loss whatsoever, and still only small (polylogarithmic) additive loss?  If so, how small can we make the additive loss?
    
    %In other words, all known algorithms suffer both multiplicative and additive losses with the latter being  $O(\log^2 n)$ or higher. On the other hand, lower bounds suggest that no multiplicative loss is required, and the additive loss may be as small as %  
%
%    These results suggest a number of natural questions.  Most obviously, there is a large gap between the known upper and lower bounds.  This is true in the additive term (the lower bound is $\Omega(\sqrt{\log n / \epsilon})$ while the upper bound of~\cite{NV21} has additive $O(\log^2 n / \epsilon)$ loss), but more intriguingly, is also true in the multiplicative term.  Note that the lower bound does not imply that \emph{any} multiplicative loss is required unless the additive loss is at most 
%    $O(\sqrt{\log n})$.  So it may be possible to design an approximation which has no multiplicative loss whatsoever, and with additive loss as low as $ O\left(\sqrt{\frac{\log n}{\epsilon}}\right)$. %, despite the fact that all known algorithms \emph{do} have some multiplicative loss. Is this actually possible, or is multiplicative loss necessary? 

% Another intriguing question is in understanding the additional loss encurred by the $\epsilon$-LEDP model: the best known upper bound has a multiplicative loss of $4+\eta$, while in the centralized DP model there is a $(1+\eta)$-approximation (both of which also have polylogarithmic additive loss).  Can this gap be closed or improved?

    \begin{table}[t!]
    \begin{center}
    {\renewcommand{\arraystretch}{2.0}
    \begin{tabular}{|l|c|c|}
        \hline
       Privacy Model  & Approximation Ratio & Reference / Notes  \\
       \hline
       \multirow{2}{*}{
       $\epsilon$-DP} & $\left(2, O\left(\frac{1}{\epsilon} \log^{3.5} n\right)\right)$ & High probability version in ~\cite{FHS22}\\
       \cline{2-3}
       %$\epsilon$-DP 
       & $\left(1+\eta, O\left(\frac{1}{\epsilon} \log^4 n\right)\right)$ & ~\cite{dhulipala2022differential}\\
       \hline
       \hline
       \multirow{3}{*}{$\epsilon$-LEDP} & $\left(4 + \eta, O\left(\frac{1}{\epsilon} \log^3 n\right)\right)$ & \cite{dhulipala2022differential}\\
       \cline{2-3}
       &$\left(2, O\left(\frac{1}{\epsilon} \log n\right) \right)$ & \cite{DLL23} (Concurrent with this paper) \\
       \cline{2-3}
       %$\epsilon$-LEDP 
       & $\left(2 + \eta, O\left(\frac{1}{\epsilon} \frac{1}{\eta} \log^2 n\right)\right)$ & This work (Corollary~\ref{cor:simple-main})\\ 
       \hline
       \hline
       \multirow{3}{*}{$(\epsilon, \delta)$-DP} & $\left(2, O\left(\frac{1}{\epsilon} \log \frac{1}{\delta} \log n\right)\right)$ & \cite{NV21}\\    
       \cline{2-3}
       %$(\epsilon, \delta)$-DP 
       & $\left(1, O\left(\frac{1}{\epsilon} \sqrt{\log n\log \frac{n}{\delta}}\right)\right)$ & This work (Theorem~\ref{thm:centralized-main})\\
       \cline{2-3}
       & $\left( \alpha, \Omega\left(\frac{1}{\alpha} \sqrt{\frac{1}{\epsilon} \log n} \right)\right)$ for any $\alpha \geq 1$ & Lower bound: \cite{FHS22} (high probability, $\delta \leq 1/n$) \\
       & $\left(1, \Omega\left(\sqrt{\frac{1}{\epsilon} \log n}\right) \right)$ & Lower bound: \cite{NV21} (expectation, $\delta \leq 1/n$)  \\
       \hline
       \hline
       $(\epsilon, \delta)$-LEDP & $\left(1, O\left(\frac{1}{\epsilon} \log n\sqrt{\log \frac{1}{\delta}}\right)\right)$ & This work (Corollary~\ref{cor:main}) \\
       \hline
       
    \end{tabular}
    \caption{A summary of results on differentially private DSG. Different privacy models represent different assumptions on the attacker. Recall that $\epsilon$-LEDP is a stronger privacy guarantee than $\epsilon$-DP, which is, in turn stronger than $(\epsilon, \delta)$-DP. On the other hand, $(\epsilon, \delta)$-LEDP is stronger than $(\epsilon, \delta)$-DP, but is incomparable with $\epsilon$-DP. }
    }
    \label{tab:results}
    \end{center}
    \end{table}

    \subsection{Our Results and Techniques} \label{sec:results}

%    Our main results in the context of previous work are summarized in Table~\ref{tab:results}.

    \subsubsection{Main result: $(\epsilon,\delta)$-(L)EDP}
   Our main set of results gives a surprising answer to the above question: in both the $\left(\epsilon, \delta\right)$-DP model and the $\left(\epsilon, \delta\right)$-LEDP model, it is possible to design a private algorithm with no multiplicative loss whatsoever!  Slightly more carefully, we design a polynomial time algorithm in the LEDP model and prove that it is $(\epsilon, \delta)$-differentially private and that it returns a $\left(1, O\left(\frac{1}{\epsilon} \log n\sqrt{\log(1/\delta)}\right)\right)$-approximation with high probability.  Note that this is the first purely additive approximation in \emph{any} version of differential privacy.  

    We can then improve this LEDP algorithm in the classic $(\epsilon,\delta)$-DP setting to achieve nearly optimal (and state of the art) results in term of additive and multiplicative bound at the same time. In particular, we give a $\left(1, O\left(\frac{1}{\epsilon} \sqrt{\log n \log(n/\delta)}\right)\right)$-approximation with high probability.  We emphasize that this is almost tight: in the regime where $\delta$ is an inverse polynomial we are only a $\sqrt{\log n / \epsilon}$ factor away from the known lower bounds.

    While our focus is on achieving no multiplicative loss, we note that our algorithms also have better additive loss than any previous result.  And compared to the independent concurrent result of~\cite{DLL23}, we achieve the same additive loss in the centralized setting and only $\sqrt{\log(1/\delta)}$ more additive loss in the local setting, while not losing anything multiplicatively (compared to their multiplicative factor $2$ loss).  
    
    Finally, we observe that our algorithms are simple to run and implement: all of the complexity lies in the analysis bounding the exact amount of noise we have to add, and in the analysis of the resulting techniques. 

\paragraph{Techniques.}
    At a very high level, our algorithm uses the classical multiplicative weights method to privately solve an LP formulation of DSG.  However, standard settings and analyses of differentially private multiplicative weights do not suit our needs, and standard methods of privately solving LPs using multiplicative weights are also too weak for us.  So we need to design a new private multiplicative weights algorithm, and show via a complex and delicate analysis that it can be used in a noisy version of the classical Plotkin, Shmoys, Tardos framework~\cite{PST95} to privately solve the LP (at least approximately).  While differentially private algorithms have been introduced in the past to solve linear programs~\cite{HRRU14}, we are able to leverage the problem structure to obtain substantially stronger bounds on the quality of the solution. 
    
    We first show (in Section~\ref{sec:DP-hedge}) that the classical Multiplicative Weights Update/Hedge algorithm can be extended to handle \emph{noisy} losses (adversarial losses plus random noise) while still having low regret compared to the adversarial losses.  This allows us to give a differentially private version of Hedge. To the best of our knowledge this is the first analysis of the multiplicative update algorithm with noisy updates for differential privacy and we expect it to be a useful technique of independent interest.\footnote{Note that our paper is not the first paper analyzing multiplicative updates in the field of differential privacy~\cite{HR10,gupta2010differentially, gupta2012iterative, hardt2012simple}, although previous work focused on answering linear queries and the noise is added to determine whether to ``update'' in current round. Importantly, updates in this setting are always exact and so their analysis is orthogonal to our result and cannot be used in our context.}

    Then in Section~\ref{sec:LEDP} we use Hedge with noisy losses in combination with an idea from a recent result of Chekuri, Quanrud, and Torres~\cite{CQT22} in the non-private setting to get our main result.  Their goals were quite different: they wanted to give provable bounds on a heuristic called Greedy++ due to Boob et al.~\cite{Boob20} which seems to work extremely well in practice despite a lack of theoretical guarantees, and they wanted to extend beyond densest subgraph to the more general setting of ``Densest Supermodular Subset'' (DSS).  To do this, they showed that Greedy++ (and an extension to DSS which they call ``Super-Greedy++'') can be interpreted as actually being a multiplicative-weights algorithm which approximately solves a highly non-obvious LP formulation of DSG in the Plotkin, Shmoys, Tardos (PST)~\cite{PST95} LP-solving framework, resulting in a $(1+\epsilon)$-multiplicative approximation.

    Taking this algorithm as our starting point, we want to add privacy.  It turns out that we can again show that a simple combinatorial algorithm based on counting (noisy) degrees can be ``interpreted'' as running multiplicative-weights, and in particular as running Hedge with noisy losses in the PST framework.  Interestingly, there is a well-known differentially private version of multiplicative weights~\cite{HR10}, so one might hope that we can simply plug it into the algorithm and analysis of~\cite{CQT22}.  Unfortunately this fails, since it turns out that the versions of privacy and noise that we want are essentially orthogonal to those achieved by~\cite{HR10}.  In particular, in order to preserve privacy we need to add (carefully chosen) noise to the updates in multiplicative weights, and be able to precisely figure out the true optimum of the Lagrangean relaxation without any noise. In~\cite{HR10}, on the other hand, the authors update the weights \emph{exactly} but choose the query (i.e., solve the Lagrangean relaxation) only approximately.  Moreover, the algorithm of~\cite{HR10} requires knowledge of the ``true'' answer, and so cannot be implemented in the local model.  Similarly, there are well-known ways of privately solving LPs~\cite{HRRU14}, but these algorithms were developed for generic LPs, rather the specific LP arising from DSG, they do not give strong enough bounds for our purposes (and are also not in the local model).
        
    The limitations of~\cite{HRRU14} show the difficulty of ``simply'' applying Hedge with noisy losses to the PST framework: for generic LPs, this approach does not work (see the discussion at the beginning of Section~\ref{sec:LEDP}). % However, we are able to leverage the special structure in DSG to achieve strong bounds. 
    Fortunately, we can leverage the special structure of the DSG LP.  Even with that structure, since we have noisy updates in multiplicative weights, we end up with a natural tension. The more iterations that multiplicative weights runs the more accurate its answer (it converges at a rate of roughly $1/\sqrt{T}$, where $T$ is the number of iterations).  But the more iterations it runs the more noise we need to add to maintain privacy due to sequential composition.  It is not clear which of these ``wins out''. For $\epsilon$-DP, unfortunately the noise wins: running for $T$ iterations requires adding noise on the order of $T$, which dominates the $1/\sqrt{T}$ convergence.  But by moving to $(\epsilon, \delta)$-DP, we are able to add significantly less noise, on the order of $\sqrt{T}$.  So now we need to add $\sqrt{T}$ noise but get convergence at a rate of $1/\sqrt{T}$.  These almost exactly balance out, but we can show through a sufficiently delicate analysis that the noise required grows \emph{slower} than the accuracy obtained, giving us the claimed bounds.  

    \subsubsection{Extensions and Other Results}
    \paragraph{Weights.} In Section~\ref{sec:weighted} we show that an extension of our techniques allows us to give similar bounds in the presence of node weights, i.e., where there is a weight $c_v$ for every node $v$ and our goal is to find the set $S \subseteq V$ which maximizes $|E(S)| / \sum_{v \in V} c_v$.  However, we must assume that $c_v \geq 1$ for all $v \in V$.  We note that this is a relatively standard assumption in the literature, see, e.g., \cite{SW20}\footnote{We remark that it is obviously impossible to handle edge weights under edge differential privacy, because with arbitrary edge weights two neighboring databases could have arbitrarily different optimal densities.}.  Unfortunately, in the LEDP setting, we lose both an arbitrarily small multiplicative factor as well as an additional additive loss.  This arises from the fact that in the weighted setting, we have to \emph{actually run} multiplicative weights, unlike in the unweighted setting where we could show that a combinatorial algorithm can be ``interpreted'' as running multiplicative weights.  Doing this requires binary search over an unknown parameter, causing the extra loss. In the centralized setting, on the other hand, we do not suffer any multiplicative error and get the same additive error as in the centralized unweighted setting.

    \paragraph{Directed.} In Section~\ref{sec:directed} we show that a further extension of our main techniques allows us to design an $(\epsilon, \delta)$-LEDP algorithm for \emph{directed} densest subgraph~\cite{KV99,Charikar00,SW20}.  It was recently shown that in the non-private setting there is a reduction from the directed version to the undirected node-weighted version~\cite{SW20}, so we can use our extension to vertex-weighted graphs.  Unfortunately, this reduction requires node weights that are less than $1$, violating our assumption.  So we only get a weaker bound, with additional additive loss that depends on how ``balanced'' the optimal solution is (see Theorem~\ref{thm:directed-utility}).

    \paragraph{$\epsilon$-LEDP.} As a secondary result, in Section~\ref{sec:pure} we also give an algorithm with improved bounds in the $\epsilon$-LEDP model. We show a very simple algorithm which is $\epsilon$-LEDP and is a $(2+\eta, O(\frac{1}{\epsilon \eta} \log^2 n))$-approximation, and which moreover can be parallelized to run in $O(\frac{1}{\eta} \log n)$ rounds.  This is an improvement to the $\epsilon$-LEDP algorithm of~\cite{dhulipala2022differential}, and as a side benefit is significantly simpler than the previous $\epsilon$-DP algorithms. Its bounds are dominated by the concurrent work of~\cite{DLL23}, but on the other hand our algorithm runs in $O(\frac{1}{\eta} \log n)$ rounds, while~\cite{DLL23} requires $O(n)$ rounds.  % It does not quite match the state of the art $\epsilon$-DP algorithm of~\cite{dhulipala2022differential} since the multiplicative loss is $2+\eta$ rather than $1+\eta$, but it improves in the additive loss.  And compared to~\cite{FHS22}, it gives improved accuracy with a significantly simpler algorithm and analysis.   Most importantly, it decreases the gap between $\epsilon$-LEDP and $\epsilon$-DP.   

    Our main technical idea is to start with Charikar's sequential 2-approximation~\cite{Charikar00}, as was done in~\cite{NV21} and~\cite{FHS22}.  Both of the previous papers show that through very clever analysis, it is possible to prove privacy guarantees much stronger than would be obtained by simply using sequential composition for the $n$ iterations of Charikar's algorithm.  Instead of carefully reasoning about composition of noise as in those papers, we take a different and far simpler approach: we use the parallel version of Charikar's algorithm from~\cite{BKV12} to get an algorithm that only requires $O(\log n)$ rounds, and then use parallel composition~\cite{DR14}.  Since there is a small number of rounds we do not need to add much noise.  A few more relatively straightforward modifications allow us to do this in the $\epsilon$-LEDP model.

    \paragraph{Numeric approximation.} A natural question in many settings is whether it is possible to privately output a \emph{structure} with the same quality guarantees that we would get if we only cared about outputting the \emph{value} of the structure.  For example, it was recently shown that while the \emph{value} of the min $s-t$ cut can be computed privately with small noise (constant in expectation, logarithmic with high probability), actually outputting the cut itself requires a \emph{linear} additive loss~\cite{DMN23}.  One interpretation of our main result is that DSG does not exhibit this phenomenon: outputting the structure requires loss that is ``close'' to the noise required to output the value.

    But how close is close?  That is, how much loss is necessary when outputting the density of the densest subgraph, rather than the node set?  It is not hard to see that the density has global sensitivity of at most $1$, and so the Laplacian or Gaussian mechanism can be used to output a value with additive loss that is at most $1/\epsilon$ in expectation and at most $O(\log n /\epsilon)$ with high probability (for the Gaussian mechanism, assuming that $\delta$ is at most inverse polynomial in $n$).

    We show that it is actually possible to do better.  Intuitively, if the density is small then we are free to give an inaccurate answer, while if it is large then this implies (by the definition of density) that the optimal subgraph actually has many nodes and edges, and hence the sensitivity of the density is significantly less than $1$.  We formalize this intuition, showing that a simple variant of the propose-test-release framework~\cite{DL09} can be used to give algorithms that lose only $\sqrt{1/\epsilon}$ in expectation or $\sqrt{\log n /\epsilon}$ with high probability.  Note that the lower bound of~\cite{NV21} of $\Omega(\sqrt{\log n /\epsilon})$ for computing the densest subgraph actually holds even in expectation, and hence our upper bound of $O(\sqrt{1/\epsilon})$ provides a separation between the value and the structure.

    \subsubsection{Followup and Concurrent Work} Since the initial posting of this paper, there have been two pieces of concurrent or followup work, both of which have focused on the related \emph{$k$-core problem} but have also given results on DSG.  Due to their focus on $k$-core, their techniques are quite different than ours, and in particular cannot be used to obtain a private algorithm for DSG with purely additive loss.
    
    First, as discussed, \cite{DLL23} gave a $(2, O(\log n / \epsilon))$-approximation in the $\epsilon$-LEDP model.  This improves on the accuracy of our $\epsilon$-LEDP algorithm, but compared to to our main $(\epsilon, \delta)$-LEDP algorithm has a multiplicative $2$ loss rather than being purely additive (and has essentially the same additive loss).  Moreover, their algorithm takes a linear number of rounds, while our $\epsilon$-LEDP algorithm is arguably simpler and takes only $O(\log n)$ rounds.
    
    Second, \cite{HSZ24} gave similar but slightly weaker bounds.  They first obtain a $(2, O(\log^2 n))$-approximation in the $\epsilon$-LEDP model. This does not match~\cite{DLL23} or our main $(\epsilon, \delta)$-(LE)DP result), but it slightly beats our $2+\eta$ multiplicative loss in the $\epsilon$-LEDP model.  However, it requires a linear number of rounds.  They also give a $(4+\eta, O(\log n \log \log n))$-approximation in $O(\log^2 n)$ rounds.  This is a factor two multiplicatively worse than our $\epsilon$-LEDP algorithm while also being slower, but improves the additive loss to $O(\log n \log\log n)$ rather than our $O(\log^2 n)$.  

%\paragraph{Centralized Edge-DP.}  In the centralized setting we can give slightly improved bounds by simply combining our algorithms with the private selection mechanism of Liu and Talwar~\cite{liu-talwar}.  Essentially all of our algorithms involve running many repetitions to decrease the failure probability and/or try all of a possible set of parameters.  In the local model we can do this via sequential composition theorems, which cause extra loss in the approximation ratio in order to maintain privacy across all of the repetitions.  But in the centralized model we do not have to lose anything, since~\cite{liu-talwar} allows us to find the best solution across many repetitions with almost no privacy loss.  

%We note that a key idea of~\cite{liu-talwar} is the ability to keep intermediate computations secret, so we do not need to worry about their privacy implications; we only require that the eventual output be differentially private.  But in the local model, with an untrusted curator, all of the intermediate computations are essentially public and so need to satisfy differential privacy, so this approach does not work. 

\subsection{Preliminaries and Notation} \label{sec:prelims}

We use $\log(\cdot)$ to denote the natural logarithm. Binary logarithm is specified as $\log_2(\cdot)$. For a natural number $n$, we let $[n] := \{1, 2, \ldots, n\}$.

Given a graph $G = (V, E)$, and a subset $S \subseteq V$, let $\density(S) = |E(S)| / |S|$, where $E(S) \subseteq E$ is the set of edges induced by $S$ (i.e., with both endpoints in $S$). Let $\density(G) = \max_{S \subseteq V} \density(S)$.  

\begin{definition}
The \emph{Densest Subgraph Problem} (DSG) is the optimization problem in which we attempt to find an $S$ maximizing $\density(S)$, i.e., find an $S$ with $\density(S) = \density(G)$.  In the presence of node weights $\{c_v\}_{v \in V}$, we redefine the density to be $\density(S) = |E(S)| / \sum_{v \in S} c_v$.  If the edges are directed, then we let $E(S,T) = \{(u,v) \in E : u \in S, v \in T\}$ and define the density with reference to two sets: $\density(S,T) = |E(S,T)| /\sqrt{|S| \cdot |T|}$.  
\end{definition}

For a vertex $v \in V$ and a set $S \subseteq V$, let $d_S(v)$ denote the degree of $v$ into $S$, i.e., the number of edges incident on $v$ with other endpoint in $S$.  Let $d_{avg}(S) = \sum_{v \in S} d_S(v) / |S|$ denote the average degree in the subgraph induced by $S$.  Note that $d_{avg}(S) = 2 \cdot \density(S)$.

Given a randomized algorithm $\mathcal{A}$ for the DSG problem, differential privacy captures the impact of small changes in the input (captured by the neighboring relation) on the distribution of outputs. See the excellent book by Dwork and Roth ~\cite{DR14} for a thorough introduction to the topic. 

There are two natural notions of differential privacy on graphs: node differential privacy, where two graphs are considered neighboring if they differ in one node, and edge differential privacy, where two graphs are considered neighbors if they differ in one edge.  We will be concerned with edge differential privacy in this paper. % In the central edge-DP model, there is a trusted curator who has access to the entire graph.  
%This yields the following definition of edge-differential privacy.

\begin{definition}[Edge-Neighboring \cite{NRS07}] \label{def:edge-DP}
    Graphs $G_1 = (V_1, E_1)$ and $G_2 = (V_2, E_2)$ are \emph{edge-neighboring} if they differ in one edge, i.e., if $V_1 = V_2$ and the size of the symmetric difference of $E_1$ and $E_2$ is $1$.  
\end{definition}

\begin{definition}[Edge Differential Privacy \cite{NRS07}]
Algorithm $\mathcal A(G)$ that takes as input a graph $G$ and outputs an object in $R(\mathcal A)$ is \emph{$(\epsilon, \delta)$-edge differentially private} ($(\epsilon, \delta)$-edge DP) if for all $S \subseteq R(\mathcal A)$ and all edge-neighboring graphs $G$ and $G'$, 
\[
\Pr[\mathcal A(G) \in S] \leq e^{\epsilon} \Pr[\mathcal A(G') \in S] + \delta
\]
\end{definition}

\noindent If an algorithm is $(\epsilon, 0)$-edge DP then we say that it is $\epsilon$-edge DP (or $\epsilon$-DP). 

In the case of $(\epsilon, \delta)$-DP, most algorithms give a trade-off between the $\epsilon$ and $\delta$ they achieve. To make use of this trade-off in our analysis, we will use the concept of zCDP:

\begin{definition}[Zero-Concentrated Differential Privacy(zCDP) \cite{bun2016concentrated}]
Algorithm $\mathcal A(G)$ that takes as input a graph $G$ and outputs something in $R(\mathcal A)$ is \emph{$\rho$-zCDP} if for all $\alpha\in (0,1)$ and all edge-neighboring graphs $G$ and $G'$, 
\[
D_\alpha( \mathcal A(G) || \mathcal A(G') ) \leq \alpha\rho,
\]
\noindent where $D_\alpha$ is the R\'enyi divergence of order $\alpha$.
\end{definition}

Importantly, zCDP and $(\epsilon, \delta)$-edge DP are connected and a result in one setting can be translated in the other. In essentially all of our analysis we will use zCDP, and then we will translate into $(\epsilon, \delta)$-edge DP using standard results from~\cite{bun2016concentrated,Mironov17}.

%the following result of~\cite{bun2016concentrated}:
%\begin{lemma}[\cite{bun2016concentrated}] \label{lem:zcdp-to-dp}
%If a mechanism satisfied $\rho$-zCDP, then it satisfies $(\rho + 2\sqrt{\rho \log(1/\delta)}, \delta)$-DP for all $\delta > 0$.
%\end{lemma}

\paragraph{Local DP.}
Finally, all of the definitions above assume the existence of a trusted central curator who can execute the algorithm $\mathcal{A}$. 
In \emph{local} differential privacy, instead of assuming a trusted curator we assume that each agent controls its own data and there is an \emph{untrusted} curator who does not initially know the database, but can communicate with the agents.  We require that the entire \emph{transcript} of communication and outputs satisfies differential privacy, so even the curator cannot learn private information.  This model was originally introduced by~\cite{KLNRS11} in the context of learning, and has since become an important model of privacy (including in practice; see, for example, the discussion in~\cite{LDP-practice}).  In the context of graphs and edge-DP, this corresponds to the \emph{local} edge-differential privacy model from~\cite{dhulipala2022differential} (suitably modified to handle $(\epsilon, \delta)$-DP rather than just pure $\epsilon$-DP).  In this model there is a curator which initially knows only the vertex set (and thus the number of vertices $n$), and each node initially knows its incident edges, but nothing else (except possibly $n$, which can be sent by the curator initially).  During each round, the curator first queries a set of nodes for information.  Individual nodes, which have access only to their own (private) adjacency lists (and whatever information was sent by the curator), then send information back to the curator.  However, since we require that the entire transcript satisfy $(\epsilon, \delta)$-DP, without loss of generality all transmissions are really \emph{broadcasts} rather than point-to-point. 
% Nevertheless, it is sometimes simpler to think of point-to-point to communication.  
The actual formalization of this model is somewhat involved, so we defer it to Appendix~\ref{app:DP} (or to~\cite{dhulipala2022differential}).

We will use a number of standard DP mechanisms (the Laplacian, geometric, and Gaussian mechanisms in particular) and composition theorems (parallel composition, adaptive sequential composition, and advanced composition), as well as the standard post-processing theorem.  We give all of these formally in Appendix~\ref{app:DP}, but they are all well-known and can be found in standard textbooks~\cite{DR14,Vad17}.  It was observed by~\cite{dhulipala2022differential} that they all hold in the LEDP model as well (see Appendix~\ref{app:DP} for more discussion).  Since we use the Gaussian mechanism frequently, we will use $N(\mu, \tau^2)$ to denote the normal (Gaussian) distribution with mean $\mu$ and variance $\tau^2$ (so with standard deviation $\tau$).

\section{Hedge with Noisy Losses} \label{sec:DP-hedge}

In this section we analyze the classic Hedge/Multiplicative Weights Update (MWU) algorithm for the fundamental online learning setting of prediction with expert advice. The main twist here is that the expert losses can be noisy. This will be useful in developing our  algorithms for DSG which add noise for preserving privacy.

% in a setting where the loss v that handles noisy updates and so is able to preserve the privacy of the underlying dataset that is used to compute the update. This new algorithm is a key ingredient in our analysis and is also a useful tool of independent interest. 

The setting of online learning with expert advice is as follows. In each of $T$ rounds, indexed by $t = 1, 2, \ldots, T$, we have access to $n$ ``experts'', and are required to choose a distribution over the experts, after which losses of all experts are revealed. The losses are given by a random vector $\hat{m}^{(t)} \in \mathbb{R}^n$. Let $m^{(t)} = \mathbb{E}[\hat{m}^{(t)}\ |\ \hat{m}^{(1)}, \hat{m}^{(2)}, \ldots, \hat{m}^{(t-1)}]$. We assume that $m^{(t)} \in [-1, 1]^n$, and that the distribution $\hat{m}^{(t)} - m^{(t)}$, conditioned on all $\hat{m}^{(s)}$ for $s < t$, is sub-Gaussian with variance proxy $\nu^2 \leq 1$. In each round $t$, we are required to output a distribution on the experts $p^{(t)}$, and suffer the expected loss $\langle \hat{m}^{(t)}, p^{(t)}\rangle$. The regret of an online algorithm for this problem is the difference between the expected cumulative loss of the algorithm and the (expected) cumulative loss of the best fixed expert in hindsight. The goal is to develop an online algorithm with regret growing sublinearly in $T$.

Consider the classic Hedge algorithm~\cite{FreundS97} given as Algorithm~\ref{alg:DP-Hedge}. While it is usually used in the setting where losses are adversarial but bounded, we can provide a new regret bound for it in the above setting with unbounded but noisy losses. See Appendix~\ref{app:DP-hedge} for the proof of the following theorem.  The proof is essentially along the lines of the standard proof; the main change needed is the use log-sum-of-weights as a potential function and an application of Jensen's inequality.    
%\mdnoteinline{Maybe provide a line or two of intuition for what has to change compared to the usual proof of Hedge?}
%\sknoteinline{Added a line.}

\begin{algorithm}
    \caption{Hedge$(T)$}
    \label{alg:DP-Hedge}
    \begin{algorithmic}[1]
        \STATE Set $\eta = \sqrt{\frac{\log n}{T}}$.
        \STATE Set $w^{(1)}_i = 1$ for all $i \in [n]$.
        \FOR{$t = 1$ to $T$}
            \STATE Output the distribution $p^{(t)}$ such that $p^{(t)}_i \propto w^{(t)}_i$ for all $i \in [n]$.

            \STATE Observe $\hat{m}^{(t)}$.

            % \STATE Sample $\hat{m}^{(t)} \sim N(m^{(t)}, \nu^2 I)$.

            \STATE For all $i \in [n]$: set
            \[w^{(t+1)}_i = w^{(t)}_i \cdot \exp(-\eta \hat{m}^{(t)}_i).\]
        \ENDFOR
    \end{algorithmic}
\end{algorithm} 
% This algorithm is exactly the Hedge algorithm~\cite{FreundS97}, except that Gaussian noise is added to the loss vectors in order to preserve privacy of the underlying dataset via the Gaussian mechanism.
% We maintain a weight $w_i^{(t)}$ for each expert $i$ and time $t$.  All weights are initially set to $1$.  At time $t$, we choose an arm from the probability distribution $p^{(t)}$ where each arm gets probability proportional to its weight (i.e., $p^{(t)}_i = w^{(t)}_i / W^{(t)}$, where $W^{(t)} = \sum_{i \in [n]} w^{(t)}_i$). We then observe the loss vector $m^{(t)}$, and make it private via the Gaussian mechanism, i.e. we construct the noisy loss vector $\hat{m}^{(t)} = m^{(t)} + \nu^{(t)}$ where $\nu^{(t)} \sim N(0, \varsigma^2 I)$ for some parameter $\varsigma \geq 0$. We update the weights by setting $w_i^{(t+1)} = w_i^{(t)} \cdot \exp(-\eta  \hat{m}^{(t)}_i)$, where $\eta \leq 1$ is some parameter of the algorithm (this method of updating is known as \emph{Hedge updates}, hence the name). 

\begin{theorem} \label{thm:MWU-main}
% \textbf{(Utility bound)}
Suppose $\nu \leq 1$ and $T \geq \log n$. The Hedge algorithm guarantees that after $T$ rounds, for any expert $i \in [n]$, we have
% \begin{align*}
%     \E\left[\sum_{t=1}^T \langle \hat{m}^{(t)}, p^{(t)} \rangle\right] \leq \E\left[\sum_{t=1}^T \hat{m}_i^{(t)}\right] + 4\sqrt{\log(n) T},
% \end{align*}
% or, equivalently,
\begin{align*}
    \E\left[\sum_{t=1}^T \langle m^{(t)}, p^{(t)} \rangle\right] \leq \E\left[\sum_{t=1}^T m_i^{(t)}\right] + 4\sqrt{T\log(n) }.
\end{align*}
\end{theorem}
%\mdnoteinline{Reviewer is unclear why these are equivalent.  It's just linearity of expectations, right?}
%\sknoteinline{I just removed the first inequality since we derive the second one in the proof and that's the one we use in the proofs.}
\iflong
\fi

We now discuss the application of this algorithm in a privacy-preserving setting. Rather than thinking of the distributions of $\hat m^{(t)}$ as the fundamental object and the $m^{(t)}$ values as their expectations, we can equivalently think of the $m^{(t)}$ values as being the fundamental objects and can add noise to them to get distributions with the correct expectations.   So suppose that the loss vectors $m^{(t)}$ depend upon some underlying private dataset, and the past losses $m^{(s)}$ for $s < t$. 
Furthermore, assume that the $\ell_2$ sensitivity (w.r.t.\ changes of a single item in the underlying dataset) of the loss vectors is bounded by $\Delta$ (see~\cite{DR14} for a formal definition of $\ell_2$-sensitivity). The goal is to design an online learning algorithm with sublinear regret such that sequence of distributions over the experts output by the algorithm are collectively differentially private.
% \begin{enumerate}
%     \item \textbf{(Privacy)} 
%     \item \textbf{(Utility)} The expected regret of the algorithm % which is the difference between the total cost of the algorithm when it chooses an expert in each round sampled from the distribution it generates over the experts in that round, and the loss of the best fixed expert in hindsight, 
%     is sublinear in $T$. 
% \end{enumerate}

Consider the \textbf{DP-Hedge} algorithm, which simply  uses Hedge with the noisy losses $\hat{m}^{(t)}$ constructed by adding Gaussian noise to the true losses:
\[\hat{m}^{(t)}_i \sim N(m^{(t)}_i, \nu^2), \quad \forall i \in [n].\]
Clearly, $\hat{m}^{(t)}$ satisfies the assumptions of Theorem~\ref{thm:MWU-main}, so the regret bound holds. The privacy guarantee of DP-Hedge is given below:
\begin{theorem} \label{thm:DP-Hedge-privacy}
    % \textbf{(Privacy bound)} 
    The DP-Hedge algorithm is $\frac{ \Delta^2 T}{2\nu^2}$-zCDP.
\end{theorem}
\begin{proof}
    This is a direct consequence of the fact that the Gaussian mechanism (Lemma~\ref{lem:gaussian}) ensures that each $\hat{m}^{(t)}$ is $\frac{\Delta^2}{2\nu^2}$-zCDP, and then appealing to the adaptive sequential composition theorem for zCDP (Theorem~\ref{thm:adseqcomp}).
\end{proof}

% \textsc{Remark.} While in DP-Hedge Gaussian noise is explicitly added to the obtained loss vectors, from the analysis it should be clear that the same bound continues to hold if the noisy loss vectors $\hat{m}^{(t)}$ are obtained instead of the actual loss vectors $m^{(t)}$. The only assumption needed is that the distribution of $\hat{m}^{(t)} - m^{(t)}$, conditioned on all $\hat{m}^{(s)}$ for $s < t$, is sub-Gaussian with variance proxy $\nu^2 \leq 1$. This observation may have applications in other contexts.  In particular, if the algorithm directly observes the noisy loss vectors $\hat{m}^{(t)}$ and then does not add any noise, both the privacy and utility bounds still hold.  This will be useful when we use this algorithm in Section~\ref{sec:LEDP}.

Note that the regret bound given in Theorem~\ref{thm:MWU-main} is identical, up to constant factors, to that of the standard Hedge/Multiplicative Weights Update algorithm (see, e.g., \cite{AHK12}) despite potentially unbounded losses. This allows us to leverage Hedge with noisy losses in applications where losses can be unbounded but have bounded expectation, including specifically the Lagrangian relaxation approach to solving packing/covering LPs due to Plotkin, Shmoys and Tardos \cite{PST95}, which is relevant to the DSG problem.

%\mdnoteinline{STOC reviewer 1 did not like how we wrote this section.  In particular, they seem confused by the fact that the $m_t$ values are defines as conditional expectations at first, but then are defined as true losses in the DP setting.  Satyen, could you take a look and try to rewrite this a bit?
%
%The reviewer also didn't like that we claim the sequence is DP rather than the algorithm is DP. 
% Technically they're right, but this seems like a crazy objection.  Maybe what we should say to be fully accurate is that the \emph{distribution} over sequences of distributions is DP?  But that seems pretty
% confusing.}
% \sknoteinline{I am not sure how to rewrite things to make them more clear -- in the end, this is just math, and a reviewer should be able to understand it. I also updated the statement of Theorem 2 to say that DP-Hedge (rather than its output distributions) is DP-Hedge. I don't think we need to do more about this, the reviewer seems needlessly pedantic.}
% \begin{lemma} \label{lem:DP-Hedge-utility}
%     \textbf{(Utility bound)} Suppose $\varsigma \leq 1$. Setting $\eta = \sqrt{\frac{\log n}{T}}$, DP-Hedge has expected regret bounded by $4\sqrt{\log(n) T}$.
% \end{lemma}
% \begin{proof}
%     This follows directly from Theorem~\ref{thm:MWU-main}.
% \end{proof}

\section{Main Result: Purely Additive Private DSG} \label{sec:LEDP}

We now show our main result: it is possible to get purely additive logarithmic loss, even in the strong $(\epsilon, \delta)$-LEDP model.  As discussed in Section~\ref{sec:results}, at a very high level, our approach follows the non-private algorithm of~\cite{CQT22}, by replacing the version of multiplicative weights that they use with our DP-Hedge algorithm from Section~\ref{sec:DP-hedge}.  In particular, they use the Plotkin, Shmoys, and Tardos (PST) framework~\cite{PST95} to solve an LP formulation of the problem via multiplicative weights.  We similarly use PST on the same LP formulation, but with noise added to preserve privacy (i.e., with DP-Hedge).  While this may sound straightforward, the interaction between the noise needed for privacy and the ability of the PST framework to solve the LP is technical, and requires significantly extending the standard analysis of PST.  In fact, this is the main reason why we cannot give guarantees for pure $\epsilon$-LEDP using this method: the noise we would have to add to preserve pure differential privacy would overwhelm the ability of PST to find a good solution.

%Their algorithm can be thought of as two separate but related algorithms: one algorithm which uses the Plotkin, Shmoys, and Tardos framework~\cite{PST95} to solve an LP formulation of the problem via multiplicative-weights, and then a different (and simpler) algorithm which they analyze by showing that it is actually \emph{equivalent} to their multiplicative-weights algorithm.    % the In order to use their strategy we have to carefully modify both algorithms by adding the right amount of noisy to preserve privacy and approximation and to significantly extend the analysis to handle the noise.

%We follow this strategy, both for simplicity reasons but also for some important technical reasons.  The simpler algorithm actually requires slightly less knowledge about the input graph: in particular, it does not require knowledge of the optimal value $\lambda^*$.  In the non-private setting it is of course easy to use standard guess-and-double ideas to assume knowledge of $\lambda^*$ without loss of generality, but in a differentially private setting this comes at an additional accuracy cost.  Moreover, assuming knowledge of $\lambda^*$ makes the algorithm more centralized, and since one of our main goals is a \emph{locally} private algorithm, the simpler algorithm will be much easier to make local.

We are certainly not the first to use multiplicative weights to solve LPs privately: most notably, see~\cite{HRRU14}, which explored the limits of this approach.  Unfortunately, the results of~\cite{HRRU14} are not applicable in our setting.  Since they were attempting to solve an extremely general problem (privately solving LPs), they were not able to take advantage of any problem-specific structure, and so in many settings they were only able to give impossibility results.  In particular, for the type of LP that we will be considering, one edge change will result in two constraints having many coefficients change by one.  This is what they call a ``high-sensitivity constraint-private LP'', for which they were able to prove only lower bounds.  

But we of course are not concerned with general LPs, but rather the specific LP corresponding to DSG.  We show that this LP has some very nice features which, if we are careful enough in the analysis, allow us to solve it privately.  Most notably, the PST framework~\cite{PST95} requires an ``oracle'' to solve the Lagrangian relaxation of the problem.  This would normally require extra noise in order to compute it privately, and this extra noise could propagate throughout the computation to cause essentially a complete loss of utility.  But for the specific LP we use, it turns out that this Lagrangian problem can be solved \emph{exactly} even in the private setting!  

%\mdnoteinline{Maybe move some of the previous discussion to the intro}

We begin in Section~\ref{sec:CQT} with some background from~\cite{CQT22}; most importantly, the definition of the LP formulation.  We then show in Section~\ref{sec:MWU-private} that we can use DP-Hedge inside of the PST framework to solve this LP privately with purely additive loss, but with the limitation that we have to assume knowledge of the optimal value $\lambda^*$ and we only succeed with constant probability.  In Section~\ref{sec:peeling} we give a private version of the useful ``peeling'' primitive.  This allows us in Section~\ref{sec:alg-main} to finally give our full $(\epsilon, \delta)$-LEDP algorithm, and its full analysis in Section~\ref{sec:PST-analysis}.  Finally, in Section~\ref{sec:centralized-unweighted} we show how to give improved bounds in the centralized (rather than local) model.  

%We begin in Section~\ref{sec:peeling} with a useful subroutine: a differentially private version of the classical ``peeling'' method where we are given an ordering in which to peel.  After defining and analyzing this subroutine, we give our actual algorithm in Section~\ref{sec:alg-main} and analyze its privacy bounds.  The most difficult task is to prove its approximation guarantee, and to do this we first start in Section~\ref{sec:CQT} with some necessary background from~\cite{CQT22}, particularly the surprising LP formulations that they devised.  This allows us, in Section~\ref{sec:MWU-private}, to give a \emph{different} algorithm to approximately solve the LP relaxation based on the Plotkin, Shmoys, Tardos framework~\cite{PST95}, but with added noise.  This algorithm is not local, but we show in Section~\ref{sec:true-algorithm-analysis} that it essentially ``simulates'' our true algorithm, and hence our true algorithm inherits the same approximation guarantee.  \iflong \else All missing proofs can be found in Appendix~\ref{app:LEDP}. \fi

\subsection{Background From Chekuri, Quanrud, Torres~\texorpdfstring{\cite{CQT22}}{Chekuri, Quanrud, and Torres}} \label{sec:CQT}
In order to prove the approximation quality of our algorithm, we first need to give some background from~\cite{CQT22}.  They showed that a particularly simple non-private algorithm (essentially a variation of the Greedy++ algorithm of~\cite{Boob20}) can be interpreted as running multiplicative-weights in the dual of a particular LP formulation for DSG.  

We define some notation (mostly taken from~\cite{CQT22}).  Given an ordering $\sigma$ of the vertices $V$, let $q(\sigma) \in \mathbb{R}^n$ be the vector where the coordinate for $v \in V$ is $q(\sigma)_v = |\{\{u,v\} \in E : u \prec_{\sigma} v\}|$.  In other words, $q(\sigma)_v$ is the number of edges from $v$ to nodes that precede $v$ in $\sigma$.  Another way of thinking about this is that if we ``peeled'' (removed) the vertices in the \emph{reverse} order of $\sigma$, then $q(\sigma)_v$ would be the degree of $v$ when it is removed (see Section~\ref{sec:peeling}).  Given an ordering $\sigma$ and a vertex $v$, let $S^{\sigma}_v = \{u \in V : u \preceq_{\sigma} v\}$ be the vertices in the prefix of $\sigma$ defined by $v$, and let $E^{\sigma}_v = \{ \{x,y\} \in E : x,y \in S^{\sigma}_v\}$ be the edges induced by $S^{\sigma}_v$.  Given $\sigma$, let $S^*_{\sigma}$ be the prefix of $\sigma$ with maximum density, i.e., $S^*_{\sigma}$ is the $S^{\sigma}_v$ which maximizes $\density(S^{\sigma}_v) = |E^{\sigma}_v| / |S^{\sigma}_v|$.

A particularly useful lemma from~\cite{CQT22} is the following. 

\begin{lemma}[Lemma 4.2 of~\cite{CQT22}] \label{lem:opt-ordering}
For any $x \in [0,1]^V$, we have that $\sum_{\{u,v\} \in E} \min\{x_u, x_v\} = \min_{\sigma} \langle x, q(\sigma) \rangle = \sum_{v \in V} x_v q(\sigma)_v$.  Moreover, for every $x$ we have that  $\argmin_{\sigma} \langle x, q(\sigma) \rangle$ is the permutation $\sigma$ which orders the vertices in nonincreasing order of values $x_v$.
\end{lemma}

This allowed them to write the following covering LP formulation of DSG, where $\lambda^*$ is the optimal density and $S_V$ is the symmetric group of $V$ (i.e., the set of all permutations of $V$).  
\begin{align*}
    \min\ & \sum_{v \in V} x_v \\ 
    \text{s.t.}\ & \langle x, q(\sigma) \rangle \geq \lambda^* & \forall \sigma \in S_V \\
    & x_v \geq 0 & \forall v \in V
\end{align*}

It was shown in~\cite{CQT22} (and it is not too hard to see) that for the correct value of $\lambda^*$, this LP is feasible with objective $1$ (we refer the interested reader to~\cite{CQT22} for more discussion about this LP, or to Section~\ref{sec:weighted-lp} for a discussion of a related generalized LP for the weighted setting).  The dual of this LP is the following packing LP, which also (by strong duality) has objective value $1$ for the correct value of $\lambda^*$:
\begin{align*}
    \max\ & \lambda^* \sum_{\sigma \in S_V} y_{\sigma} \\
    \text{s.t.}\ & \sum_{\sigma \in S_V} q(\sigma)_v y_{\sigma} \leq 1 & \forall v \in V \\
    & y_{\sigma} \geq 0 & \forall \sigma \in S_V
\end{align*}

\subsection{Noisy PST to Solve LP Privately} \label{sec:MWU-private}

\subsubsection{The Algorithm} \label{sec:PST-alg}

We can now give the first algorithm that we analyze, which (as discussed) is what we would get if we run Hedge with noisy updates in the Plotkin, Shmoys, Tardos (PST) LP solving framework~\cite{PST95} applied to the dual packing LP.  In particular, there will be an ``expert'' for each constraint (node), and we will keep track of a weight for each expert.  We get Algorithm~\ref{alg:DSG-MWU}, which we call ``Noisy-Order-Packing-MWU''.  It takes two parameters: $T$ (the number of iterations to run) and $\tau$ (a parameter which controls the added noise and thus the privacy guarantee).  
%We will set $\eta$ later.  
%It is not hard to see that the added noise precisely corresponds to the noise added in DSG-LEDP-core, so it also satisfies $(\epsilon, \delta)$-differential privacy, and it can \emph{almost} be implemented in the local model: 
We note that we are assuming knowledge of $\lambda^*$, which we do not actually know.  We could do binary search over guesses of $\lambda^*$ (which is what we do in the weighted setting, see Section~\ref{sec:weighted}), but this would incur an additional loss.  Fortunately, we will be able to get around this by eventually using a simpler algorithm with the same behavior that does not need $\lambda^*$ (see Section~\ref{sec:alg-main}), so for now we will assume knowledge of $\lambda^*$ for convenience.  %only use this algorithm to \emph{analyze} our true algorithm from Section~\ref{sec:alg-main}, so we do not actually need to implement it locally.

%\sknoteinline{... till here.}

%\mdnoteinline{I removed $\nu$ as an argument to Hedge, since Hedge doesn't use it (it's just a bound on the variance of the loss distributions).  But in Algorithm 2 we were explicitly compute $\nu$ and passing it to Hedge.  Should probably replace with a proof that the variance is bounded by $\nu$.}
%\sknoteinline{Sounds good. I added a line in the proof of Theorem 5 showing that $\nu \leq 1$.}

\begin{algorithm}
    \caption{Noisy-Order-Packing-MWU$(T, \tau)$}
    \label{alg:DSG-MWU}
    \begin{algorithmic}[1]
        \STATE Set $\rho = \frac{n+\tau}{\lambda^*}$ and $\nu = \frac{\tau}{\rho \lambda^*}$.
        \STATE Instantiate Hedge$(T)$ with $n$ experts corresponding to nodes in $V$.
        % \STATE $w_v^{(1)} \leftarrow 1$ for all $v \in V$
        \FOR{$t = 1$ to $T$}
            \STATE Obtain the distribution $p^{(t)}$ on $V$ from Hedge.
            \STATE \label{line:oracle} Let $\sigma^{(t)}$ be the ordering of $V$ in nonincreasing order of $p^{(t)}$, breaking ties consistently, e.g., by node IDs.  Set $y^{(t)}_{\sigma^{(t)}} = 1/\lambda^*$, and set $y^{(t)}_{\sigma} = 0$ for all $\sigma \neq \sigma^{(t)}$ (this is the Lagrangean oracle from PST~\cite{PST95}).
            % , set $y^{(t)}_{\sigma^{(t)}} = 1/\lambda^*$, and set $y^{(t)}_{\sigma} = 0$ for all $\sigma \neq \sigma^{(t)}$.
            \STATE For each $v \in V$, set the loss to be
        \[
        \hat{m}_v^{(t)} \sim N(m_v^{(t)}, \nu^2), \text{ where } m_v^{(t)} = \frac{1}{\rho}\left( 1 - \sum_{\sigma \in S_V} q(\sigma)_v y_{\sigma}^{(t)}\right) = \frac{1}{\rho}\left( 1 - \frac{q(\sigma^{(t)})_v}{\lambda^*}\right).
        \]
        \vspace{-0.5cm}
            \STATE Supply $\hat{m}^{(t)}$ as the loss vector to Hedge.
            % Use MWU to update the weights according to the perceived costs: set
        %\begin{align*}
            % $w_v^{(t+1)} \leftarrow e^{-\eta \widehat m_v^{(t)}} \cdot w_v^{(t)}$
        %\end{align*}
        \ENDFOR
        \STATE Let $t$ be chosen uniformly at random from $[T]$.  
        \RETURN $(p^{(t)}, \sigma^{(t)})$.
    \end{algorithmic}
\end{algorithm} 

% We note that this algorithm, as phrased, computes the losses $m_v^{(t)}$ directly and then passes them to DP-Hedge.  This is not possible in the local model.  
%Clearly the noisy losses used in Hedge can be computed locally at each node, and hence this algorithm can be implemented in the LEDP setting.  
We note that since the algorithm runs Hedge but passes it noisy losses, we will be able to use our previous analysis from Section~\ref{sec:DP-hedge}. 

% each of which can then send the noisy losses back to the curator, and then DP-Hedge can just observe the noisy losses directly (see the remark at the end of Section~\ref{sec:DP-hedge}).  

%\mdnoteinline{Slightly weird mismatch: can't ``provide'' losses to DP-hedge since can't compute them.  Remark after DP-hedge should take care of this, but feels a little weird}
%\sknoteinline{Maybe we just say that the algorithm can be implemented in the local model since DP-Hedge only uses noisy losses and those can be computed locally by the nodes?}
%We emphasize that we \emph{do not} show that this algorithm preserves privacy, or is implementable in the LEDP model.  Instead, we will show that it gives a good approximation, and then show that because the noise we added is ``the same'' as our true algorithm DSG-LEDP, the two algorithms behave the same, and hence the approximation bound that we prove for this algorithm also applies to DSG-LEDP.

\paragraph{Oracle.} The PST framework requires an ``oracle'' that solves the Lagrangian relaxation.  In our setting, a crucial feature is that we can implement this oracle \emph{exactly} despite the noise added for privacy. Using the notation of~\cite{AHK12}, we let $\mathcal P$ be the polytope defined by the ``easy'' constraints: $\mathcal P = \{y : y_{\sigma} \geq 0 \text{ for all } \sigma \in S_V \text{ and } \sum_{\sigma \in S_V} y_{\sigma} = 1/\lambda^*\}$.  Let $w$ denote the current weights in Hedge, and $p$ the distribution obtained by dividing eacj weight by the sum of the weights.  Then we need to find a feasible solution for the Lagrangean relaxation, which in our case is the problem of finding a $y \in \mathcal P$ such that $\sum_{v \in V} w_v \sum_{\sigma \in S_V} q(\sigma)_v y_{\sigma} \leq \sum_{v \in V} w_v$.  Equivalently, in terms of $p$, we need to find a $y \in \mathcal P$ such that $\sum_{v \in V} p_v \sum_{\sigma \in S_V} q(\sigma)_v y_{\sigma} \leq 1$).  To do this, it is sufficient to find a $y$ minimizing the left hand side, which (after changing the order of summations) is the same as finding a $y$ minimizing $\sum_{\sigma \in S_V} y_{\sigma} \left(\sum_{v \in V} p_v q(\sigma)_v \right)$.  Since by the definition of $\mathcal P$ we must have that $\sum_{\sigma \in S_V} y_{\sigma} = 1/\lambda^*$, finding such a $y$ is the same thing as finding the $\sigma$ which minimizes $\sum_{v \in V} p_v q(\sigma)_v$ and setting that $y_{\sigma}$ to $1/\lambda^*$, and setting all others to $0$. But we know from Lemma~\ref{lem:opt-ordering} that this $\sigma$ is precisely the permutation which orders the vertices in nonincreasing order of $p$.  So we can actually compute this ordering even without knowing the values of $q(\sigma)_v$ (as these are degrees and so are private).  So our oracle (line~\ref{line:oracle} in the algorithm) computes such a $\sigma$ and sets $y$ appropriately.

\subsubsection{Analysis}

We now analyze this algorithm, showing both privacy and utility bounds.  We first show that Noisy-Order-Packing-MWU is edge-differentially-private, under the assumption that we know $\lambda^*$.  As with our earlier analysis of DP-Hedge, we will use R\'enyi differential privacy. 

\begin{lemma} \label{lem:nopmwu-privacy}
    Noisy-Order-Packing-MWU satisfies $\frac{T}{\tau^2}$-zCDP.
\end{lemma}
\begin{proof}
    This is essentially straightforward from our analysis of DP-Hedge.  Consider one iteration $t$ of Noisy-Order-Packing-MWU.  Since changing an edge can affect at most one value of $q(\sigma^{(t)})$, the $\ell_2$ sensitivity of the loss vector in round $t$ is at most $\Delta = \frac{1}{\rho \lambda^*}$. Furthermore, note that $\nu = \frac{\tau}{\rho \lambda^*}$. Thus, by Theorem~\ref{thm:DP-Hedge-privacy}, Noisy-Order-Packing-MWU satisfies $\frac{\Delta^2 T}{2\nu^2} = \frac{T (\rho \lambda^*)^2}{2(\rho \lambda^*)^2 \tau^2} = \frac{T}{2\tau^2}$
    zCDP.
\end{proof}

We now move on to our utility bounds.  We define a key quantity that will be needed in the analysis:
% , and give a bound on it via the values of $\rho$ and $\tau$ specified in Algorithm~\ref{alg:DSG-MWU}:
\begin{equation}
\alpha = 8\rho \sqrt{\frac{\log n}{T}}
\end{equation}

% \mdnoteinline{Need to set $\tau$ somewhere?}
% \sknoteinline{It is set to $\sqrt{T}\varsigma$ in DSG-LEDP. We should probably highlight this setting somewhere.}
% \mdnoteinline{Yeah, in the current organization DSG-LEDP doesn't happen until much later.  What from this current section actually requires $\tau$ to be set in a particular way?  And don't we earlier set $\varsigma$ to be $\tau / \rho \lambda^*$ so $\tau = \varsigma \rho \lambda^*$?  In which case I'm not sure we can also set it to $\sqrt{T} \varsigma$.}
% \sknoteinline{Well the current analysis doesn't require $\tau$ to be set to anything in particular; it holds for any setting of $\tau$. Regarding $\varsigma$--there was an unfortunate notation clash due to $\varsigma$ being used in two different ways: the variance parameter in DP-Hedge, and the variance parameter used in Peeling. I have now changed the notation in DP-Hedge and use $\nu$ to denote its variance parameter, so hopefully there's no confusion now.}

\begin{theorem} \label{thm:primal-solve-robust}
Suppose $\alpha \leq \frac{1}{2}$. Then with probability at least $\frac{1}{2}$ over the choice of an index $t \in [T]$ chosen uniformly at random, and over the randomness in the noise, the point 
% $(1+2\alpha) w^{(t)} / W^{(t)} = 
$(1+2\alpha) p^{(t)}$ is a feasible solution to the primal covering LP with objective value $1+2\alpha$.
\end{theorem}
\begin{proof}
We aim to apply Theorem~\ref{thm:MWU-main}, so we first verify that it assumptions are met. Note that the setting $\rho = \frac{n+\tau}{\lambda^*}$ ensures that $\nu = \frac{\tau}{\rho\lambda^*} \leq 1$, and for any $v \in V$, we have $m^{(t)}_v = \frac{1}{\rho}\left(1 - \frac{q(\sigma^{(t)})_v}{\lambda^*}\right) \in [-1, 1]$ since $q(\sigma^{(t)})_v \leq n$, as required.

The objective value is equal to $1+2\alpha$  by construction (for every $t$). So let $r$ be the probability that for a randomly chosen $t \in [T]$, the point 
% $(1+2\alpha) w^{(t)} / W^{(t)} = 
$(1+2\alpha) p^{(t)}$ is a feasible solution to the primal covering LP. 

For each round $t$, we have the following:
\begin{align*}
    \langle m^{(t)}, p^{(t)} \rangle &= \frac{1}{\rho} \sum_{v \in V} \left(\left(1- \sum_{\sigma \in S_V} q(\sigma)_v y^{(t)}_{\sigma}\right) p^{(t)}_v\right) \tag{definition of $m_v^{(t)}$} \\
    % &= \frac{1}{\rho} \sum_{v \in V} \left(w^{(t)}_v - w^{(t)}_v \sum_{\sigma \in S_V} q(\sigma)_v y^{(t)}_{\sigma} \right) \\
    &= \frac{1}{\rho} \left(1 - \sum_{v \in V} p_v^{(t)} \sum_{\sigma \in S_V} q(\sigma)_v y^{(t)}_{\sigma}\right).
\end{align*}
For all $t$, we know that $\frac{1}{\rho} \left(1 - \sum_{v \in V} p_v^{(t)} \sum_{\sigma \in S_V} q(\sigma)_v y_{\sigma}\right) \geq 0$, since the feasibility of the LP implies that the oracle always returns a $y$ such that the Lagrangean relaxation is satisfied, i.e., $1 \geq \sum_{v \in V} p_v^{(t)} \sum_{\sigma \in S_V} q(\sigma)_v y^{(t)}_{\sigma}$.  So for all $t$, we have that $\langle m^{(t)}, p^{(t)} \rangle \geq 0$.

Now if $(1+2\alpha) p^{(t)}$ is an \emph{infeasible} solution to the primal covering LP, then we can prove a stronger statement.  In particular, there exists a permutation $\sigma '$ that $\sum_{v \in V} p^{(t)}_v q(\sigma')_v < \lambda^* / (1+2\alpha)$.  Since $\sigma^{(t)}$ is the permutation which minimizes $\sum_{v \in V} p^{(t)} q(\sigma)_v$ (by the definition of the oracle and Lemma~\ref{lem:opt-ordering}), we have that
\begin{align*}
\sum_{v \in V} p^{(t)}_v q(\sigma^{(t)})_v &\leq \sum_{v \in V} p^{(t)}_v q(\sigma ')_v < \frac{\lambda^*}{1+2\alpha}. 
\end{align*}
Now the definition of $y^{(t)}$ (from the Oracle) implies that
\begin{equation} \label{eq:oracle1}
    \sum_{v \in V} p^{(t)}_v \sum_{\sigma \in S_V} q(\sigma)_v y^{(t)}_{\sigma}  = \sum_{v \in V} p^{(t)}_v q(\sigma^{(t)})_v / \lambda^* < \frac{1}{1+2\alpha}.
\end{equation}
So we have that 
\begin{align*}
    \langle m^{(t)}, p^{(t)} \rangle &= \frac{1}{\rho} \left( 1 - \sum_{v \in V} p_v^{(t)} \sum_{\sigma \in S_V} q(\sigma)_v y^{(t)}_{\sigma}\right) \geq \frac{1}{\rho} \left(1 - \frac{1}{1+2\alpha}\right) \tag{Eq.~\eqref{eq:oracle1}} = \frac{1}{\rho} \cdot \frac{2\alpha}{1+2\alpha} \geq \frac{\alpha}{\rho},
\end{align*}
since $\alpha \leq \frac{1}{2}$. Thus, for all $t$, we have:
\[\langle m^{(t)}, p^{(t)} \rangle \geq \frac{\alpha}{\rho} \cdot \mathbf{1}[(1+2\alpha)p^{(t)} \text{ is infeasible}].\]
Hence, recalling that $r$ is the probability that $(1+2\alpha)p^{(t)}$ is feasible for a randomly chosen index $t$, we get
\begin{equation} \label{eq:mw-lhs}
\E\left[\sum_{t=1}^T \langle m^{(t)}, p^{(t)} \rangle  \right] \geq \frac{\alpha}{\rho} \cdot \sum_{t=1}^T \Pr[(1+2\alpha)p^{(t)} \text{ is infeasible}] = \frac{\alpha}{\rho} \cdot (1-r)T.
\end{equation}
This gives us a bound on the LHS of the inequality in Theorem~\ref{thm:MWU-main}. 

To bound the RHS, we proceed as follows. Let $\bar y = \E[\frac{1}{T} \sum_{t=1}^T y^{(t)}]$, and fix any $v \in V$. By the structure of the oracle, we know that
\[
\sum_{t=1}^T \frac{q(\sigma^{(t)})_v}{\lambda^*} 
= \sum_{t=1}^T \sum_{\sigma \in S_V} q(\sigma)_v y^{(t)}_\sigma 
= \sum_{\sigma \in S_V} q(\sigma)_v \cdot \left(\sum_{t=1}^T y^{(t)}_\sigma\right).
\]
Thus we have
\begin{equation} \label{eq:mw-rhs}
\E\left[\sum_{t=1}^T m^{(t)}_v\right] = \E\left[\sum_{t=1}^T \frac{1}{\rho}\left(1 - \frac{q(\sigma^{(t)})_v}{\lambda^*} \right)\right] = \E\left[\frac{1}{\rho} \left(T - \sum_{\sigma \in S_V} q(\sigma)_v \cdot \left(\sum_{t=1}^T y^{(t)}_\sigma\right)\right)\right] = \frac{T}{\rho} \left(1 - \sum_{\sigma \in S_V} q(\sigma)_v \bar{y}_{\sigma}\right).
\end{equation}
Using \eqref{eq:mw-lhs} and \eqref{eq:mw-rhs} in Theorem~\ref{thm:MWU-main}, we get
\[\frac{\alpha}{\rho} \cdot (1-r)T \leq \frac{T}{\rho} \left(1 - \sum_{\sigma \in S_V} q(\sigma)_v \bar{y}_{\sigma}\right) + 4\sqrt{\log(n) T}.\]
Simplifying, and using the value $\alpha = 8\rho \sqrt{\frac{\log n}{T}}$, we get
\[ \sum_{\sigma \in S_V} q(\sigma)_v \bar{y}_{\sigma} \leq 1 + \frac{\alpha}{2} - \alpha \cdot (1-r). \]
Note that the above inequality holds for \emph{all} $v \in V$. We claim that this implies that $r \geq \frac{1}{2}$. Otherwise, the RHS above is strictly less than $1$. This implies that $\bar{y}$ is strictly feasible: none of the constraints are tight.  Thus we can scale it up by a small amount, to get a new solution (to the dual) which is also feasible but has strictly larger objective value than $\bar{y}$.  Since $\bar{y}$ has objective value $\lambda^* \cdot \langle \vec{1}, \bar{y}\rangle = \frac{\lambda^*}{T}\E[\langle \vec{1}, y^{(t)}\rangle] = 1$, by construction, this means that there is a feasible solution with objective value strictly larger than $1$.  But this is a contradiction, since the primal and dual both have objective value equal to $1$.
\end{proof}

\iflong
We will use the following result from~\cite{CQT22}.

\begin{lemma}[Lemma 4.3 of~\cite{CQT22}] \label{lem:primal-to-discrete}
Given a feasible solution $x$ to the primal LP, there exists a $\tau \in [0,1]$ such that the set $S_{\tau} =\{v \in V : x_v \geq \tau\}$ has density at least $\lambda^* / \sum_{v \in V} x_v$.
\end{lemma}
\fi

We can now prove our main theorem about Noisy-Order-Packing-MWU.  The proof is a relatively straightforward application of Theorem~\ref{thm:primal-solve-robust}\iflong \else, so we defer it to Appendix~\ref{app:LEDP}\fi.  Recall that for a permutation $\sigma$, we defined $S^*_{\sigma}$ to be the prefix of $\sigma$ with maximum density.
\begin{theorem} \label{thm:NOP-MWU}
Noisy-Order-Packing-MWU returns $(p, \sigma)$ such that $\sigma$ is just $V$ in nonincreasing order of $p$, and $\density(S^*_{\sigma}) \geq (1-2\alpha)\lambda^*$ with probability at least $1/2$.
\end{theorem}
\iflong
\begin{proof}
It is easy to see from the definition of Noisy-Order-Packing-MWU that it returns $(p, \sigma)$ such that $\sigma$ is just $V$ in nonincreasing order of $p$.

If $\alpha \geq 1/2$ then $(1-2\alpha)\lambda^* \leq 0$, and so the theorem is trivially true; any permutation $\sigma$ will work.  On the other hand, if $\alpha < 1/2$ then we can combine Theorem~\ref{thm:primal-solve-robust} and Lemma~\ref{lem:primal-to-discrete} to get that with probability at least $1/2$, Noisy-Order-Packing-MWU returns weights $p$ such that there is a $\tau \in [0,1]$ where the set $S_{\tau} = \left\{v \in V : (1+2\alpha) p_v \geq \tau \right\}$ has density at least 
\begin{align*}
\frac{\lambda^*}{1+2\alpha} &\geq \lambda^*(1 - 2\alpha).
\end{align*}
Since $\sigma$ is just non-increasing order of $p$, this implies that there is some prefix of $\sigma$ with the same density, i.e., $\density(S^*_{\sigma}) \geq (1-2\alpha)\lambda^*$ as claimed.
\end{proof}
\fi

\subsection{Subroutine: Peeling} \label{sec:peeling}
We now give a useful subroutine that will allow us to find the (approximate) best subset when we do peeling according to a given permutation $\sigma$, as well as an estimate of its density, in the LEDP model.  Consider Algorithm~\ref{alg:peeling}, Peeling($\sigma, \varsigma$).

\begin{algorithm}
    \caption{Peeling($\sigma, \varsigma$)}
    \label{alg:peeling}
    \begin{algorithmic}[1]
        \STATE The curator sends $\sigma$ to all nodes.
        \STATE Each node $v$ computes $\hat q(\sigma)_v = q(\sigma)_v + N_v$, where $N_v \sim N(0, \varsigma^2)$, and sends $\hat q(\sigma)_v$ to the curator.
        \STATE For each $x \in V$, the curator computes $\widehat \density(S^{\sigma}_x) = \frac{\sum_{v \in S^{\sigma}_x} \hat q(\sigma)_v}{|S^{\sigma}_x|}$.
        \STATE Let $u = \argmax_{x \in V} \widehat \density(S^{\sigma}_x)$.  The curator returns $(S^{\sigma}_u, \widehat \density(S^{\sigma}_u))$.
    \end{algorithmic}
\end{algorithm}

\begin{lemma} \label{lem:peeling-DP}
    Peeling($\sigma, \varsigma$) is $\frac{1}{2\varsigma^2}$-zCDP.
\end{lemma}
\iflong
\begin{proof}
    Since every edge can contribute to $q(\sigma)_v$ for exactly one $v$, the Gaussian mechanism (Lemma~\ref{lem:gaussian}) 
    and parallel composition (Theorem~\ref{thm:parallel-comp}) 
    imply that step 2 is $\frac{1}{2\varsigma^2}$-zCDP.  Steps 3 and 4 are post-processing, 
    % so we can apply Theorem~\ref{thm:post-processing} to 
    % get that 
    hence the entire algorithm is $\frac{1}{2\varsigma^2}$-zCDP.  Note that this algorithm is implementable in the local model, so is $\frac{1}{2\varsigma^2}$-zCDP. In fact, step 2 is local and the other steps are only aggregation and post-processing steps.
\end{proof}
\fi

Next, we analyze the output of Peeling:
\begin{lemma} \label{lem:peeling-approx}
    Let $S \subseteq V$ be the vertices returned by Peeling($\sigma, \varsigma$). For any constant $c > 0$, with probability at least $1-2n^{-c}$, we have $|\widehat \density(S) - \density(S)| \leq 2 \varsigma \sqrt{(1+c)\log n}$ and $ \density(S) \geq  \density(S^*_{\sigma}) - 2\varsigma \sqrt{(1+c)\log n}$.
    % , for suitably large constants in the $O(\cdot)$ notation depending on $c$.
\end{lemma}
\iflong
\begin{proof}
    Consider some $x \in V$.  Let $N = \sum_{v \in S^{\sigma}_x} N_v$ be the total noise added to nodes in $S^{\sigma}_x$.  Then $N$ is distributed as $N(0, |S^{\sigma}_x| \varsigma^2)$.  So Lemma~\ref{lem:gaussian-concentration} (the standard Gaussian concentration bound) implies that
    \begin{align*}
    \Pr[|N| \geq \sqrt{(1+c)\log n}   \sqrt{|S^{\sigma}_x|} \varsigma] &\leq 2\cdot \exp\left(-\frac{\left(\sqrt{(1+c)\log n} \sqrt{|S^{\sigma}_x|} \varsigma\right)^2}{|S^{\sigma}_x| \varsigma^2}\right) \\
    &= 2 \cdot \exp(-(1+c) \log n) = 2n^{-(1+c)}.
    \end{align*}
    Note that 
    \begin{align*}
        \widehat \density(S^{\sigma}_x) &= \frac{\sum_{v \preceq_{\sigma} x} \hat q(\sigma)_v}{|S^{\sigma}_x|} = \frac{\sum_{v \preceq_{\sigma} x} (q(\sigma)_x + N_x)}{|S^{\sigma}_x|} = \frac{|E^{\sigma}_x| + N}{|S^{\sigma}_x|} = \density(S^{\sigma}_x) + \frac{N}{|S^{\sigma}_x|}.
    \end{align*}
    Hence we have that
    \begin{align*}
        \Pr\left[|\widehat \density(S^{\sigma}_x) - \density(S^{\sigma}_x)| \geq \frac{2\sqrt{(1+c)\log n}  \varsigma}{\sqrt{|S^{\sigma}_x|}}\right] \leq 2n^{-(1+c)}
    \end{align*}
    Taking a union bound over all $x \in V$ implies that with probability at least $1-2n^{-c}$, we have $|\widehat \density(S^{\sigma}_x) - \density(S^{\sigma}_x)| \leq 2\sqrt{(1+c)\log n} \varsigma$ for all $x \in V$.  This clearly implies the lemma, since every prefix has estimated density within $2\sqrt{(1+c)\log n} \varsigma$ of its true density and the algorithm returns the prefix with the highest estimated density.
\end{proof}
\fi

\subsection{The Final Algorithm} \label{sec:alg-main}

We can now give our true algorithm.  We first give a ``simpler'' version of Noisy-Order-Packing-MWU which does not need to know $\lambda^*$, and which can easily be implemented in the local model.  We call this algorithm DSG-LEDP-core (Algorithm~\ref{alg:core}), and it is essentially a ``noisy'' version of the Greedy++ algorithms of~\cite{Boob20,CQT22}.  For $T$ iterations, we will repeatedly update loads on all of the nodes as a function of each node's (noisy) degree. These judiciously chosen updates are, in fact, simulating Noisy-Order-Packing-MWU.

Since the method only succeeds with constant probability we need to repeat the process $O(\log n)$ times.  When combined with Peeling, this gives our final algorithm, DSG-LEDP (Algorithm \ref{alg:final}). The $c$ parameter in DSG-LEDP is used to specify success probability of at least $1-n^{-c}$ in the utility bound (see Theorem~\ref{thm:dsg-ledp-utility}).  
%Proving that DSG-LEDP is $(\epsilon, \delta)$-LEDP is relatively straightforward, requiring only the standard bound on the Gaussian mechanism together with parallel composition, advanced composition, and post-processing.  

\begin{algorithm}
    \caption{DSG-LEDP-core($T, \tau$)}
    \label{alg:core}
    \begin{algorithmic}[1]
        \STATE The curator initializes $\ell_v^{(1)} = 0$ for all $v \in V$
        \FOR{$t=1$ to $T$}
            \STATE The curator computes the permutation $\pi^{(t)}$ of $V$ defined by ordering the nodes in non-increasing order of $\{\ell_v^{(t)}\}_{v \in V}$ (breaking ties in some consistent way, e.g., by node ID).  The curator then sends $\pi^{(t)}$ to each node.
            \STATE Each node $v$ computes $\hat q(\pi^{(t)})_v = q(\pi^{(t)})_v +  N_v^{(t)}$, where $N_v^{(t)} \sim N\left(0, \tau^2\right)$
            % is random Gaussian noise with variance $\frac{CT \log^2 (T/\delta)}{\epsilon^2}$.  
            \STATE Each node $v$ then sends $\hat q(\pi^{(t)})_v$ to the curator.
            \STATE The curator updates all loads by setting $\ell_v^{(t+1)} = \ell_v^{(t)} + \hat q(\pi^{(t)})_v$
        \ENDFOR
        \STATE The curator chooses $t$ uniformly at random from $[T]$
        \RETURN $\pi^{(t)}$
    \end{algorithmic}
\end{algorithm}
\begin{algorithm}
    \caption{DSG-LEDP($T, \varsigma, c$)}
    \label{alg:final}
    \begin{algorithmic}[1]
        \STATE Set $\tau = \sqrt{T}\varsigma$.
        % where $C$ is the constant depending on $c$ from the proof of Lemma~\ref{lem:privacy-overall}.
        \FOR{$i=1$ to $c\log_2 n$}
            \STATE Let $\pi^{(i)} \leftarrow $ DSG-LEDP-core$(T, \tau)$.
            \STATE Let $(S^{(i)}, \widehat \density(S^{(i)})) \leftarrow \text{Peeling}\left(\pi^{(i)}, \varsigma\right)$
        \ENDFOR
        \STATE The curator computes $i^* = \argmax_{i \in [c \log n]} \widehat \density(S^{(i)})$
        \RETURN $S^{(i^*)}$
    \end{algorithmic}
\end{algorithm}

We now show that there is a tight relationship between Noisy-Order-Packing-MWU and DSG-LEDP-core: they are essentially the same algorithm!  Slightly more carefully, we show that if they are provided with the same random string to use for their sampling, then they will construct the exact same sequence of orderings.  As a corollary, the distributions of their outputs are identical. %The proof of this theorem is a calculation of $w^{(t)}_v$ in Hedge with noisy losses to show that it is monotonically increasing with $\ell_v$\iflong \else; the formal proof can be found in Appendix~\ref{app:LEDP}\fi.

\begin{theorem} \label{thm:alg-relationship}
    If Noisy-Order-Packing-MWU and DSG-LEDP-core are provided with the same random string to use for sampling, then for every $t \in [T]$, the ordering $\sigma^{(t)}$ computed in Noisy-Order-Packing-MWU($T, \tau$) and the ordering $\pi^{(t)}$ computed in DSG-LEDP-core($T, \tau$) are the same.
\end{theorem}
\iflong
%\mdnoteinline{Add something in statement about having same randomness.}
\begin{proof}
    We use induction on $t$.  This is obviously true for $t=1$, since for all $v$ we have that $w_v^{(1)}  = \ell_v^{(1)} = 1$ and so $p_v^{(1)} = 1/n$ and both algorithms use the same consistent tiebreaking.  
    
    Now consider some $t > 1$.  The loss that Noisy-Order-Packing-MWU supplies to Hedge for node $v$ is, by construction, 
    \begin{align*}
        \hat m_v^{(t)} & \sim N\left(\frac{1}{\rho}\left( 1 - \frac{q(\sigma^{(t)})_v}{\lambda^*}\right), \left(\frac{\tau}{\rho \lambda^*}\right)^2\right). 
    \end{align*}
    By standard properties of the normal distribution, we know that this distribution is identical to 
    \begin{align*}
        & \frac{1}{\rho}\left( 1 - \frac{q(\sigma^{(t)})_v}{\lambda^*}\right) - N\left(0,\left(\frac{\tau}{\rho \lambda^*}\right)^2\right) 
        = \frac{1}{\rho}\left( 1 - \frac{q(\sigma^{(t)})_v}{\lambda^*}\right) - \frac{N(0, \tau^2)}{\rho \lambda^*}. 
    \end{align*}

    Hence we may assume without loss of generality that Noisy-Order-Packing-MWU compute $\hat m_v^{(t)}$ by sampling a value $J_v^{(t)} \sim N(0,\tau^2)$ and setting $\hat m_v^{(t)} = \frac{1}{\rho}\left( 1 - \frac{q(\sigma^{(t)})_v}{\lambda^*}\right) - \frac{J_v^{(t)}}{\rho \lambda^*}$.  Since $N_v^{(t)} \sim N(0, \tau^2)$ in DSG-LEDP-core, this implies that if the two algorithms are given the same random bits then $N_v^{(t)} = J_v^{(t)}$.  Thus $\hat m_v^{(t)} = \frac{1}{\rho}\left( 1 - \frac{q(\sigma^{(t)})_v +N_v^{(t)}}{\lambda^*}\right)$.
    
    Now by definition of the weight updates in Hedge used in Noisy-Order-Packing-MWU, we know that 
    \begin{align*}
        w^{(t)}_v &= \prod_{i=1}^{t-1} e^{-\eta \widehat m_v^{(i)}} = \exp\left(-\eta \sum_{i=1}^t \hat m_v^{(i)}\right) = \exp\left( -\eta \sum_{i=1}^{t-1} \frac{1}{\rho}\left( 1 - \frac{q(\sigma^{(i)})_v + N_v^{(i)}}{\lambda^*}\right) \right) \\
        &= \exp\left(-\frac{\eta}{\rho} (t-1) + \eta \sum_{i=1}^{t-1} \frac{q(\sigma^{(i)})_v + N_v^{(i)}}{\rho\lambda^*}\right) \\
        &= \exp\left(-\frac{\eta}{\rho} (t-1) + \eta \sum_{i=1}^{t-1} \frac{q(\pi^{(i)})_v + N_v^{(i)}}{\rho\lambda^*}\right) \tag{induction} \\
        &= \exp\left(-\frac{\eta}{\rho} (t-1) + \eta \sum_{i=1}^{t-1} \frac{\hat q(\pi^{(i)})_v}{\rho\lambda^*}\right) \tag{def of $\hat q(\pi^{(i)})_v$} \\
        &= \exp\left(-\frac{\eta}{\rho} (t-1) + \frac{\eta}{\rho\lambda^*} \ell_v^{(t)} \right). \tag{def of $\ell_v^{(t)}$}
    \end{align*}

  Since $\eta, \rho$, and $\lambda^*$ are independent of $v$, this means that $w_v^{(t)} < w_{v'}^{(t)}$ if and only if $\ell_v^{(t)} < \ell_{v'}^{(t)}$, and hence $p_v^{(t)} < p_{v'}^{(t)}$ if and only if $\ell_v^{(t)} < \ell_{v'}^{(t)}$.  Since $\sigma^{(t)}$ is by definition the ordering of $V$ in non-increasing order of $p_v^{(t)}$, and $\pi^{(t)}$ is the ordering of $V$ in non-increasing order of $\ell_v^{(t)}$, and we break ties in the same consistent way in both algorithms, this implies that $\sigma^{(t)} = \pi^{(t)}$. 
\end{proof}
\fi

\subsection{Final Analysis} \label{sec:PST-analysis}

Now that we have shown that DSG-LEDP-core and Noisy-Order-Packing-MWU are essentially the same algorithm, we can finally analyze our complete algorithm DSG-LEDP.
% showing that it is $(\epsilon, \delta)$-LEDP and gives the claimed utility bound.

\subsubsection{Privacy} 

We first discuss the privacy of the algorithm.  This is not quite direct from the fact that Noisy-Order-Packing-MWU is private, since Theorem~\ref{thm:alg-relationship} only implies that the \emph{outputs} are the same, not the computation itself, and the LEDP model requires that the full transcript be private (not just the output).  But the intuition and analysis is essentially identical.

\begin{lemma} \label{lem:DSG-LEDP-privacy}
    DSG-LEDP-core is $\frac{T}{2\tau^2}$-zCDP.
\end{lemma}
\begin{proof}
    Fix an iteration $t$.  Computing the permutation $\pi^{(t)}$ does not require using any private information, so it is private by post-processing (Theorem~\ref{thm:post-processing}).  For each $v$, the Gaussian Mechanism (Lemma~\ref{lem:gaussian}) implies that $\hat q(\pi^{(t)})$ is $1/2\tau^2$-zCDP.  Since every edge contributes to $\hat q(\pi^{(t)}))_v$ for a single $v$, parallel composition (Theorem~\ref{thm:parallel-comp}) then implies that the full vector $q(\pi^{(t)})$ is $1/2\tau^2$-zCDP, and then post-processing implies that iteration $t$ as a whole is $1/2\tau^2$-zCDP.  
    
    Finally, sequential composition (Theorem~\ref{thm:adseqcomp}) implies that all $T$ iterations combined is $T/2\tau^2$-zCDP, and then post-processing implies that DSG-LEDP-core is $T/2\tau^2$-zCDP.
\end{proof}

\begin{lemma} \label{lem:privacy-overall}
    DSG-LEDP($T, \varsigma, c)$ is $\frac{c\log_2(n)}{\varsigma^2}$-zCDP.
\end{lemma}
\iflong
\begin{proof}
    Note that DSG-LEDP runs $c\log_2 n$ copies of DSG-LEDP-core followed by a run of Peeling.  Lemma~\ref{lem:DSG-LEDP-privacy} implies that each run of DSG-LEDP-core is $\frac{T}{2\tau^2}$-zCDP. Since $\tau = \sqrt{T} \varsigma$, this can be rewritten as $\frac{1}{2\varsigma^2}$-zCDP. Lemma~\ref{lem:peeling-DP} shows that each call to Peeling is $\frac{1}{2\varsigma^2}$-zCDP. Thus, by the sequential composition for zCDP (Theorem~\ref{thm:adseqcomp}), we get that DSG-LEDP is $\frac{c\log_2(n)}{\varsigma^2}$-zCDP, as required.
\end{proof}

\begin{corollary} \label{cor:dsg-ledp-privacy}
    Let $\delta \in (0, 1)$ and $\epsilon \in (0, 8\log(1/\delta))$ be given privacy parameters. Set $\varsigma = \frac{4\sqrt{c\log_2(n) \log(1/\delta)}}{\epsilon}$ in DSG-LEDP. Then DSG-LEDP is $(\epsilon, \delta)$-LEDP.
\end{corollary}
\begin{proof}
The zCDP guarantee for DSG-LEDP translates to $(\alpha, \frac{\alpha c\log_2(n)}{\varsigma^2})$-RDP guarantee for any $\alpha \geq 1$, which in turn translates to $(\frac{\alpha c \log_2(n)}{\varsigma^2} + \frac{\log(1/\delta)}{\alpha-1}, \delta)$-LEDP for any $\alpha \geq 1$ (see Proposition 3 in \cite{Mironov17}). Choosing $\alpha = \varsigma\sqrt{\frac{\log(1/\delta)}{c\log_2(n)}} + 1$ for $\varsigma$ as specified above in the lemma statement, we get that DSG-LEDP is $(\epsilon, \delta)$-DP.

Furthermore it is easy to note that DSG-LEDP-core only does local operations. In fact, in step 5 each node computes a noisy estimate of its $q(\pi^{(t)})$ and all the remaining steps are post-processing done by a central coordinator. Thus DSG-LEDP is $(\epsilon, \delta)$-LEDP.
\end{proof}
\fi

\subsubsection{Utility}

We can now analyze the utility of our full algorithm, DSG-LEDP (Algorithm~\ref{alg:final}).  \iflong \else Again, the proof can be found in Appendix~\ref{app:LEDP}. \fi

\begin{theorem} \label{thm:dsg-ledp-utility} Suppose $T$ and $\varsigma$ in DSG-LEDP are set so that $\tau = \sqrt{T}\varsigma \geq n$. Then with probability at least $1 - 3n^{-c}$, DSG-LEDP returns a set $S$ with $\density(S) \geq \lambda^* - O\left(\sqrt{\log n}\varsigma\right)$.
\end{theorem}
\iflong
\begin{proof}
    Since DSG-LEDP-core chooses a random $t \in [T]$ and returns the permutation from it, Theorem~\ref{thm:alg-relationship} and Theorem~\ref{thm:NOP-MWU} imply that with probability at least $1/2$, DSG-LEDP-core returns a permutation $\sigma$ with 
    \[\density(S^*_{\sigma}) \geq (1-2\alpha)\lambda^* \geq \lambda^* - O\left(\sqrt{\log n}\varsigma\right),\]
    where the second inequality follows since
    \[
    \alpha = 8\rho \sqrt{\frac{\log n}{T}} 
    = O\left(\frac{n+\tau}{\lambda^*} \sqrt{\frac{\log n}{T}}\right) 
    = O\left(\frac{\sqrt{T}\varsigma}{\lambda^*} \sqrt{\frac{\log n}{T}}\right)
    = O\left(\frac{\sqrt{\log n}\varsigma}{\lambda^*}\right).
    \]
     Since DSG-LEDP runs $c \log_2 n$ independent copies of DSG-LEDP-core, we conclude that with probability at least $1 - n^{-c}$, there is at least one index $i \in [c \log_2 n]$ in which $\density(S^*_{\pi^{(i)}}) \geq \lambda^* - O\left(\sqrt{\log n}\varsigma\right)$.  We do not know which iteration this is, but DSG-LEDP then runs Peeling (Algorithm~\ref{alg:peeling}) on each of these permutations. Using Lemma~\ref{lem:peeling-approx}, we conclude that with probability at least $1-2n^{-c}$, for all $i \in [c \log_2 n]$, the set $S^{(i)}$ that we get from calling Peeling has both true and estimated density within $O\left(\sqrt{\log n}\varsigma\right)$. of $S^*_{\pi^{(i)}}$.  Thus with probability at least $1-3n^{-c}$, we have
    \begin{align*}
        \density(S) &\geq \max_{i \in [c \log n]} \density(S^*_{\pi^{(i)}}) - O\left(\sqrt{\log n}\varsigma\right) \\
        &\geq \lambda^* - O\left(\sqrt{\log n}\varsigma\right) - O\left(\sqrt{\log n}\varsigma\right) \\
        &\geq \lambda^* - O\left(\sqrt{\log n}\varsigma\right)
    \end{align*}
    as claimed.
\end{proof}
\fi
\begin{corollary} \label{cor:main}
    Let $\delta \in (0, 1)$ and $\epsilon \in (0, 8\log(1/\delta)$ be given privacy parameters. Set $\varsigma = \frac{4\sqrt{c\log_2(n) \log(1/\delta)}}{\epsilon}$ and $T = \lceil\frac{n^2}{\varsigma^2}\rceil$ in DSG-LEDP. Then DSG-LEDP is $(\epsilon,\delta)$-LEDP, and with probability at least $1 - 3n^{-c}$, DSG-LEDP returns a set $S$ with $\density(S) \geq \lambda^* - O\left( \frac{\log(n) \sqrt{\log(1/\delta)} }{\epsilon}\right)$.
\end{corollary}
\begin{proof}
The privacy guarantee follows from Corollary~\ref{cor:dsg-ledp-privacy}, and the output guarantee from Theorem~\ref{thm:dsg-ledp-utility}.
\end{proof}

\subsection{Improvement in the Centralized Setting} \label{sec:centralized-unweighted}

In this section we provide improved approximation guarantees in the centralized (not local) edge differential-privacy setting via the techniques of Papernot and Steinke~\cite{papernot-steinke}, which builds upon the work of Liu and Talwar~\cite{liu-talwar}. This will allow us to improve our additive loss in our main result. %, and will also allow us to remove all multiplicative loss in the weighted and directed settings.

The setup of that paper, adapted to our setting, is as follows. Suppose $Q$ is a randomized mechanism operating on datasets $D$ with output $Q(D)$ of the form $(s, q)$ where $s$ is the actual desired output (e.g., a subset of nodes in the DSG problem) and $q \in \mathbb{R}$ is a measure of its quality (e.g., the density of the output subset of nodes in the DSG problem) -- higher quality is more desirable. 
Then, the results of Papernot and Steinke~\cite{papernot-steinke} imply the following result (which we note requires the centralized model, since the intermediate computation which it does is not private, and so the requirement in the LEDP model that the entire transcript be DP would be violated): 

\begin{theorem} \label{thm:papernot-steinke}
Suppose $Q$ is $\rho$-zCDP. Given  $\gamma \in (0, 1)$, consider the algorithm $\AQ$ that samples $J$ from the standard geometric distribution with success probability $\gamma$, runs $J$ copies of $Q$ with independent random seeds and return the output $(s, q)$ of $Q$ with the highest value of $q$ among all outputs. Then for any $\delta \in (0, 1)$, the algorithm is $(6\sqrt{\rho \log(\frac{1}{\gamma\delta})}, \delta)$-DP. Furthermore, the quality of the output of $\AQ$ is at least $q^*$ with probability at least $1 - \frac{\gamma}{\Pr_{(s, q) \sim Q(D)}[q \geq q^*]}$.
\end{theorem}
\begin{proof}
In the terminology of Papernot and Steinke~\cite{papernot-steinke}, the mechanism $\AQ$ is obtained by using the distribution $\mathcal{D}_{1,\gamma}$, which is the geometric distribution with success probability $\gamma$. Then using Corollary 4 of their paper, we get that for any $\lambda \geq 1 + \sqrt{\frac{1}{\rho}\log(1/\gamma)}$, the mechanism $\AQ$ satisfies $(\lambda, \epsilon')$-RDP where
\[\epsilon' = \rho\cdot(\lambda-1) + \tfrac{1}{\lambda-1}\log(1/\gamma) + 4\sqrt{\rho\log(1/\gamma)} - \rho.\]
For any $\delta \in (0, 1)$, this translates (see Proposition 3 in \cite{Mironov17}) to 
\[\left(\rho\cdot(\lambda-1) + \tfrac{1}{\lambda-1}\log(1/\gamma) + 4\sqrt{\rho\log(1/\gamma)} - \rho + \tfrac{1}{\lambda-1}\log(1/\delta), \delta\right)\text{-DP}.\]
Setting $\lambda = 1 + \sqrt{\frac{1}{\rho}\log(\frac{1}{\gamma\delta})}$, this simplifies to
\[\left(2\sqrt{\rho\log(\tfrac{1}{\gamma\delta})}+4\sqrt{\rho\log(1/\gamma)} - \rho, \delta\right)\text{-DP}.\]
The privacy guarantee stated in the theorem statement is a weaker, but simpler, form of the above guarantee.

Next, the utility guarantee follows directly from Theorem 3.3 of \cite{liu-talwar}.
% , but follows from the same calculations as in Theorem 3.3 of that paper. Let $p := \Pr_{(s, q) \sim Q(D)}[q \geq q^*]$. WLOG, we may assume that $\gamma < p$; otherwise, the utility statement is trivial. If $(s, q)$ is the output of $\AQ$, then
% \begin{align*}
% \Pr[q < q^*] &= \sum_{j=1}^\Jmax \Pr[\tilde{J} = j] \cdot \Pr[q < q^* \; | \; \tilde{J} = j] \\
% &\leq \sum_{j=1}^{\Jmax-1} \gamma (1-\gamma)^{j-1} (1-p)^j + \left(1 - (1-\gamma)^{\Jmax-1}\right) (1-p)^{\Jmax} \\
% &= \gamma(1-p) \frac{1-(1-\gamma)^{\Jmax-1}(1-p)^{\Jmax-1}}{1-(1-\gamma)(1-p)} + \left(1 - (1-\gamma)^{\Jmax-1}\right) (1-p)^{\Jmax}\\
% &\leq \frac{\gamma}{p} + \delta_2.
% \end{align*}
% The last inequality uses the fact that since $\gamma < p$, we have $(1-p)^\Jmax \leq (1-\gamma)^\Jmax \leq \delta_2$.
\end{proof}

Clearly any mechanism in the LEDP model can also be run in the centralized model, so consider the following centralized mechanism which first runs DESG-LEDP-core (centralized) and then 
% Centralized-
Peeling:

\begin{algorithm}
    \caption{Centralized-DSG-core($T, \varsigma$)}
    \label{alg:centralized-dsg}
    \begin{algorithmic}[1]
        \STATE Set $\tau = \sqrt{T}\varsigma$.

        \STATE Run DSG-LEDP-core($T, \tau$) and obtain a permutation $\sigma$.

        \STATE Output $(S, \tilde \density(S)) \leftarrow$ Peeling($\sigma, \varsigma$).
    \end{algorithmic}
\end{algorithm}

\begin{lemma} \label{lem:centralized-privacy}
    Centralized-DSG-core($T, \varsigma)$ is $\frac{1}{\varsigma^2}$-zCDP.
\end{lemma}
\begin{proof}
    Since DSG-LEDP-core is equivalent to Noisy-Order-Packing-MWU, Lemma~\ref{lem:DSG-LEDP-privacy} implies that DSG-LEDP-core is $\frac{T}{2\tau^2}$-zCDP. Since $\tau = \sqrt{T} \varsigma$, this can be rewritten as $\frac{1}{2\varsigma^2}$-zCDP. Lemma~\ref{lem:peeling-DP} shows that the call to Peeling is $\frac{1}{2\varsigma^2}$-zCDP. Thus, by the composition results for zCDP~\cite{bun2016concentrated}, we get that DSG-LEDP is $\frac{1}{\varsigma^2}$-zCDP, as required.   
\end{proof}

\begin{lemma} \label{lem:centralized-dsg-ledp-utility} Suppose $T$ and $\varsigma$ in Centralized-DSG-core are set so that $\tau = \sqrt{T}\varsigma \geq n$. Then with probability at least $1/4$, Centralized-DSG-core($T, \varsigma$) returns a set $S$ with $\tilde\density(S) \geq \lambda^* - O\left(\sqrt{\log n}\varsigma\right)$.
\end{lemma}
\iflong
\begin{proof}
    Exactly as argued in Theorem~\ref{thm:dsg-ledp-utility}, with probability at least $1/2$, DSG-LEDP-core returns a permutation $\sigma$ with 
    \[\density(S^*_{\sigma}) \geq (1-2\alpha)\lambda^* \geq \lambda^* - O\left(\sqrt{\log n}\varsigma\right).\]
    Using Lemma~\ref{lem:peeling-approx} with $c = \frac{2}{\log n}$, we conclude that with probability at least $1-2e^{-2} \geq 1/2$, the set $S$ that we get from calling Peeling has both true and estimated density within $O\left(\sqrt{\log n}\varsigma\right)$ of $S^*_{\sigma}$.  Thus with probability at least $1/4$, we have
    \begin{align*}
        \tilde\density(S) &\geq \density(S^*_{\sigma}) - O\left(\sqrt{\log n}\varsigma\right) \\
        &\geq \lambda^* - O\left(\sqrt{\log n}\varsigma\right) - O\left(\sqrt{\log n}\varsigma\right) \\
        &\geq \lambda^* - O\left(\sqrt{\log n}\varsigma\right)
    \end{align*}
    as claimed.
\end{proof}
\fi

% Next, note that Theorem~\ref{thm:NOP-MWU}, Theorem~\ref{thm:alg-relationship}, and Theorem~\ref{thm:centralized-peeling} imply that with probability at least $1/8$, the following hold for the set $S$ and estimated density $\tilde{\density}(S)$ output by Centralized-DSG-core:
% \begin{equation} \label{eq:centralized-utility}
% \density(S) \geq \density(S^*_\sigma) - \frac{8 \ln(2n)}{\epsilon'} \geq (1-2\alpha) \lambda^* - \frac{8 \ln(2n)}{\epsilon'}
% \end{equation}
% and
% \begin{equation} \label{eq:centralized-estimated-density}
%     |\density(S) - \tilde{\density}(S)| \leq \frac{4 \ln(2)}{\epsilon'}
% \end{equation}

% Recall that $\epsilon' = \frac{\epsilon}{6}$ and $\delta' = \frac{\delta^2\gamma^2}{8 \log^2(1/\gamma)}$.  So for the specified value of $T$ and $\tau$, we obtain $\alpha = O\left(\frac{\sqrt{\log(n)} \log(n/\delta')}{\lambda^* \epsilon'}\right) = O\left(\frac{\sqrt{\log(n)} \log(n/\delta)}{\lambda^* \epsilon}\right)$, where the latter bound holds using the specified values of $\epsilon'$ and $\delta'$ in terms of $\epsilon$ and $\delta$, and assuming $\gamma = n^{-c}$ for some constant $c$. Hence Eq.~\eqref{eq:centralized-utility} and \eqref{eq:centralized-estimated-density} imply that
% \[\tilde\density(S) \geq \lambda^* - O\left(\frac{\sqrt{\log(n)} \log(n/\delta)}{\epsilon}\right) =: q^*.\]

Now, let $c > 0$ be a given failure probability parameter, and define $\gamma = n^{-c}$ for notational convenience. Let Centralized-DSG be the mechanism $\AQ$ obtained by applying the mechanism of Theorem~\ref{thm:papernot-steinke} with $\gamma = n^{-c}$ as specified above, $Q = $ Centralized-DSG-core, $s = S$, the set it outputs, and $q = \tilde{\density}(S)$, the estimated density it outputs, and $q^*$ set to $\lambda^* - O\left(\sqrt{\log n}\varsigma\right)$ as specified in Lemma~\ref{lem:centralized-dsg-ledp-utility}.  

\begin{theorem} \label{thm:centralized-main}
    Let $\delta \in (0, 1)$ and $\epsilon > 0$ be given privacy parameters. Set $\varsigma = \frac{6\sqrt{\log(n^c/\delta)}}{\epsilon}$ and $T = \lceil \frac{n^2}{\varsigma^2}\rceil$ in Centralized-DSG-core. Then Centralized-DSG is $(\epsilon, \delta)$-DP, and with probability at least $1 - 6n^{-c}$, the density of the set it outputs is at least $\lambda^* -O\left(\frac{\sqrt{\log(n)\log(n/\delta)}}{\epsilon}\right)$.
\end{theorem}
\begin{proof}
    With the notation as defined in the paragraph right before this theorem, Lemma~\ref{lem:centralized-dsg-ledp-utility} implies that
    \[\Pr_{(s, q) \sim Q(D)}[q \geq q^*] \geq \frac{1}{4}.\]
    Since Centralized-DSG-core($T, \varsigma$) is $\frac{1}{\varsigma^2}$-zCDP, applying Theorem~\ref{thm:papernot-steinke}, we conclude that Centralized-DSG is 
    \[
    \left(\tfrac{6\sqrt{\log(1/(\gamma\delta))}}{\varsigma}, \delta\right)\text{-DP} = (\epsilon, \delta)\text{-DP}
    \]
    for the specified value of $\varsigma$, and the estimated density $\tilde\density(S)$ of the set $S$ it outputs is at least $q^*$ with probability at least $1 - 4\gamma$. Now, recall the random variable $J$ used in Centralized-DSG which is drawn from the geometric distribution with success probability $\gamma$. Since $\Pr[J > k] = (1-\gamma)^k \leq \exp(-\gamma k)$, we conclude that $\Pr[J \leq \frac{\log(1/\gamma)}{\gamma}] \geq 1 - \gamma$. Conditioned on $J \leq \frac{\log(1/\gamma)}{\gamma}$, using Lemma~\ref{lem:peeling-approx} and a union bound over the $J$ calls to $Q$ in $\AQ$, we conclude that with probability at least $1-\gamma$, in each call to $Q$, we have
    \[|\density(S) - \tilde\density(S)| \leq O\left(\sqrt{\log n}\varsigma\right)\]
    for an appropriately chosen constant in the $O(\cdot)$ notation. Thus, overall, using the union bound, with probability at least $1-6\gamma$, the true density of the set output by Centralized-DSG is at least 
    \[\lambda^* - O\left(\sqrt{\log n}\varsigma\right) = \lambda^* - O\left(\frac{\sqrt{\log(n)\log(n/\delta)}}{\epsilon}\right). \qedhere\]
\end{proof}

\section{Node-Weighted Densest Subgraph} \label{sec:weighted}

\newcommand{\cmax}{C_{\max}}

Recall that in the node-weighted setting the edges are unweighted, but vertices can have weights.  To fix notation, we will say that every node $v$ has cost $c_v \geq 1$ and $\cmax = \max_{v \in V} c_v$.  We note that the assumption that every $c_v \geq 1$ is not without loss of generality: rescaling weights to all be at least $1$ will incur a cost due to the fact that we have an additive part to our approximation (if our approximation were purely multiplicative then this would be WLOG).  However, this is a relatively standard assumption; see, e.g., \cite{SW20}.  We let $c(S) = \sum_{v \in S} c_v$ for any $S \subseteq V$, and define the density as $\density(S) = |E(S)| / c(S)$.  Our goal is to find the densest subgraph.

This is in some sense a relatively minor change: as we will show, it is not too hard to adapt our version of noisy PST, Noisy-Order-Packing-MWU (Algorithm~\ref{alg:DSG-MWU}), to the node-weighted setting.  However, recall that in the unweighted setting, we didn't actually \emph{run} Noisy-Order-Packing-MWU.  We just showed that every iteration of it was identical to an iteration of our \emph{real} algorithm, DSG-LEDP-core (Algorithm~\ref{alg:core}).  The presence of node weights, unfortunately, destroys this approach: we cannot simply keep track of ``loads'' (i.e., degrees) and build a combinatorial algorithm using them which simulates the Hedge algorithm.  Intuitively, this is because in the proof of Theorem~\ref{thm:alg-relationship}, the expression for $w_v^{(t)}$ has in the exponent a term that is independent of $v$ and a term that depends on the degrees.  But in the presence of node weights, it turns out that the first term actually becomes a function of the node weight.  So to give the same ordering, we would have to trade these off in the same way as in the Hedge algorithm.  But that tradeoff depends on $\lambda^*$, which we do not actually know.  So we cannot simulate the Hedge algorithm via a combinatorial, degree-based algorithm.    

Instead, we will directly run a node-weighted version of Noisy-Order-Packing-MWU.  Unfortunately, doing so requires the curator to know $\lambda^*$.  If we were in the centralized edge-DP model we could of course just compute $\lambda^*$ and add noise to preserve privacy, but in the local model this is significantly more difficult.  Instead, we ``guess'' $\lambda^*$ and run our algorithm assuming our guess is correct.  Then, by running a grid search  for exponentially increasing guesses of $\lambda^*$, we are able to assume that we have an approximation of $\lambda^*$.  Compared to the unweighted case, this incurs a small multiplicative loss and an additional logarithmic additive loss.

\subsection{Peeling} \label{sec:peeling-weighted}

% \sknoteinline{Removed pseudocode and analysis of Weighted Peeling for conciseness.}

The Peeling algorithm can be easily modified to work for the node-weighted setting: 
% simply change the estimated density computation in line 3 to $\widehat \density(S^{\sigma}_x) = \frac{\sum_{v \in S^{\sigma}_x} \hat q(\sigma)_v}{c(S^{\sigma}_x)}$.
\begin{algorithm}
    \caption{Weighted-Peeling($G, \sigma, \varsigma$)}
    \label{alg:peeling-weighted}
    \begin{algorithmic}[1]
        \STATE The curator sends $\sigma$ to all nodes.
        \STATE Each node $v$ computes $\hat q(\sigma)_v = q(\sigma)_v + N_v$, where $N_v \sim N(0, \varsigma^2)$, and sends $\hat q(\sigma)_v$ to the curator.
        \STATE For each $x \in V$, the curator computes $\widehat \density(S^{\sigma}_x) = \frac{\sum_{v \in S^{\sigma}_x} \hat q(\sigma)_v}{c(S^{\sigma}_x)}$.
        \STATE Let $u = \argmax_{x \in V} \widehat \density(S^{\sigma}_x)$.  The curator returns $(S^{\sigma}_u, \widehat \density(S^{\sigma}_u))$.
    \end{algorithmic}
\end{algorithm}

Lemmas~\ref{lem:peeling-DP} and \ref{lem:peeling-approx} hold without change for Weighted-Peeling: the proof of the former is unchanged, and the proof of the latter follows using the fact that $c(S^{\sigma}_x) \geq |S^\sigma_x|$ since all node capacities $c_v$ are at least $1$. For concreteness, we provide the following lemma statements but omit proofs:
\begin{lemma} \label{lem:weighted-peeling-DP}
    Weighted-Peeling($\sigma, \varsigma$) is $\frac{1}{2\varsigma^2}$-zCDP for any $\alpha \geq 1$.
\end{lemma}
\begin{lemma} \label{lem:weighted-peeling-approx}
    Let $S \subseteq V$ be the vertices returned by Weighted-Peeling($\sigma, \varsigma$). For any constant $c > 0$, with probability at least $1-2n^{-c}$, we have $|\widehat \density(S) - \density(S)| \leq 2\sqrt{(1+c)\log n} \varsigma$ and $ \density(S) \geq  \density(S^*_{\sigma}) - 2\sqrt{(1+c)\log n} \varsigma$.
    % , for suitably large constants in the $O(\cdot)$ notation depending on $c$.
\end{lemma}

\subsection{Weighted versions of LPs from~\cite{CQT22}}
\label{sec:weighted-lp}
In the weighted context, the equivalent of Charikar's LP~\cite{Charikar00} is the following:
\begin{align*}
    \max\ & \sum_{e \in E} y_e \\
    \text{s.t.}\ & \sum_{v \in V} c_v x_v \leq 1 \\
    & y_{u,v} \leq x_u & \forall \{u,v\} \in E \\
    & y_{u,v} \leq x_v & \forall \{u,v\} \in E
\end{align*}

Note that since we are assuming that $c_v \geq 1$ for all $e \in E$, the first constraint implies that $x_v \leq 1$ for all $v \in V$.  It is not hard to see that this precisely characterizes the node-weighted DSG problem; we prove this for completeness.  
\begin{theorem}
    This LP is an exact formulation: for every set $S \subseteq V$ there is a feasible LP solution with objective at least $\density(S)$, and for every feasible solution $(x,y)$ there is a set $S$ with $\density(S) \geq \sum_{e \in E} y_e$.
\end{theorem}
\begin{proof}
    We begin with the first part.  Let $S \subseteq V$.  Set $x_v = \frac{1}{c(S)}$ for each $v \in S$, and $x_v = 0$ otherwise.  Set $y_{u,v} = \min(x_u, x_v)$ for each $\{u,v\} \in E$.  Then $\sum_{v \in V} c_v x_v = \sum_{v \in S} c_v / c(S) = 1$, so the solution is feasible, and $\sum_{e \in E} y_e = \sum_{e \in E(S)} 1/c(S) = |E(S)| / c(S) = \density(S)$ as claimed.

    For the other direction, let $(x,y)$ be a feasible solution.  For every $\tau \in [0,1]$, let $S_{\tau} = \{v \in V : x_v \geq \tau\}$.  We claim that there exists some $\tau$ such that $\rho(S_{\tau}) = |E(S_{\tau})| / c(S_{\tau}) \geq \sum_{e \in E} y_e$, which is enough to prove the theorem.  To see this, suppose for contradiction that it is false.  Then $\rho(S_{\tau}) < \sum_{e \in E} y_e$ for all $\tau$.  Consider choosing $\tau$ uniformly at random from $[0,1]$.  Then we have that 
    \begin{align*}
        \E[|E(S_{\tau})|] &< \E\left[ c(S_{\tau}) \sum_{e \in E} y_e \right]  = \left(\sum_{e \in E} y_e\right) \E[c(S_{\tau})]\tag{assumption for contradiction} \\
        &= \left( \sum_{e \in E} y_e \right) \sum_{v \in V} c_v x_v \tag{linearity of expectations} \\
        &\leq \sum_{e \in E} y_e \leq \sum_{\{u,v\} \in E} \min(x_u, x_v) \tag{feasible} \\
        &= \sum_{\{u,v\} \in E} \Pr[\{u,v\} \in E(S_{\tau})] \\
        &= \E[|E(S_{\tau})|] \tag{linearity of expectations}
    \end{align*}
    Since there is a strict inequality, this is a contradiction.  Hence some such $\tau$ exists (and we can find it efficiently since there are at most $n$ distinct values for $\tau$).  
\end{proof}

To put this into the language used by~\cite{CQT22}, we can define $\hat f(x) = \E[|E(S_{\tau})|]$.  Then the proof of the previous theorem implies that Charikar's LP is identical to the following:
\begin{align*}
    \max\ & \hat f(x) \\
    \text{s.t.}\ & \sum_{v \in V} c_v x_v \leq 1 \\
    & x_v \geq 0 & \forall v \in V
\end{align*}

Since $\hat f$ is independent of weights, Lemma~\ref{lem:opt-ordering} is still true.  Hence we have yet another equivalent formulation:
\begin{align*}
    \max\ & \min_{\sigma \in S_V} \langle x, q(\sigma) \rangle \\
    \text{s.t.}\ & \sum_{v \in V} c_v x_v \leq 1 \\
    & x_v \geq 0 & \forall v \in V
\end{align*}

Let $\lambda^*$ be the value of the optimal solution.  Then we can switch the objective and the constraint to get the equivalent of the LP from~\cite{CQT22}, i.e., the equivalent of the covering LP discussed in Section~\ref{sec:CQT}:
\begin{align*}
    \min\ & \sum_{v \in V} c_v x_v \\
    \text{s.t.}\ &\langle x, q(\sigma) \rangle \geq \lambda^* & \forall \sigma \in S_V \\
    & x_v \geq 0 & \forall v \in V
\end{align*}

By the above discussion and the definition of $\lambda^*$, we know that the optimal value of this LP is equal to $1$.  When we take its dual, we get the following packing LP, which by strong duality also has optimal solution equal to $1$.
\begin{align*}
    \max\ & \lambda^* \sum_{\sigma \in S_V} y_{\sigma} \\
    \text{s.t.}\ & \sum_{\sigma \in S_V} q(\sigma)_v y_{\sigma} \leq c_v & \forall v \in V \\
    & y_{\sigma} \geq 0 & \forall \sigma \in S_V
\end{align*}

This is the main LP which we will be arguing about.  Note that it's the exact same as in the unweighted case, except for the right hand side being the costs rather than just $1$.

\subsection{Noisy PST in the Presence of Weights}
Now let's apply the PST framework~\cite{PST95}, as described by~\cite{AHK12}, to the above LP in the presence of noise.  This is quite similar to Section~\ref{sec:PST-alg}; we just have to be careful how the costs on the right hand side affect the definitions and the algorithm. The big difference here is that unlike in the unweighted case we were able to get away without knowing the exact value of $\lambda^*$ in the algorithm due to the special structure of the problem, here we do not have a way to use the same strategy and we must resort to a search to guess the value of $\lambda^*$. To simplify the discussion, in the rest of this section, we discuss the analysis where we have a particular guess $\lambda$ for $\lambda^*$ and the consequences of it being too high or too low. Then in the next section we discuss the  grid search algorithm to get the correct $\lambda$.

\newcommand{\plp}{\text{PLP}}
\newcommand{\dlp}{\text{DLP}}

\paragraph{Feasibility LPs.} We use the following pair of feasibility LPs to decide whether the guess $\lambda$ is too high or too low. The primal LP is the following:
\begin{align}
    &\sum_{v \in V} c_v x_v = 1 \notag \\
    \text{s.t.}\ &\langle x, q(\sigma) \rangle \geq \lambda & \forall \sigma \in S_V \notag \\
    & x_v \geq 0 & \forall v \in V \tag{$\plp(\lambda)$} \label{eq:plp-lambda}
\end{align}
It is not hard to see based on the discussion in Section~\ref{sec:weighted-lp} that \eqref{eq:plp-lambda} is feasible if and only if $\lambda \leq \lambda^*$.  

To produce feasible solutions for the primal LP, we will actually use the following dual LP:
\begin{align}
    & \lambda \sum_{\sigma \in S_V} y_{\sigma} = 1 \notag \\
    & \frac{1}{c_v}\sum_{\sigma \in S_V} q(\sigma)_v y_{\sigma} \leq 1  & \forall v \in V \notag \\
    & y_{\sigma} \geq 0 & \forall \sigma \in S_V \tag{$\dlp(\lambda)$} \label{eq:dlp-lambda}
\end{align}
It is easy to see that \eqref{eq:dlp-lambda} is feasible if and only if $\lambda \geq \lambda^*$. Note that this is the exact opposite criterion as for the feasibility of \eqref{eq:plp-lambda}.

\paragraph{Oracle.}  To check for feasibility for \eqref{eq:dlp-lambda}, we proceed as in the unweighted case. Define $\mathcal P = \{y : y_{\sigma} \geq 0 \text{ for all } \sigma \in S_V \text{ and } \sum_{\sigma \in S_V} y_{\sigma} = 1/\lambda\}$. In each iteration of Hedge with noisy losses, if $p$ is the current distribution on $V$, then we need to find a feasible solution for the Lagrangean relaxation i.e., we need to find a $y \in \mathcal P$ such that $\sum_{v \in V} \frac{p_v}{c_v} \sum_{\sigma \in S_V} q(\sigma)_v y_{\sigma} \leq 1$.  So it is sufficient to find a $y$ minimizing the left hand side. Exactly as in unweighted case, this $y$ can be found by choosing the permutation $\sigma$ that orders nodes in nonincreasing order of $\langle \frac{p_v}{c_v}\rangle_v$, and setting $y_\sigma = \frac{1}{\lambda}$, and $y_\sigma' = 0$ for all permutations $\sigma' \neq \sigma$. This defines our Oracle.

% \paragraph{Perceived Costs.}  Let $y^{(t)} \in \mathcal P$ be the vector found by the Oracle at the beginning of iteration $t$, which corresponds to setting one permutation $\sigma^{(t)}$ to $y_{\sigma^{(t)}} = 1/\lambda$ and all other $y$ values to $0$.  As in the unweighted case, we define 
% the ``true'' cost of vertex $v$ to be 
% \begin{align*}
%     m_v^{(t)} = \frac{1}{\rho} \left(1 - \frac{1}{c_v}\sum_{\sigma \in S_v} q(\sigma)_v y_{\sigma}^{(t)} \right) = \frac{1}{\rho} \left(1 - \frac{q(\sigma^{(t)})_v}{c_v \lambda}\right).
% \end{align*}
% Node costs do not change anything about the noise necessary for privacy, so we still add Gaussian noise $N_v^{(t)}$ with appropriate standard deviation $\tau$ to the reported value of $q(\sigma^{(t)})_v$.  Thus we have ``perceived'' cost to vertex $v$ at time $t$ of
% \begin{align*}
% \widehat m_v^{(t)} = \frac{1}{\rho}\left(1 - \frac{1}{c_v}\sum_{\sigma \in S_V} (q(\sigma)_v + N_v^{(t)}) y_{\sigma}^{(t)}\right) = \frac{1}{\rho}\left(1 - \frac{q(\sigma^{(t)})_v + N_v^{(t)}}{c_v\lambda}\right)
% \end{align*}

% % Recall that $\cmax = \max_{v \in V} c_v$. 
% Setting the width parameter $\rho$ to
% \begin{align*}
%     \rho &= \frac{n+\tau}{\lambda}
% \end{align*}
% % \cmax + \frac{\Delta+\tau}{\lambda} 
% for an appropriately chosen value of parameters ensures that $|m^{(t)}_v| \leq 1$ and the standard deviation of $\hat{m}^{(t)}_v$ is bounded by $1$.
%  % $T = (\epsilon n \cmax)^2$ and the facts that $\Delta$ and $\lambda$ are at most $n$ (thanks to our assumption that all costs are at least $1$).

\paragraph{Algorithm.} Plugging these costs into the Plotkin-Shmoys-Tardos framework leads to a version of Noisy-Order-Packing-MWU for weighted graphs; see Algorithm~\ref{alg:Weighted-DSG-MWU}. 
Note that this algorithm can be directly implemented in the LEDP model just as in the unweighted setting.

\begin{algorithm}
    \caption{Weighted-Noisy-Order-Packing-MWU($G, \lambda, T, \tau$)}
    \label{alg:Weighted-DSG-MWU}
    \begin{algorithmic}[1]
        \STATE Set 
        % $\eta = \sqrt{\frac{\log n}{T}}$ 
        % $\tau = \frac{\tilde{C} \sqrt{T\log(\cmax n)\log^2(T/\delta) }}{\sqrt{\beta}\epsilon}$, 
        % and 
        $\rho = \frac{n+\tau}{\lambda}$ and $\nu = \frac{\tau}{\rho \lambda}$.
        \STATE Instantiate Hedge$(T, \nu)$ with $n$ experts corresponding to nodes in $V$.
        % \STATE $w_v^{(1)} \leftarrow 1$ for all $v \in V$
        \FOR{$t = 1$ to $T$}
            \STATE Obtain the distribution $p^{(t)}$ on $V$ from Hedge.
            \STATE Use the Oracle, as discussed: let $\sigma^{(t)}$ be the ordering of $V$ in nonincreasing order of $\langle \frac{p^{(t)}_v}{c_v}\rangle_v$ (breaking ties consistently, e.g., by node IDs)
            % , set $y^{(t)}_{\sigma^{(t)}} = 1/\lambda$, and set $y^{(t)}_{\sigma} = 0$ for all $\sigma \neq \sigma^{(t)}$.
            \STATE For each $v \in V$, set the  costs:
        \[
        \hat{m}_v^{(t)} \sim N(m_v^{(t)}, \nu^2), \text{ where } m_v^{(t)} = \frac{1}{\rho}\left(1 - \frac{1}{c_v}\sum_{\sigma \in S_V} q(\sigma)_v y_{\sigma}^{(t)}\right) = \frac{1}{\rho}\left(1 - \frac{q(\sigma^{(t)})_v}{c_v\lambda}\right).
        \]
        % where $N_v^{(t)} \sim N\left(0, \tau^2\right)$
            \STATE Supply $\hat{m}^{(t)}$ as the loss vector to Hedge.
            % \STATE Use MWU to update the weights according to the perceived costs.  In particular, set
        % \begin{align*}
        %     w_v^{(t+1)} \leftarrow e^{-\eta \widehat m_v^{(t)}} \cdot w_v^{(t)}
        % \end{align*}
        \ENDFOR
        \STATE Let $t$ be chosen uniformly at random from $[T]$.  
        \STATE Define $x^{(t)}_v := \frac{p^{(t)}_v}{c_v}$ for all $v \in V$.
        \RETURN $(x^{(t)}, \sigma^{(t)})$.
    \end{algorithmic}
\end{algorithm}

\paragraph{Analysis.}
We're going to follow the analysis for the unweighted case from Section~\ref{sec:PST-analysis} 
and prove the following theorem:
\begin{theorem} \label{thm:weighted-primal-solve-robust}
Let $L \geq 4(n +\tau)\sqrt{\frac{\log n}{T}}$. Suppose $\lambda^* \leq \lambda < \lambda^* + L$. Then with probability at least $\frac{1}{2}$ over the choice of an index $t \in [T]$ chosen uniformly at random, and over the randomness in the noise, the output point $x$ of Weighted-Noisy-Order-Packing-MWU$(G, \lambda, T, \tau)$ is a feasible solution to $\plp(\lambda - 4L)$.
\end{theorem}
\begin{proof}
First, note that if $\lambda < 4L$, then the statement is true with probability $1$ since any solution $x$ returned by Weighted-Noisy-Order-Packing-MWU is feasible. So in the following, we assume that $\lambda > 4L$.

Since $\lambda^* \leq \lambda < \lambda^* + L$, \eqref{eq:dlp-lambda} is feasible. Now we follow the analysis in the proof of Theorem~\ref{thm:primal-solve-robust}. Most calculations are identical, so we skip most details. Let $r$ be the probability that the output $x^{(t)}$ is infeasible for $\plp(\lambda - 4L)$. Now, for any round $t$, we have the following:
\begin{align*}
    \langle m^{(t)}, p^{(t)} \rangle &= \sum_{v \in V} m_v^{(t)} p_v^{(t)} = \frac{1}{\rho} \sum_{v \in V} \left(\left(1- \frac{1}{c_v}\sum_{\sigma \in S_V} q(\sigma)_v y^{(t)}_{\sigma})\right) p^{(t)}_v\right) \\
    &= \frac{1}{\rho} \sum_{v \in V} \left(1 - \frac{p^{(t)}_v}{c_v} \sum_{\sigma \in S_V} q(\sigma)_v y^{(t)}_{\sigma}) \right) \\
    &= \frac{1}{\rho} \left( 1 - \sum_{v \in V} \frac{p_v^{(t)}}{c_v} \sum_{\sigma \in S_V} q(\sigma)_v y^{(t)}_{\sigma}\right).
\end{align*}
Since \eqref{eq:dlp-lambda} is feasible, the oracle (by definition) always returns a $y$ such that 
\[\sum_{v \in V} \frac{p_v^{(t)}}{c_v} \sum_{\sigma \in S_V} q(\sigma)_v y_{\sigma} \leq 1.\] 
Hence, we conclude that $\langle m^{(t)}, p^{(t)} \rangle  \geq 0$ for all $t$. 

Next, suppose $t$ is a round where $x^{(t)}$ is infeasible for $\plp(\lambda - 4L)$. Note that by definition, $\sum_v x^{(t)}_vc_v = \sum_v p^{(t)}_v = 1$. Thus, the infeasibility must arise because there exists a $\sigma' \in S_V$ such that $\langle x^{(t)}, q(\sigma')\rangle < \lambda - 4L$, i.e., $\sum_v \frac{p^{(t)}_v}{c_v}q(\sigma'')_v < \lambda - 4L$.
% or, in other words, $\langle p^{(t)}, q(\sigma') \rangle < (\lambda^* - 4L) \sum_v p^{(t)}_v c_v$. 
Since $\sigma^{(t)}$ is the permutation that minimizes $\sum_v \frac{p^{(t)}_v}{c_v}q(\sigma'')_v$ over all $\sigma'' \in S_V$, we must have $\sum_v \frac{p^{(t)}_v}{c_v}q(\sigma^{(t)})_v < (\lambda - 4L)$, which is equivalent to 
\[\sum_{v \in V} \frac{p_v^{(t)}}{c_v} \sum_{\sigma \in S_V}  q(\sigma)_v y_{\sigma^{(t)}} < 1 - \frac{4L}{\lambda}.\] This implies that
\begin{align*}
    \langle m^{(t)}, p^{(t)} \rangle &= \frac{1}{\rho} \left( 1 - \sum_{v \in V} \frac{p_v^{(t)}}{c_v} \sum_{\sigma \in S_V} q(\sigma)_v y_{\sigma^{(t)}}\right) \geq \frac{4L}{\lambda\rho}.
\end{align*}
% where the last inequality uses the facts that  $c_v \geq 1$ for all $v \in V$. 
% and $\beta \leq \frac{1}{8}$. 
Thus, for all $t$, we have, exactly as in the proof of Theorem~\ref{thm:primal-solve-robust} that
\[\langle m^{(t)}, p^{(t)} \rangle \geq \frac{4L}{\lambda\rho} \cdot \mathbf{1}[x^{(t)} \text{ is infeasible}],\]
and hence,
\begin{equation} \label{eq:weighted-mw-lhs}
\E\left[\sum_{t=1}^T \langle m^{(t)}, p^{(t)} \rangle  \right] \geq \frac{4L}{\lambda\rho} \cdot \sum_{t=1}^T \Pr[x^{(t)} \text{ is infeasible}] = \frac{4L}{\lambda\rho} \cdot (1-r)T.
\end{equation}

Then following the proof of Theorem~\ref{thm:primal-solve-robust} from that point, we end up with the following version of \eqref{eq:mw-rhs}:
\begin{equation} \label{eq:mw-rhs-weighted}
\E\left[\sum_{t=1}^T m^{(t)}_v\right] = \E\left[\sum_{t=1}^T \frac{1}{\rho}\left(1 - \frac{q(\sigma^{(t)})_v}{c_v\lambda} \right)\right] = \E\left[\frac{T}{\rho} \left(1 - \frac{1}{c_v}\sum_{\sigma \in S_V} q(\sigma)_v \bar y_{\sigma}\right)\right] = \frac{T}{\rho} \left(1 - \frac{1}{c_v}\sum_{\sigma \in S_V} q(\sigma)_v \bar{y}_{\sigma}\right).
\end{equation}
Continuing with the analysis, we get the following bound via Theorem~\ref{thm:MWU-main} for $\bar{y} = \E[\frac{1}{T} \sum_{t=1}^T y^{(t)}]$:
\[\frac{4L}{\lambda\rho} \cdot (1-r)T \leq \frac{T}{\rho} \left(1 - \frac{1}{c_v}\sum_{\sigma \in S_V} q(\sigma)_v \bar{y}_{\sigma}\right) + 4\sqrt{\log(n) T}.\]
Simplifying, we get
% and using the values $T = (\epsilon n \cmax)^2$, $\eta = \sqrt{\frac{\log n}{T}}$ and
% $\rho = \frac{3\tau}{\lambda}$, we get
\[ \frac{1}{c_v}\sum_{\sigma \in S_V} q(\sigma)_v \bar{y}_{\sigma} \leq 1 + 4\rho\sqrt{\frac{\log n}{T}} - \frac{4L}{\lambda} \cdot (1-r) 
% \leq c_v + \frac{12\tau}{\lambda} \sqrt{\frac{\log n}{T}} -  \frac{4L}{\lambda} \cdot (1-r) 
\leq 1 + \frac{L}{\lambda} - \frac{4L}{\lambda} \cdot (1-r), \]
where the last inequality follows from the assumption that $L \geq 4(n +\tau)\sqrt{\frac{\log n}{T}}$ and the setting $\rho = \frac{n+\tau}{\lambda}$
 % and $c_v \geq 1$.
 % and $\lambda \geq \lambda^*$. 
Note that the inequality displayed above holds for \emph{all} $v \in V$. We claim that this implies that $r \geq \frac{1}{2}$. Otherwise, the RHS above is less than $1-\frac{L}{\lambda}$. Now consider $\tilde{y} = \frac{\lambda}{\lambda-L}\bar{y}$ and define $\lambda' = \lambda - L < \lambda^*$. Note that $\dlp(\lambda')$ is infeasible. But $\tilde{y}$ satisfies all the inequality constraints of $\dlp(\lambda')$, and furthermore we have
\[\lambda' \sum_{\sigma \in S_V} \tilde{y}_\sigma = (\lambda - L) \cdot \frac{\lambda}{\lambda-L} \sum_{\sigma \in S_V} \bar{y}_\sigma = 1,\]
since $\lambda \sum_{\sigma \in S_v} \bar{y}_\sigma = 1$ by construction. But this is a contradiction to the infeasibility of $\dlp(\lambda')$, and hence we conclude that $r \geq \frac{1}{2}$ as claimed.
\end{proof}

\subsection{Full Algorithm and Analysis}
We now can give our full algorithm, Weighted-DSG-LEDP (Algorithm~\ref{alg:weighted-DSG-LEDP}), which essentially wraps Weighted-Noisy-Order-Packing-MWU in a grid search for $\lambda^*$, and repeats each guess a number of times in order to amplify the probability of getting a good solution from $\frac{1}{2}$ to $1-n^{-c}$. The grid points in the search for $\lambda^*$ are arranged in a geometric progression with factor $1+\beta$ covering the entire possible range for $\lambda^*$. The $c$ parameter in Weighted-DSG-LEDP is used to specify success probability of at least $1-3n^{-c}$ in the utility bound (see Theorem~\ref{thm:weighted-dsg-ledp-utility}).

\begin{algorithm}
    \caption{Weighted-DSG-LEDP($G, T, \varsigma, c, \beta$)}
    \label{alg:weighted-DSG-LEDP}
    \begin{algorithmic}[1]
        \STATE Let $\lambda_0 = 1/(2\cmax)$, $K = c\log_2(n) \cdot (\log_{1+\beta}(2\cmax n)+1)$, 
        % $T = (\epsilon n \cmax)^2$, 
        and $\tau = \sqrt{T}\varsigma$
        % \frac{\tilde{C} \sqrt{T\log(\cmax n)\log^2(T/\delta) }}{\sqrt{\beta}\epsilon}$. Here, $\tilde{C}$ is the constant depending on $c$ from the proof of Theorem~\ref{thm:weighted-privacy}.
        \FOR{$i=1$ \TO $\log_{1+\beta}(2\cmax n)+1$}
            \STATE Let $\lambda_i = (1+\beta)^{i-1} \lambda_0$
            \FOR{$j=1$ \TO $c \log_2 n$}
                \STATE Let $(\pi^{(i,j)}, x^{(i,j)}) \leftarrow$ Weighted-Noisy-Order-Packing-MWU$(G, \lambda_i, T, \tau)$.
                \STATE Let $(S^{(i,j)}, \widehat \density(S^{(i,j)})) \leftarrow $ Weighted-Peeling$\left(G, \pi^{(i,j)}, 
                % \frac{\epsilon}{4\sqrt{2K \log(4/\delta)}}, \frac{\delta}{4K} 
                \varsigma
                \right)$
            \ENDFOR     
        \ENDFOR
        \STATE The curator computes $(i^*, j^*) = \argmax_{i,j} \widehat \density(S^{(i,j)})$ and returns $S^{(i^*, j^*)}$.
    \end{algorithmic}
\end{algorithm}

\subsubsection{Privacy} 

% We begin by arguing that this algorithm preserved privacy.
%We first discuss the privacy of the algorithm.
The following lemma has the exact same proof as Lemma~\ref{lem:nopmwu-privacy}:
\begin{lemma} \label{lem:weighted-nopmwu-privacy}
    Weighted-Noisy-Order-Packing-MWU is $\frac{T}{\tau^2}$-zCDP.
\end{lemma}

\begin{lemma} \label{lem:weighted-privacy-overall}
    Weighted-DSG-LEDP($G, T, \varsigma, c, \beta)$ is $\frac{K}{\varsigma^2}$-zCDP.
\end{lemma}
\iflong
\begin{proof}
    Note that Weighted-DSG-LEDP runs $K$ copies of Weighted-Noisy-Order-Packing followed by a run of Weighted-Peeling. By Lemma~\ref{lem:DSG-LEDP-privacy} each run of Weighted-Noisy-Order-Packing is $\frac{T}{2\tau^2}$-zCDP. Since $\tau = \sqrt{T} \varsigma$, this can be rewritten as $\frac{1}{2\varsigma^2}$-zCDP. Lemma~\ref{lem:peeling-DP} shows that each call to Weighted-Peeling is $\frac{1}{2\varsigma^2}$-zCDP. Thus, by the composition results for zCDP~\cite{bun2016concentrated}, we get that Weighted-DSG-LEDP is $\frac{K}{\varsigma^2}$-zCDP, as required.
\end{proof}

\begin{corollary} \label{cor:weighted-dsg-ledp-privacy}
    Let $\delta \in (0, 1)$ and $\epsilon \in (0, 8\log(1/\delta))$ be given privacy parameters. Set $\varsigma = \frac{4\sqrt{K \log(1/\delta)}}{\epsilon}$ in Weighted-DSG-LEDP. Then Weighted-DSG-LEDP is $(\epsilon, \delta)$-LEDP.
\end{corollary}
\begin{proof}
The zCDP guarantee for Weighted-DSG-LEDP translates to $(\alpha, \frac{\alpha K}{\varsigma^2})$-RDP guarantee for any $\alpha \geq 1$, which in turn translates to $(\frac{\alpha K}{\varsigma^2} + \frac{\log(1/\delta)}{\alpha-1}, \delta)$-LEDP for any $\alpha \geq 1$ (see Proposition 3 in \cite{Mironov17}). Choosing $\alpha = \varsigma\sqrt{\frac{\log(1/\delta)}{K}} + 1$ for $\varsigma$ as specified above in the corollary statement, we get that Weighted-DSG-LEDP is $(\epsilon, \delta)$-DP.

Furthermore it is easy to note that Weighted-DSG-LEDP only does local operations. Thus Weighted-DSG-LEDP is $(\epsilon, \delta)$-LEDP.
% Turning to the $c\log_2 n$ runs of Peeling: each run is $\left(\frac{\epsilon}{4\sqrt{2c\log(n) \log(4/\delta)}}, \frac{\delta}{4c\log n}\right)$-LEDP due to Lemma~\ref{lem:peeling-DP}. Thus, by the advanced composition theorem (Theorem~\ref{thm:adv-composition}), all the runs put together are $(\frac{\epsilon}{2}, \frac{\delta}{2})$-LEDP.
% Combining the above privacy guarantees, the proof is complete.
\end{proof}
\fi

\subsubsection{Utility}

In order to prove the utility of Weighted-DSG-LEDP, we first need to prove the weighted equivalent of Lemma~\ref{lem:primal-to-discrete}: we need to show that a good solution to~\eqref{eq:plp-lambda} not only implies a good subgraph, but that it has a particular structure.

\begin{lemma} \label{lem:weighted-primal-to-discrete}
Given a feasible solution $x$ to~\eqref{eq:plp-lambda}, there exists a $\tau \in [0,1]$ such that the set $S_{\tau} =\{v \in V : x_v \geq \tau\}$ has density at least $\lambda$.
\end{lemma}
\begin{proof}
    For each $\tau \in [0,1]$, let $E_{\tau} = \{\{u,v\} \in E : \min(x_u, x_v) \geq \tau\}$ be the edges induced by $S_{\tau}$.  Suppose for contradiction that the lemma is false: $\density(S_{\tau})  = |E_{\tau}| / c(S_{\tau}) < \lambda$ for every $\tau \in [0,1]$.  Consider picking $\tau$ uniformly at random from $[0,1]$.  Since $x$ is feasible for~\eqref{eq:plp-lambda}, we know that $\langle x, q(\sigma) \rangle \geq \lambda$ for all $\sigma \in S_V$.  By Lemma~\ref{lem:opt-ordering} (which still applies since it is independent of the weights), this means that $\sum_{\{u,v\} \in E} \min(x_u, x_v) \geq \lambda$.  Since $\Pr[\{u,v\} \in E_{\tau}] = \min(x_u, x_v)$, we know from lienarity of expectations that $\E[|E_{\tau}|] \geq \lambda$.
    
    Now we have that
    \begin{align*}
        \lambda \leq \E[|E_{\tau}|] &= \E\left[c(S_{\tau}) \density(S_{\tau}) \right] < \E\left[ c(S_{\tau}) \lambda\right] \tag{assumption for contradiction} \\
        &= \lambda \E[c(S_{\tau})] = \lambda \sum_{v \in V} c_v \Pr[v \in S_{\tau}] \tag{linearity of expectations} \\
        &= \lambda \sum_{v \in V} c_v x_v = \lambda \tag{$\sum_{v \in V} c_v x_v = 1$ since $x$ is feasible for~\eqref{eq:plp-lambda}}
    \end{align*}

    This is a contradiction, since obviously it is not true that $\lambda < \lambda$.  Thus the lemma is true.
\end{proof}

We can now analyze the utility of Weighted-DSG-LEDP (Algorithm~\ref{alg:weighted-DSG-LEDP}).

\begin{theorem} \label{thm:weighted-dsg-ledp-utility}
    Suppose $T$ and $\varsigma$ in Weighted-DSG-LEDP are set so that $\tau = \sqrt{T}\varsigma \geq n$. Then with probability at least $1 - (2K+1)n^{-c}$, Weighted-DSG-LEDP returns a set $S$ with $\density(S) \geq (1-4\beta)\lambda^* - O\left(\tfrac{1}{\beta}\sqrt{\log n}\varsigma\right)$.
    % With probability at least $1 - 3n^{-c}$, Weighted-DSG-LEDP returns a set $S$ with 
    % \[\density(S) \geq (1-4\beta)\lambda^* - O\left(\frac{\log(n)\log^{1.5}(n \cmax/\delta)}{\beta^{1.5} \epsilon} \right).\]
\end{theorem}
\begin{proof}
    We may assume without loss of generality that $\lambda^* \geq \tfrac{8\sqrt{\log n}\varsigma}{\beta}$; otherwise, the guarantee in the theorem holds trivially. 
    % \[\lambda^* \geq \tfrac{8\sqrt{\log n}\varsigma}{\beta} \geq 
    % \frac{4(n+\tau)}{\beta} \sqrt{\frac{\log n}{T}}.\] 
    The geometric grid search for $\lambda^*$ ensures that there is at least one grid point, $\lambda_i$, which satisfies $\lambda^* \leq \lambda_i < (1+\beta)\lambda^*$. Note that $\beta \lambda^* \geq 8\sqrt{\log n} \varsigma \geq 4(n+\tau) \sqrt{\frac{\log n}{T}}$, since $\sqrt{T}\varsigma \geq n$. Thus, we can apply Theorem~\ref{thm:weighted-primal-solve-robust} with $L = \beta \lambda^*$ and $\lambda = \lambda_i$ to conclude that with probability at least $1/2$, Weighted-Noisy-Order-Packing-MWU$(G, \lambda_i, T, \tau)$ outputs a feasible solution $x$ to $\plp(\lambda_i - 4\beta \lambda^*)$, and its corresponding permutation $\sigma$. Recall that $S^*_{\sigma}$ is the prefix of $\sigma$ with largest density.  Since $\sigma$ is just nondecreasing order of $x$, Lemma~\ref{lem:weighted-primal-to-discrete} implies that 
    \[\density(S^*_{\sigma}) \geq \lambda_i - 4\beta \lambda^* \geq (1-4\beta)\lambda^*,\]
    since $\lambda_i \geq \lambda^*$. Since Weighted-DSG-LEDP runs $c \log_2 n$ independent copies of Weighted-Noisy-Order-Packing-MWU$(G, \lambda_i, T, \tau)$, we conclude that with probability at least $1 - n^{-c}$, there is at least one index $j \in [c \log_2 n]$ in which $\density(S^*_{\pi^{(i, j)}}) \geq (1-4\beta)\lambda^*$. 

    We do not know which specific pair of indices $(i, j)$ this is, but Weighted-DSG-LEDP then runs Weighted-Peeling (Algorithm~\ref{alg:peeling-weighted}) on each of the permutations obtained from Weighted-Noisy-Order-Packing-MWU. Using Lemma~\ref{lem:weighted-peeling-approx}, we conclude that via a union bound that with probability at least $1-2Kn^{-c}$, for all $i \in [c \log_2 n]$ and $j \in [\log_{1+\beta}(2\cmax n)+1]$, the set $S^{(i, j)}$ that we get from calling Weighted-Peeling has both true and estimated density within $O\left(\sqrt{\log n}\varsigma\right)$
    % \[O\left(\frac{\sqrt{\log(n) K \log(1/\delta) \cdot \log(K/\delta)}}{\epsilon}\right) = O\left(\frac{\log(n) \cdot \sqrt{\log(n\cmax)} \cdot \log(\log(n\cmax) / (\beta\delta))}{\sqrt{\beta} \epsilon} \right)\] 
    of $S^*_{\pi^{(i, j)}}$.
    Thus, with probability at least $1-(2K+1)n^{-c}$, we have that
    \begin{align*}
        \density(S) &\geq \max_{i \in [c \log n], j \in [\log_{1+\beta}(2\cmax n)+1]} \density(S^*_{\pi^{(i, j)}}) - O\left(\sqrt{\log n}\varsigma\right)  \\
        &\geq (1-4\beta)\lambda^* - O\left(\tfrac{1}{\beta}\sqrt{\log n}\varsigma\right), 
    \end{align*}
    as required.

    %     \\
    %     &\geq (1-4\beta)\lambda^* - O\left(\frac{\log(n)\sqrt{\log (\cmax n)\log(1/\delta)}}{\sqrt{\beta}\epsilon}\right)
    % \end{align*}
    % since $\varsigma = \frac{4\sqrt{K \log(1/\delta)}}{\epsilon} = O\left(\frac{\sqrt{\log(n)\log(\cmax n)\log(1/\delta)}}{\sqrt{\beta}\epsilon}\right)$.

    % Finally, if $\lambda^* < \tfrac{8\sqrt{\log n}\varsigma}{\beta}$, then $(1-4\beta)\lambda^* - \tfrac{8}{\beta}\sqrt{\log n}\varsigma \leq 0$. So in this case, any output set satisfies the density guarantee.
\end{proof}
\begin{corollary}
    Let $\delta \in (0, 1)$ and $\epsilon \in (0, 8\log(1/\delta)$ be given privacy parameters. Set $\varsigma = \frac{4\sqrt{K\log(1/\delta)}}{\epsilon}$ and $T = \lceil\frac{n^2}{\varsigma^2}\rceil$ in Weighted-DSG-LEDP. Then Weighted-DSG-LEDP is $(\epsilon,\delta)$-LEDP, and with probability at least $1 - (2K+1)n^{-c}$, DSG-LEDP returns a set $S$ with $\density(S) \geq (1-4\beta)\lambda^* - O\left(\frac{\log(n)\sqrt{\log (\cmax n)\log(1/\delta)}}{\beta^{1.5}\epsilon}\right)$.
\end{corollary}
\begin{proof}
The privacy guarantee follows from Corollary~\ref{cor:weighted-dsg-ledp-privacy}, and the output guarantee from Theorem~\ref{thm:weighted-dsg-ledp-utility}, using the bound $\varsigma = \frac{4\sqrt{K \log(1/\delta)}}{\epsilon} = O\left(\frac{\sqrt{\log(n)\log(\cmax n)\log(1/\delta)}}{\sqrt{\beta}\epsilon}\right)$.
\end{proof}

\subsection{Improvement in the Centralized Setting}

Similar to the unweighted case, the Centralized-Weighted-DSG algorithm (Algorithm~\ref{alg:centralized-weighted-dsg}) runs Weighted-Noisy-Order-Packing-MWU (centralized) followed by Centralized-Peeling. The main twist here is that Weighted-Noisy-Order-Packing-MWU requires a guess $\lambda$ for $\lambda^*$. This is chosen randomly from an arithmetic grid with spacing $12\tau \sqrt{\frac{\log n}{T}}$ covering the entire range of $\lambda^*$ so as to ensure Theorem~\ref{thm:weighted-primal-solve-robust} can be applied.

\begin{algorithm}
    \caption{Centralized-Weighted-DSG-core($G, T,  \varsigma$)}
    \label{alg:centralized-weighted-dsg}
    \begin{algorithmic}[1]
        \STATE Set 
        % $T = \lceil \frac{n^2}{\varsigma^2}\rceil$, 
        $\tau = \sqrt{T}\varsigma$
        % $\frac{16\sqrt{T \log^2(5T/\delta')}}{\epsilon'}$, 
        and $N = \left \lceil \frac{2\cmax n}{4(n+\tau) \sqrt{\log(n)/T}}\right\rceil$.

        \STATE Sample $k \sim [N]$ uniformly at random and set $\lambda = k \cdot 4(n+\tau) \sqrt{\frac{\log n}{T}}$.

        \STATE Run Weighted-Noisy-Order-Packing-MWU($G, \lambda, T, \tau$) and obtain a permutation $\sigma$.

        \STATE Output $(S, \tilde \density(S)) \leftarrow$ Weighted-Peeling($G, \sigma, \varsigma$).
    \end{algorithmic}
\end{algorithm}

The following lemma is proved exactly on the lines of  Lemma~\ref{lem:centralized-privacy}, together with the observation that the sampling of the index $k$ doesn't depend on any private information, i.e., the edges of $G$:
\begin{lemma} \label{lem:centralized-weighted-privacy}
    Centralized-Weighted-DSG-core($G, T, \varsigma)$ is $\frac{1}{\varsigma^2}$-zCDP.
\end{lemma}

\begin{lemma} \label{lem:weighted-centralized-dsg-ledp-utility}
    Suppose $T$ and $\varsigma$ in Centralized-Weighted-DSG-core are set so that $\tau = \sqrt{T}\varsigma \geq n$. Then with probability at least $\frac{1}{4N}$, Centralized-Weighted-DSG-core($G, T, \varsigma$) returns a set $S$ with $\tilde\density(S) \geq \density(G) - O(\sqrt{\log n}\varsigma)$.
\end{lemma}
\begin{proof}
Let $\lambda^* = \density(G)$. Note that with probability $\frac{1}{N}$, the guess $\lambda$ chosen in Centralized-Weighted-DSG-core satisfies $\lambda^* \leq \lambda < \lambda^* + 4(n+\tau) \sqrt{\frac{\log n}{T}}$. Conditioned on this happening,  Theorem~\ref{thm:weighted-primal-solve-robust} and Lemma~\ref{lem:weighted-primal-to-discrete} imply that with probability at least $1/2$, Weighted-Noisy-Order-Packing-MWU outputs a permutation $\sigma$ such that 
\[\density(S^*_\sigma) \geq \lambda - 16(n+\tau) \sqrt{\frac{\log n}{T}} \geq \lambda^* - 16(n+\tau) \sqrt{\frac{\log n}{T}}\] 
since $\lambda \geq \lambda^*$. Thus, using Lemma~\ref{lem:weighted-peeling-approx}, we conclude that with probability at least $\frac{1}{4N}$, the following holds for the set $S$ and estimated density $\tilde\density(S)$ output by Centralized-Weighted-DSG-core:
\begin{equation} \label{eq:centralized-weighted-utility}
\density(S) \geq \density(S^*_\sigma) - O(\sqrt{\log n} \varsigma) \geq \lambda^*  - O\left((n+\tau) \sqrt{\frac{\log n}{T}} + \sqrt{\log n} \varsigma\right)
\end{equation}
and
\begin{equation} \label{eq:centralized-weighted-estimated-density}
    |\density(S) - \tilde{\density}(S)| \leq O(\sqrt{\log n} \varsigma).
\end{equation}
For the specified value of $T$ and $\tau$, we have
\[
(n+\tau) \sqrt{\frac{\log n}{T}} + \sqrt{\log n} \varsigma = O(\sqrt{\log n} \varsigma).
\]
% \[48\tau \sqrt{\frac{\log n}{T}} + \frac{8 \ln(2n)}{\epsilon'} = O\left(\frac{\sqrt{\log(n)} \log(n/\delta')}{\epsilon'}\right) = O\left(\frac{\sqrt{\log(n)} \log(n/\delta)}{\epsilon}\right),\] 
% where the latter bound holds using the specified values of $\epsilon'$ and $\delta'$ in terms of $\epsilon$ and $\delta$, and assuming $\gamma = n^{-c}$ for some constant $c$. 
Hence Eq.~\eqref{eq:centralized-weighted-utility} and \eqref{eq:centralized-weighted-estimated-density} imply that
\[\tilde\density(S) \geq \lambda^* - O\left(\sqrt{\log(n)}\varsigma\right),\]
as required.
\end{proof}

Now, let Centralized-Weighted-DSG be the mechanism $\AQ$ obtained by applying the mechanism of Theorem~\ref{thm:papernot-steinke} with $Q = $ Centralized-Weighted-DSG-core, $s = S$, the set it outputs, and $q = \tilde{\density}(S)$, the estimated density it outputs, $q^* = \density(G) - O\left(\sqrt{\log(n)}\varsigma\right)$ as specified in Lemma~\ref{lem:weighted-centralized-dsg-ledp-utility}, and $\gamma = n^{-c}$, for any given constant $c$.

\begin{theorem} \label{thm:weighted-centralized-main}
    Let $\delta \in (0, 1)$ and $\epsilon > 0$ be given privacy parameters. Set $\varsigma = \frac{6\sqrt{\log(n^c/\delta)}}{\epsilon}$ and $T = \lceil \frac{n^2}{\varsigma^2}\rceil$ in Centralized-Weighted-DSG-core. Then Centralized-Weighted-DSG is $(\epsilon, \delta)$-DP, and with probability at least $1 - (4N+2)n^{-c}$, the density of the set it outputs is at least $\density(G) -O\left(\frac{\sqrt{\log(n)\log(n/\delta)}}{\epsilon}\right)$.
\end{theorem}
\begin{proof}
    Using Lemma~\ref{lem:weighted-centralized-dsg-ledp-utility}, we have
    \[\Pr_{(s, q) \sim Q(D)}[q \geq q^*] \geq \frac{1}{4N}.\]
    Since Centralized-Weighted-DSG-core($G, \varsigma$) is $\frac{1}{\varsigma^2}$-zCDP, applying Theorem~\ref{thm:papernot-steinke}, we conclude that Centralized-Weighted-DSG is 
    \[
    \left(\tfrac{6\sqrt{\log(1/(\gamma\delta))}}{\varsigma}, \delta\right)\text{-DP} = (\epsilon, \delta)\text{-DP}
    \]
    for the specified value of $\varsigma$, and the estimated density $\tilde\density(S)$ of the set $S$ it outputs is at least $q^*$ with probability at least $1 - 4N\gamma$. Now, recall the random variable $J$ used in Centralized-Weighted-DSG which is drawn from the geometric distribution with success probability $\gamma$. Since $\Pr[J > k] = (1-\gamma)^k \leq \exp(-\gamma k)$, we conclude that $\Pr[J \leq \frac{\log(1/\gamma)}{\gamma}] \geq 1 - \gamma$. Conditioned on $J \leq \frac{\log(1/\gamma)}{\gamma}$, using Lemma~\ref{lem:peeling-approx} and a union bound over the $J$ calls to $Q$ in $\AQ$, we conclude that with probability at least $1-\gamma$, in each call to $Q$, we have
    \[|\density(S) - \tilde\density(S)| \leq O\left(\sqrt{\log n}\varsigma\right)\]
    for an appropriately chosen constant in the $O(\cdot)$ notation. Thus, overall, using the union bound, with probability at least $1-(4N+2)\gamma$, the true density of the set output by Centralized-Weighted-DSG is at least 
    \[\density(G) - O\left(\sqrt{\log n}\varsigma\right) = \density(G) - O\left(\frac{\sqrt{\log(n)\log(n/\delta)}}{\epsilon}\right). \qedhere\] 
\end{proof}

% \begin{theorem}
%     Centralized-Weighted-DSG is $(\epsilon, \delta)$-DP, and the density of the set it outputs is at least $\lambda^* - O\left(\frac{\sqrt{\log(n)} \log(n/\delta)}{\epsilon}\right)$ with probability at least $1 - 17\cmax n^{1-c} - \frac{\delta}{2}$.
% \end{theorem}
% \begin{proof}
%     The above analysis implies that 
%     \[\Pr_{(s, q) \sim Q(D)}[q \geq q^*] \geq \frac{1}{8K}.\]
%     Then, applying Theorem~\ref{thm:liu-talwar}, we conclude that Centralized-DSG is $(\epsilon, \delta)$-DP, and the density of the set it outputs is at least $q^*$ with probability at least $1 - 8K \gamma - \frac{\delta}{2}$. Using Theorem~\ref{thm:centralized-peeling} and a union bound over the at most $\Jmax$ calls to $Q$ in $\AQ$, we conclude that with probability at least $1-\gamma$, $\density(S) \geq \tilde\density(S) - \frac{4\ln(\Jmax/\gamma)}{\epsilon'}$. Since $\frac{4\ln(\Jmax/\gamma)}{\epsilon'} = O\left(\frac{\sqrt{\log(n)} \log(n/\delta)}{\epsilon}\right)$ and $K \leq 2\cmax n$, the theorem follows.
% \end{proof}

% Note that as long as $\cmax$ is at most polynomial in $n$ and $\delta$ is at most an inverse polynomial in $n$ (both of which are relatively standard assumptions), then the failure probability in the above theorem is at most an inverse polynomial in $n$ by taking $c$ to be large enough. 

\section{Directed Densest Subgraph} \label{sec:directed}
In the directed version we are given a \emph{directed} graph $G = (V, E)$.  Given two subsets $S, T \subseteq V$, we let $E(S,T)$ denote the edges from $S$ to $T$.  The density of $S, T$ is defined as,
\begin{align*}
    \density_G(S,T) = \frac{|E(S,T)|}{\sqrt{|S||T|}},
\end{align*}
and the maximum subgraph density is $\density(G) = \max_{S, T \subseteq V} \density_G(S,T)$.  This problem was introduced by~\cite{KV99}, and has been received significant study since.  Charikar~\cite{Charikar00} showed how to solve this problem exactly using $O(n^2)$ separate linear programs, and also showed that it was possible to give a $1-\epsilon$ approximation while solving only $O(\log n / \epsilon)$ linear programs.  More recently, Sawlani and Wang~\cite{SW20} gave a \emph{black-box} reduction to the node-weighted undirected setting.  Since we designed an $(\epsilon, \delta)$-LEDP algorithm for the node-weighted undirected setting in Section~\ref{sec:weighted}, this reduction will form our starting point. As one piece of notation, for any node $v$ and subset $T \subseteq V$, we let $\deg_T(v) = |\{(v, u) : u \in T\}|$ be the number of edges from $v$ to nodes in $T$. 

\subsection{The reduction of~\texorpdfstring{Sawlani and Wang \cite{SW20}}{Sawlani and Wang}}
We now give a quick overview of the reduction from the directed case to the node-weighted undirected case given by~\cite{SW20}.  Suppose we are given a directed graph $G = (V, E)$.  Given a parameter $t$, we build a new node-weighted, undirected bipartite graph $G_t = (V_t, E_t, c_t)$.  We do this by setting $V_t = V_t^L \cup V_t^R$, where $V_t^L$ and $V_t^R$ are both copies of $V$.  For every edge $(u,v)$ of $E$, we create an edge $\{u,v\} \in E_t$ with $u \in V_t^L$ and $v \in V_t^R$ (i.e., we create an undirected edge where the left endpoint is the tail of the original directed edge and the right endpoint is the head of the original directed edge).  Our weight function $c_t$ sets $c_t(u) = 1/(2t)$ for $u \in V_t^L$, and $c_t(u) = t/2$ for $u \in V_t^R$.  

The first lemma says that no matter what the value of $t$, the density of sets in the directed graph is no smaller than the density of sets in the weighted undirected graph (and hence the optimal density of the weighted graph is at most the optimal density of the directed graph).  This is primarily a consequence of the AM-GM inequality.
\begin{lemma}[Lemma 5.1 of~\cite{SW20}, slightly rephrased] \label{lem:reduction1}
    Let $S, T \subseteq V$, let $t$ be arbitrary, let $S^L$ denote the copies of $S$ in $V_t^L$, and let $T^R$ denote the copies of $T$ in $V_t^R$.  Then $\density_G(S,T) \geq \density_{G_t}(S^L \cup T^R)$.
\end{lemma}

It turns out that for the ``correct'' guess of $t$, we have equality.
\begin{lemma}[Lemma 5.2 of~\cite{SW20}, slightly rephrased] \label{lem:reduction3}
    Let $(S, T)$ be the optimal densest subgraph in $G$, i.e., $\density(G) = \density_G(S, T)$.  Let $t = \sqrt{|S|/|T|}$.  Let $S^L$ be the copies of $S$ in $V_t^L$, and let $T^R$ be the copies of $T$ in $V_t^R$. Then $\density_G(S,T) = \density_{G_t}(S^L \cup T^R)$.  
\end{lemma}

Since there are at most $n$ possible choices for $|S|$ and $n$ possible choices for $|T|$, there are at most $n^2$ different possible values of $t$ for us to try.  So the previous two lemmas imply that we can just try all $n^2$ of these possibilities and choose the best.  Not surprisingly, it was also shown in~\cite{SW20} that it is possible to lose a $(1-\epsilon)$-approximation factor and reduce this to only trying $O(\frac{1}{\epsilon} \log n)$ different values of $t$.  We will use this fact, but not as a black box.  

The main difficulty in combining this reduction with our node-weighted algorithm from Section~\ref{sec:weighted} to get an LEDP algorithm for directed DSG is that for any value of $t$ other than $1$, either $V_t^L$ or $V_t^R$ will have node weights that are less than $1$, and for the extreme value of $t$, the weights could be as small as $1/\sqrt{n}$.  But our algorithm from Section~\ref{sec:weighted} requires all node weights to be at least $1$.

\subsection{An LEDP algorithm for Directed DSG}

Our algorithm is relatively simple: the curator can try different values of $t$, and then we can run Weighted-DSG-LEDP on each of the node-weighted undirected graphs resulting from the reduction of~\cite{SW20}.  For each value of $t$ the curator tries, it can announce that value to the nodes, and every node $v$ will know the weight of its two copies.  Node $v$ can then act as if it is each of its copies when responding.  The only subtlety is that we need to rescale so all node weights are at least $1$, but this can clearly be done locally by the nodes since they will know $t$.  This algorithm is presented as Directed-DSG-LEDP in Algorithm~\ref{alg:directed-DSG-LEDP}. 

\begin{algorithm}
    \caption{Directed-DSG-LEDP($G, T', \varsigma,c, \beta$)}
    \label{alg:directed-DSG-LEDP}
    \begin{algorithmic}[1]
        \STATE Let $t_0 = 1/\sqrt{n}$
        \FOR{$i = 0$ \TO $\log_{1+\beta} n$}
            \STATE Let $t_i = (1+\beta)^i t_0$
            \STATE The curator announces $t_i$
            \STATE For every node $v$, the curator and $v$ both compute $\alpha_i = \min(1/(2t_i), t_i / 2)$, and compute the weight $c_{v,i}^{L} =  1/(2t_i \alpha_i)$ of its copy in $V_{t_i}^L$ and the weight $c_{v,i}^R = t_i / (2\alpha_i)$ of its copy in $V_{t_i}^R$.  Let $G'_{t_i}$ denote $G_{t_i}$ but with these weights (which are scaled up by exactly $1/\alpha_i$ from the weights in $G_{t_i}$)
            \STATE Let $U_i \leftarrow$ Weighted-DSG-LEDP($G'_{t_i}, T', \varsigma, c, \beta$)  
            \STATE Let $S_i = U_i \cap V_{t_i}^L$ and $T_i = U_i \cap V_{t_i}^R$.
            \STATE The curator announces $S_i$ and $T_i$
            %\item Let $\tau =\varsigma$ %\sknoteinline{Maybe this should be just $\varsigma$? Why is $\tau$ set this way? In other Peeling calls, we simply use $\varsigma$ for the noise.}%\frac{\sqrt{2 \log ((5 \log_{1+\beta} n) / \delta)}}{\epsilon}$
            \STATE Every node $v \in S_i$ draws a random value $N_v^{(i)} \sim N(0, \varsigma^2)$ and sends $\widehat \deg_{T_i}(v) = \deg_{T_i}(v) + N_v^{(i)}$ to the curator. 
            \STATE The curator computes $\widehat \density(S_i, T_i) = \frac{\sum_{v \in S_i} \widehat \deg_{T_i}(v)}{\sqrt{|S_i| |T_i|}}$
        \ENDFOR
        \STATE The curator computes $i^* = \argmax_i \widehat \density(S_i, T_i)$ and returns $(S_{i^*}, T_{i^*})$.
    \end{algorithmic}
\end{algorithm}

We begin by showing that this algorithm is private.  To simplify notation, we will set $M = 2c \log_2(n) \cdot (\log_{1+\beta}(2n^{3/2})+1) \cdot \log_{1+\beta}(n)$.

\begin{theorem} \label{thm:directed-LEDP}
    Directed-DSG-LEDP($G, T', \varsigma, c, \beta)$ satisfies $\frac{M}{\varsigma^2}$-zCDP.
\end{theorem}
\begin{proof}
    Fix an iteration $i$.  We know from Lemma~\ref{lem:weighted-privacy-overall} that $U_i$ is $\frac{K}{\varsigma^2}$-zCDP, and we know from the Gaussian mechanism (Lemma~\ref{lem:gaussian}) and parallel composition that the collection of all $\widehat \deg_{T_i}(v)$ values is $\frac{1}{\varsigma^2}$-zCDP.  Everything else in iteration $i$ is postprocessing, and hence each iteration is $\frac{2K}{\varsigma^2}$-zCDP.  Note that $K$ is a function of $\cmax$, but in our setting we know that $\cmax \leq \sqrt{n}$.  Hence each iteration is $\frac{2c \log_2(n) \cdot (\log_{1+\beta}(2n^{3/2})+1)}{\varsigma^2}$-zCDP.  Since there are $\log_{1+\beta}(n)$ iterations in total, sequential composition (Theorem~\ref{thm:adseqcomp}) implies the theorem.
\end{proof}

\begin{corollary} \label{cor:directed-private}
Let $\delta \in (0, 1)$ and $\epsilon \in (0, 8\log(1/\delta))$ be given privacy parameters. Set $\varsigma = \frac{4\sqrt{M \log(1/\delta)}}{\epsilon}$ in Directed-DSG-LEDP. Then Directed-DSG-LEDP is $(\epsilon, \delta)$-LEDP.
\end{corollary}
\begin{proof}
    We first show that the algorithm can be implemented in the local model.  This is straightforward, as all computations are local, except the fact that it calls Weighted-DSG-LEDP in step $6$ on $G'_{t_i}$, rather than on $G$.  But note that every node knows the weight and incident edges of both of its copies in $G'_{t_i}$ and so can simulate both of its copies, and the curator also knows the weights and the nodes in $G'_{t_i}$ so can simulate the curator of Weighted-DSG-LEDP.  So Directed-DSG-LEDP can be implemented in the local model, and hence we only need to show differential privacy.

    Now we just need to transfer the zCDP guarantee of Theorem~\ref{thm:directed-LEDP} to $(\epsilon, \delta)$-DP. By definition, the $\frac{M}{\varsigma^2}$-zCDP guarantee of Theorem~\ref{thm:directed-LEDP} translates to a $\left(\alpha,  \frac{\alpha M}{\varsigma^2}\right)$-RDP guarantee for any $\alpha \geq 1$.  This in turn translates to $\left( \frac{\alpha M}{\varsigma^2} + \frac{\log(1/\delta)}{\alpha-1}, \delta\right)$-DP guarantee for any $\alpha \geq 1$ (see Propostion 3 in~\cite{Mironov17}).  Choosing $\alpha = \varsigma \sqrt{\frac{\log(1/\delta)}{M}} + 1$ for $\varsigma$ as specified above in the corollary statement, we get that Directed-DSG-LEDP is $(\epsilon, \delta)$-LEDP.
\end{proof}

We can now analyze the utility of Directed-DSG-LEDP.

\begin{theorem} \label{thm:directed-utility}
    Suppose $T'$ and $\varsigma$ are set so $\sqrt{T'} \varsigma \geq n$.  Then with probability at least $1 - O(n^{-c})$, Directed-DSG-LEDP returns $(S, T)$ such that
    \[
        \density_G(S,T) \geq (1-O(\beta)) \density(G) - O\left(\frac{1}{\beta} \varsigma \sqrt{\max\left(\frac{|S^*|}{|T^*|}, \frac{|T^*|}{|S^*|}\right)}  \sqrt{\log n}\right),
    \]
    where $(S^*, T^*)$ is the optimal solution.
\end{theorem}
\begin{proof}
    We will argue that for the ``correct'' $i$ the algorithm will return a good solution, and that for any of the other $\log_{1+\beta} n$ values of $i$ where the solution is significantly worse, our estimate of their density is sufficiently accurate so that we will not be fooled.

    Let $(S^*, T^*)$ be the optimal solution, and let $t = \sqrt{|S^*| / |T^*|}$.  Clearly $1/\sqrt{n} \leq t \leq \sqrt{n}$, so by the definition of Directed-DSG-LEDP there is some $i$ such that $t \leq t_{i} \leq (1+\beta)t$.  Fix this $i$.  
    
    Note that all weights in $G_t$ and in $G_{t_i}$ are within $1+\beta$ of each other, so every set has density that is only different by at most a $1+\beta$ factor between the two.  Let $S^*_i$ denote the copies of $S^*$ in $V_{t_i}^L$, and let $T^*_i$ denote the copies of $T^*$ in $V_{t_i}^R$.  As in the algorithm, let $U_i = S_i \cup T_i$ be the set returned by Weighted-DSG-LEDP when called on $G'_{t_i}$, and let $U'_i$ be the optimal solution for $G'_{t_i}$.
    We know from the Theorem~\ref{thm:weighted-dsg-ledp-utility} that with probability at least $1- O(Kn^{-c})$,
    \begin{equation} \label{eq:loss-from-weighted}
        \density_{G'_{t_i}}(U_i) \geq (1-4\beta) \density_{G'_{t_i}}(U'_i) - O\left(\frac{1}{\beta} \varsigma \sqrt{\log n}\right).
    \end{equation}
    %where we used the fact that $\cmax \leq n$ in $G'_{t_i}$ and we used the smaller values of $\epsilon$ and $\delta$ that we call Weighted-DSG-LEDP with.  
    So we have:

    \begin{align*}
        \density_G(S_i, T_i) & \geq \density_{G_t}(U_i) \tag{Lemma~\ref{lem:reduction1}} \\
        & \geq \frac{1}{1+\beta} \density_{G_{t_i}}(U_i) \\
        &= \frac{1}{1+\beta} \cdot \frac{|E(U_i)|}{|S_i| \alpha_i c_{v,i}^L + |T_i| \alpha_i c_{v,i}^R} \\
        &= \frac{1}{1+\beta} \cdot \frac{1}{\alpha_i} \cdot \density_{G'_{t_i}}(U_i) \\
        &\geq \frac{1}{1+\beta} \cdot \frac{1}{\alpha_i} \cdot \left((1-4\beta) \density_{G'_{t_i}}(U'_i) - O\left(\frac{1}{\beta} \varsigma \sqrt{\log n}\right)\right) \tag{Eq.~\eqref{eq:loss-from-weighted}} \\
        &\geq \frac{1}{1+\beta} \cdot \frac{1}{\alpha_i} \cdot \left((1-4\beta) \density_{G'_{t_i}}(S^*_i \cup T^*_i) - O\left(\frac{1}{\beta} \varsigma \sqrt{\log n}\right)\right) \tag{$U'_i$ optimal for $G'_{t_i}$}\\
        &= \frac{1}{1+\beta} \cdot \frac{1}{\alpha_i} \cdot \left(\alpha_i (1-4\beta)  \density_{G_{t_i}}(S^*_i \cup T^*_i) - O\left(\frac{1}{\beta} \varsigma \sqrt{\log n}\right)\right) \\
        &\geq  \frac{1}{(1+\beta)^2} (1-4\beta)  \density_{G_{t}}(S^*_i \cup T^*_i) - O\left(\frac{1}{\alpha_i \beta} \varsigma \sqrt{\log n}\right) \\
        &\geq (1-O(\beta)) \density_{G_{t}}(S^*_i \cup T^*_i) - O\left(\frac{1}{\alpha_i \beta} \varsigma \sqrt{\log n}\right) \\
        &= (1-O(\beta)) \density_{G}(S^*, T^*) - O\left(\frac{1}{\alpha_i \beta} \varsigma \sqrt{\log n}\right) \tag{Lemma~\ref{lem:reduction3}} \\
        &\geq (1-O(\beta)) \density_{G}(S^*, T^*) - O\left( \frac{\varsigma \sqrt{\max\left(\frac{|S^*|}{|T^*|}, \frac{|T^*|}{|S^*|}\right)}  \sqrt{\log n}}{\beta}\right).
    \end{align*}

    Now consider an \emph{arbitrary} $i$ (not necessarily the ``right'' $i$ as above).  Using essentially the argument as in Lemma~\ref{lem:peeling-approx}, we know that with probability at least $1 - O(n^{-c})$, we have
    \[|\widehat \density(S_i, T_i) - \density(S_i, T_i)| \leq O\left( \varsigma \sqrt{\log n} \right).\]  
    Since this holds for all $i$, including the correct value of $i$, we get that with probability at least $1 - O(Kn^{-c})$,
    \begin{align*}
        \density_G(S,T) &\geq \density_G(S_i, T_i) - O\left(\varsigma \sqrt{\log n} \right) \\
        & \geq (1-O(\beta)) \density_{G}(S^*, T^*) - O\left( \frac{\varsigma \sqrt{\max\left(\frac{|S^*|}{|T^*|}, \frac{|T^*|}{|S^*|}\right)}  \sqrt{\log n}}{\beta}\right) - O\left(\varsigma\sqrt{\log n} \right) \\
        &\geq (1-O(\beta)) \density_{G}(S^*, T^*) - O\left( \frac{\varsigma \sqrt{\max\left(\frac{|S^*|}{|T^*|}, \frac{|T^*|}{|S^*|}\right)}  \sqrt{\log n}}{\beta}\right) \tag{$\because \beta \leq 1$} \\
        &= (1-O(\beta)) \density(G) - O\left( \frac{\varsigma \sqrt{\max\left(\frac{|S^*|}{|T^*|}, \frac{|T^*|}{|S^*|}\right)}  \sqrt{\log n}}{\beta}\right)
    \end{align*}
    as claimed.
\end{proof}

\begin{corollary}
    Let $\delta \in (0, 1)$ and $\epsilon \in (0, 8\log(1/\delta)$ be given privacy parameters. Set $\varsigma = \frac{4\sqrt{M\log(1/\delta)}}{\epsilon}$ and $T' = \lceil\frac{n^2}{\varsigma^2}\rceil$ in Directed-DSG-LEDP. Then Directed-DSG-LEDP is $(\epsilon,\delta)$-LEDP, and with probability at least $1 - O(Kn^{-c})$, DSG-LEDP returns a set $(S,T)$ with 
    \[\density(S,T) \geq (1-O(\beta)) \density(G) - O\left(\frac{1}{\beta^2\epsilon} \log^2 n \sqrt{\log(1/\delta)}\sqrt{\max\left(\frac{|S^*|}{|T^*|}, \frac{|T^*|}{|S^*|}\right)}  \right).\]
\end{corollary}
\begin{proof}
    The privacy guarantee follows from Corollary~\ref{cor:directed-private}.  For the output guarantee, we have that 
    \begin{align*}
        \density(S,T) &\geq (1-O(\beta)) \density(G) - O\left(\frac{1}{\beta} \varsigma \sqrt{\max\left(\frac{|S^*|}{|T^*|}, \frac{|T^*|}{|S^*|}\right)}  \sqrt{\log n}\right)\tag{Theorem~\ref{thm:directed-utility}} \\
        &\geq (1-O(\beta)) \density(G) - O\left(\frac{1}{\beta\epsilon} \sqrt{M \log(1/\delta) \log n} \sqrt{\max\left(\frac{|S^*|}{|T^*|}, \frac{|T^*|}{|S^*|}\right)}  \right) \tag{def of $\varsigma$} \\
        &= (1-O(\beta)) \density(G) - O\left(\frac{1}{\beta\epsilon} \sqrt{\log_2(n) \cdot (\log_{1+\beta}(2n^{3/2})+1) \cdot \log_{1+\beta}(n) \log(1/\delta) \log n} \sqrt{\max\left(\frac{|S^*|}{|T^*|}, \frac{|T^*|}{|S^*|}\right)}  \right) \\
        &\geq (1-O(\beta)) \density(G) - O\left(\frac{1}{\beta^2\epsilon} \log^2 n \sqrt{\log(1/\delta)}\sqrt{\max\left(\frac{|S^*|}{|T^*|}, \frac{|T^*|}{|S^*|}\right)}  \right),
    \end{align*}
    as claimed.
\end{proof}

\subsection{Improvement in the Centralized Setting}

To obtain an improvement in the directed setting similar to what we obtained in the weighted setting, we can essentially replace the call to Weighted-DSG-LEDP in Directed-DSG-LEDP to a call to Centralized-Weighted-DSG, and then apply Theorem~\ref{thm:papernot-steinke}.  In order to more fully utilize the directed setting, though, we will actually call Weighted-Noisy-Order-Packing-MWU instead.  We also need to compute a new estimate of the (directed) density of the returned solution, since the estimate from the weighted reduction could be quite inaccurate if we choose the wrong parameters ($t$ in particular).  This gives the following algorithm.

\begin{algorithm}
    \caption{Centralized-Directed-DSG-core($G, T,  \varsigma$)}
    \label{alg:centralized-directed-core}
    \begin{algorithmic}[1]
        \STATE Set 
        % $T = \left(\frac{\epsilon'}{4} n^{2}\right)^2$, and 
        $\tau = \sqrt{T}\varsigma$,
        % \frac{32\sqrt{T \log^2(5T/\delta')}}{\epsilon'}$, 
        and $N = \left \lceil \frac{2 n^{2}}{4(n+\tau) \sqrt{\log(n)/T}}\right\rceil$. 
        \STATE Sample $s', t' \sim [n]$ uniformly at random and set $t = \sqrt{s'/t'}$.
        \STATE Sample $k \sim [N]$ uniformly at random and set $\lambda = k \cdot 4(n+\tau) \sqrt{\frac{\log n}{T}}$.  
        \STATE Set $\alpha = \min(1/(2t), t / 2)$
        \STATE For each node $v$, compute the weight $c_{v}^{L} =  1/(2t \alpha)$ of its copy in $V_{t}^L$ and the weight $c_{v}^R = t / (2\alpha)$ of its copy in $V_{t}^R$.  Let $G'_{t}$ denote $G_{t}$ but with these weights (which are scaled up by exactly $1/\alpha$ from the weights in $G_{t}$)
        \STATE Let $\sigma \leftarrow$ Weighted-Noisy-Order-Packing-MWU($G'_t, \lambda, T, \tau$).
        \STATE Let $(U, \tilde \density(U)) \leftarrow $ Weighted-Peeling($G', \sigma, \varsigma$).
        \STATE Let $S' = U \cap V_{t}^L$ and $T' = U \cap V_{t}^R$.
        \STATE Let $\tilde \density(S', T') = \density(S',T')+ N\left(0,\varsigma^2\right)$
        \STATE Return $(S', T', \tilde \density(S,T))$.
    \end{algorithmic}
\end{algorithm}

The following lemma is proved exactly on the lines of  Lemma~\ref{lem:centralized-weighted-privacy}, together with the observation that the sampling of the values $s, t$ and index $k$  doesn't depend on any private information, i.e., the edges of $G$, and accounting for the privacy loss in line 9 of the algorithm:
\begin{lemma} \label{lem:centralized-directed-privacy}
    Centralized-Weighted-DSG-core($G, T, \varsigma)$ is $\frac{3}{2\varsigma^2}$-zCDP.
\end{lemma}

% \begin{lemma} \label{lem:centralized-directed-privacy}
%     Centralized-Directed-DSG-core($G, \epsilon', \delta'$) is $(\epsilon', \delta')$-DP.
% \end{lemma}
% \begin{proof}
%     We first claim that the computation of $\sigma$ and of $(U, \tilde \density(U))$ is $(\epsilon'/2, \delta'/2)$-DP.  This is essentially identical to the proof of Lemma~\ref{lem:centralized-weighted-privacy}, just using the fact that the maximum weight $\cmax$ in $G'$ is at most $n$ and that we have changed $T$ and $\tau$ by factors of $2$ (to get $(\epsilon'/2, \delta'/2)$-DP rather than $(\epsilon', \delta')$-DP).  The standard bound on the Gaussian mechanism (Lemma~\ref{lem:gaussian}) implies that the computation of $\tilde \density(S', T')$ is $(\epsilon/2, \delta/2)$-DP.  Then sequential composition implies that the full algorithm is $(\epsilon, \delta)$-DP.
% \end{proof}

\begin{lemma} \label{lem:directed-centralized-dsg-ledp-utility}
    Suppose $T$ and $\varsigma$ in Centralized-Directed-DSG-core are set so that $\tau = \sqrt{T}\varsigma \geq n$. Then with probability at least $\frac{1}{8Nn^2}$, Centralized-Weighted-DSG-core($G, T, \varsigma$) returns a pair of sets $(S', T')$ with $\tilde \density(S', T') -  O\left( \sqrt{\max\left(\tfrac{|S^*|}{|T^*|}, \tfrac{|T^*|}{|S^*|}\right)\log n} \cdot \varsigma\right)$.
\end{lemma}
\begin{proof}
We can now proceed with the utility analysis essentially as we did in the weighted case.  Let $S^*, T^*$ be the optimal solution, and let $t^* = \sqrt{|S^*|/|T^*|}$.  Then $t = t^*$ with probability at least $1/n^2$.  Let $\lambda^*$ denote the optimal density of $G'_{t^*}$, and note that by Lemma~\ref{lem:reduction1}, Lemma~\ref{lem:reduction3}, and the definition of $G'_{t^*}$ we have $\lambda^* = \alpha \density_G(S^*, T^*) = \alpha \density(G)$.  With probability $1/N$, the value $\lambda$ picked by Centralized-Directed-DSG-core satisfies $\lambda^* \leq \lambda < \lambda^* + 4(n+\tau)\sqrt{\frac{\log n}{T}}$.  

So the probability that both of these events occur is at least $1/(Nn^2)$.  Conditioned on this,  Theorem~\ref{thm:weighted-primal-solve-robust} and Lemma~\ref{lem:weighted-primal-to-discrete} imply that with probability at least $1/2$, Weighted-Noisy-Order-Packing-MWU outputs a permutation $\sigma$ such that 
\[\density(S^*_\sigma) \geq \lambda - 16(n+\tau) \sqrt{\frac{\log n}{T}} \geq \lambda^* - 16(n+\tau) \sqrt{\frac{\log n}{T}}\] 
since $\lambda \geq \lambda^*$. 

Thus, using Lemma~\ref{lem:weighted-peeling-approx}, we get that with probability at least $\frac{1}{4Nn^2}$, 
\begin{equation} \label{eq:centralized-directed-utility-weighted}
    \density(U) \geq \density(S^*_\sigma) - O(\sqrt{\log n} \varsigma) \geq \lambda^*  - O\left((n+\tau) \sqrt{\frac{\log n}{T}} + \sqrt{\log n} \varsigma\right) \geq \lambda^* - O\left(\sqrt{\log n} \varsigma\right)
\end{equation}
where the last inequality follows due to the specified value of $T$ and $\tau$.

Now we can reason about the density of the solution $(S', T')$ returned by the algorithm.  We have that with probability at least $\frac{1}{4Nn^2}$, 
\begin{align*}
    \density(S',T') &\geq \frac{1}{\alpha}\density(U) \tag{Lemma~\ref{lem:reduction1}, def of $G'_t$} \\
    &\geq \frac{1}{\alpha} \left(\lambda^*  - O\left(\sqrt{\log n} \varsigma\right)\right) \tag{Eq.~\eqref{eq:centralized-directed-utility-weighted}} \\
    &\geq \density(G) - O\left(\tfrac{1}{\alpha}\sqrt{\log n} \varsigma\right) \\
    &\geq \density(G) - O\left( \sqrt{\max\left(\tfrac{|S^*|}{|T^*|}, \tfrac{|T^*|}{|S^*|}\right)\log n} \cdot \varsigma\right). \tag{def of $\alpha$} 
    % \\
    % &\geq \density(G) - O\left( \frac{\sqrt{\max\left(\frac{|S^*|}{|T^*}|, \frac{|T^*|}{|S^*|}\right)}\log(n/\delta) \sqrt{\log n}}{\epsilon}\right). \tag{$\epsilon'$ and $\delta'$ in terms of $\epsilon$ and $\delta$}
\end{align*}

Now the standard concentration bound for the Gaussian mechanism (Lemma~\ref{lem:gaussian-concentration}) implies that with with probability at least $\frac{1}{8Nn^2}$, we not only have the above bound on $\density(S', T')$, but also have that
\begin{align*}
    \tilde \density(S', T') &\geq \density(G) -  O\left( \sqrt{\max\left(\tfrac{|S^*|}{|T^*|}, \tfrac{|T^*|}{|S^*|}\right)\log n} \cdot \varsigma\right). \qedhere
\end{align*}
\end{proof}

Now, let Centralized-Directed-DSG be the mechanism $\AQ$ obtained by applying the mechanism of Theorem~\ref{thm:papernot-steinke} with $Q = $ Centralized-Directed-DSG-core, $s = (S', T')$ (the pair of sets it outputs), $q = \tilde{\density}(S', T')$ (the estimated density it outputs), $q^* = \density(G) -  O\left( \sqrt{\max\left(\tfrac{|S^*|}{|T^*|}, \tfrac{|T^*|}{|S^*|}\right)\log n} \cdot \varsigma\right)$ as specified in Lemma~\ref{lem:directed-centralized-dsg-ledp-utility}, and $\gamma = n^{-c}$ for any given constant $c$.

\begin{theorem} \label{thm:directed-centralized-main}
    Let $\delta \in (0, 1)$ and $\epsilon > 0$ be given privacy parameters. Set $\varsigma = \frac{8\sqrt{\log(n^c/\delta)}}{\epsilon}$ and $T = \lceil \frac{n^2}{\varsigma^2}\rceil$ in Centralized-Weighted-DSG-core. Then Centralized-Weighted-DSG is $(\epsilon, \delta)$-DP, and with probability at least $1 - (8Nn^2+2)n^{-c}$, the density of the set it outputs is at least $\density(G) - O\left(\frac{\sqrt{\max\left(\tfrac{|S^*|}{|T^*|}, \tfrac{|T^*|}{|S^*|}\right)\log(n)\log(n/\delta)}}{\epsilon}\right)$.
\end{theorem}
\begin{proof}
    Using Lemma~\ref{lem:directed-centralized-dsg-ledp-utility}, we have
    \[\Pr_{(s, q) \sim Q(D)}[q \geq q^*] \geq \frac{1}{8Nn^2}.\]
    Since Centralized-Directed-DSG-core($G, \varsigma$) is $\frac{3}{2\varsigma^2}$-zCDP, applying Theorem~\ref{thm:papernot-steinke}, we conclude that Centralized-Weighted-DSG is 
    \[
    \left(\tfrac{8\sqrt{\log(1/(\gamma\delta))}}{\varsigma}, \delta\right)\text{-DP} = (\epsilon, \delta)\text{-DP}
    \]
    for the specified value of $\varsigma$, and the estimated density $\tilde\density(S)$ of the pair of sets $(S', T')$ it outputs is at least $q^*$ with probability at least $1 - 8Nn^2\gamma$. Now, recall the random variable $J$ used in Centralized-Weighted-DSG which is drawn from the geometric distribution with success probability $\gamma$. Since $\Pr[J > k] = (1-\gamma)^k \leq \exp(-\gamma k)$, we conclude that $\Pr[J \leq \frac{\log(1/\gamma)}{\gamma}] \geq 1 - \gamma$. Conditioned on $J \leq \frac{\log(1/\gamma)}{\gamma}$, using the Gaussian concentration bound (Lemma~\ref{lem:gaussian-concentration}) and a union bound over the $J$ calls to $Q$ in $\AQ$, we conclude that with probability at least $1-\gamma$, in each call to $Q$, we have
    \[|\density(S', T') - \tilde\density(S', T')| \leq O\left(\sqrt{\log n}\varsigma\right)\]
    for an appropriately chosen constant in the $O(\cdot)$ notation. Thus, overall, using the union bound, with probability at least $1-(8Nn^2+2)\gamma$, the true density of the pair of sets $(S', T')$ output by Centralized-Weighted-DSG is at least 
    \[\density(G) -  O\left( \sqrt{\max\left(\tfrac{|S^*|}{|T^*|}, \tfrac{|T^*|}{|S^*|}\right)\log n} \cdot \varsigma\right) = \density(G) - O\left(\frac{\sqrt{\max\left(\tfrac{|S^*|}{|T^*|}, \tfrac{|T^*|}{|S^*|}\right)\log(n)\log(n/\delta)}}{\epsilon}\right).\]
\end{proof}

% \begin{theorem}
%     Centralized-Directed-DSG is $(2\epsilon, 2\delta)$-DP, and the density of the set it outputs is at least $\density(G) -  O\left( \frac{\sqrt{\max\left(\frac{|S^*|}{|T^*}|, \frac{|T^*|}{|S^*|}\right)}\log(n/\delta) \sqrt{\log n}}{\epsilon}\right)$ with probability at least $1 - O\left(n^{4-c}\right) - \frac{\delta}{2}$.
% \end{theorem}
% \begin{proof}
%     The differential privacy guarantee is implied by Theorem~\ref{thm:liu-talwar} and Lemma~\ref{lem:centralized-directed-privacy}.  For utility, the above analysis implies that 
%     \[\Pr_{(s, q) \sim Q(D)}[q \geq q^*] \geq \frac{1}{16Kn^2}.\]
%     So Theorem~\ref{thm:liu-talwar} implies that the estimated density of the solution it outputs is at least $q^*$ with probability at least $1 - 16Kn^2 \gamma - \frac{\delta}{2} = 1 - O(n^{4-c}) - \frac{\delta}{2}$.  Now again using the Gaussian concentration bound (Lemma~\ref{lem:gaussian-concentration}) and a union bound over the at most $J_{\max} = n^c \log(2/\delta)$ calls to $Q$ in $\AQ$, we get that the \emph{true} density of the solution it outputs is at least $\density(G) -  O\left( \frac{\sqrt{\max\left(\frac{|S^*|}{|T^*}|, \frac{|T^*|}{|S^*|}\right)}\log(n/\delta) \sqrt{\log n}}{\epsilon}\right)$ with probability at least $1 - O(n^{4-c}) - \frac{\delta}{2}$ as claimed.  
% \end{proof}

\section{A Simple \texorpdfstring{$\epsilon$}{epsilon}-LEDP Algorithm} 
\label{sec:pure}
We now show that a simple modification of the parallel version~\cite{BKV12} of Charikar's algorithm~\cite{Charikar00} can be implemented in the $\epsilon$-LEDP model with only minor accuracy loss.  We will simply add noise at each iteration in order to keep things private, and then appeal to parallel composition.  We note that the the algorithm of~\cite{BKV12} requires computing the average degree in the remaining graph in order to set the appropriate threshold, and this can easily be done in the $\epsilon$-DP model (with noise added to keep it private), but cannot easily be done in the LEDP model since the curator cannot compute the average degree.  Instead, we directly use the noisy individual degrees to compute a noisy average degree.  This actually makes the running time and privacy analysis \emph{simpler}, but makes it slightly more difficult to show the approximation bound.

\begin{algorithm}
    \caption{Simple $\epsilon$-LEDP}
    \label{alg:simple}
    \begin{algorithmic}[1]
        \STATE $S_1 := V$, $i := 1$ (this is public, or equivalently in the zero'th round the curator makes $V$ public).
        \WHILE{$S_i \neq \emptyset$}
            \STATE For every vertex $v \in S_i$, let $D_i(v)$ be a noisy version of its remaining degree: $D_i(v) := d_{S_i}(v) + \Geom(e^{\epsilon})$
            \STATE Every vertex $v \in S_i$ sends $D_i(v)$ to the curator.
            \STATE The curator computes an estimate of $\density(S_i)$ by setting $\hat \density(S_i) = \frac{1}{|S_i|}\sum_{v \in S_i} D_i(v) / 2$.
            \STATE The curator computes a noise threshold $T_i$: %which is ``supposed'' to be $(1+\epsilon)d_{avg}(S_i) = 2(1+\epsilon) \density(S_i)$:
            \[
            T_i := (1+\eta) \cdot \frac{1}{|S_i|} \sum_{v \in S_i} D_i(v)
            \]
            \STATE The curator computes $L_i := \{v \in S_i : D_i(v) \leq T_i\}$
            \STATE The curator computes and makes public $S_{i+1} := S_i \setminus L_i$; $i = i+1$
        \ENDWHILE
        \STATE The curator selects whichever $S_i$ has largest $\hat \density(S_i)$.
    \end{algorithmic}
\end{algorithm}

We give our algorithm as Algorithm~\ref{alg:simple}.  It is easy to see that this algorithm can be implemented in the LEDP model, so we need to determine the running time (number of rounds), the privacy guarantee, and the approximation bound.  We begin with the running time, since it is simple and will be useful when arguing about privacy.

\begin{lemma} \label{lem:iterations}
The number of iterations is at most $O(\frac{1}{\eta}\log n)$.
\end{lemma}
\begin{proof}
    Consider iteration $i$.  By definition, $T_i$ is $(1+\eta)$ times the average of the $D_i(v)$ values.  So a simple averaging argument (i.e., Markov's inequality) implies that $|S_i \setminus L_i| \leq \frac{1}{1+\eta} |S_i|$.  Thus the total number of iterations is at most $\log_{1+\eta} n = O(\frac{1}{\eta} \log n)$, as claimed.
\end{proof}

\subsection{Privacy Analysis}
We now show that our algorithm satisfies $\epsilon$-LEDP.  

\begin{theorem} \label{thm:privacy-simple}
The algorithm satisfies $O\left(\frac{\epsilon}{\eta} \log n\right)$-LEDP.
\end{theorem}
\begin{proof}
Let $k$ be the number of iterations in the algorithm. 
We know from Lemma~\ref{lem:iterations} that $k = O\left(\frac{1}{\eta} \log n\right)$.  Consider the entire sequence of values and sets chosen by the algorithm; we'll show that this sequence satisfies the claimed differential privacy bound.  Slightly more formally, let $\Gamma_0 = (S_1)$, and let $\Gamma_i = (\{D_i(v)\}_{v \in S_i}, \hat \density(S_i), T_i, L_i, S_{i+1})$ be the sets and values constructed in iteration $i$.  We claim that the sequence $(\Gamma_i)_{i \in [k]}$ is $\frac{2\epsilon}{\eta} \log n$-DP.  This implies that the entire algorithm is $\frac{2\epsilon}{\eta} \log n$-LEDP, since the final step is post-processing so by Theorem~\ref{thm:post-processing} does not affect the privacy.  

We claim by induction that the prefix $(\Gamma_j)_{0 \leq j \leq i}$ is $i\epsilon$-DP for all $0 \leq i \leq k$.  Plugging in $i = k$ then implies the claim.  For the base case, note that $(\Gamma_0)$ is clearly $0$-DP, since for neighboring graphs $G, G'$ we know that $G$ and $G'$ have the same vertex set $V = S_1$.  

Now for the inductive step, consider some $1 \leq i \leq \ell$ and assume that the prefix $(\Gamma_j)_{0 \leq j \leq i-1}$ satisfies $(i-1)\epsilon$-DP.  The parallel composition theorem (Theorem~\ref{thm:parallel-comp}) and the Geometric Mechanism (Lemma~\ref{lem:geo-mechanism}) imply that $\{D_i(v)\}_{v \in S_i}$ is $2\epsilon$-DP: The sensitivity of the degree is $1$, but, as discussed earlier, every edge is in two parts of the partition. Everything in iteration $i$ after computing the $\{D_i(v)\}$ values is post-processing, so Theorem~\ref{thm:post-processing} implies that $\Gamma_i$ is $2\epsilon$-DP.  Finally, sequential composition (Theorem~\ref{thm:adseqcomp}) and the inductive hypothesis imply that $(\Gamma_j)_{0 \leq j \leq i}$ is $2\epsilon$-DP as required.
\end{proof}

\subsection{Approximation}
We begin by arguing that the noisy average we use $\hat d_i = \frac{1}{|S_i}| \sum_{v \in S_i} D_i(v)$ is concentrated around $d_{avg}(S_i)$.  Before we can do this, though, we will need an additional probabilistic tool to give concentration bounds for averages of independent sub-exponential random variables: Bernstein's inequality (Theorem 2.8.3 of~\cite{vershynin}).

\begin{definition}[Sub-exponential Random Variable: Definition 2.7.5 of~\cite{vershynin}]
    A random variable $X$ is \emph{sub-exponential} if the moment-generating function of $X$ is bounded at some point, i.e., if $\E[\exp(|X|/K)] \leq 2$ for some constant $K$.  The sub-exponential norm of $X$ is the smallest $K$ for which this is true: more formally, it is $\inf \{t > 0 : \E[\exp(|X|/t)] \leq 2\}$.
\end{definition}

\begin{lemma}[Bernstein's inequality, Corollary 2.8.3 of~\cite{vershynin}] \label{lem:Bernstein}
Let $X_1, \dots, X_N$ be independent, mean zero, sub-exponential random variables. Then, for every $t \geq 0$, we have
\[
\Pr\left[ \left| \frac{1}{N} \sum_{i=1}^N X_i \geq t \right| \right] \leq 2 \cdot \exp\left( -c N \min\left( \frac{t^2}{K^2}, \frac{t}{K} \right) \right)
\]
for some absolute constant $c > 0$, where $K$ is the maximum sub-exponential norm of any $X_i$.
\end{lemma}

With this tool, we will be able to prove concentration for $\hat d_i$.

\begin{lemma} \label{lem:noisy-avg-concentration}
Consider some iteration $i$.  Then $\Pr[ |\hat d_i - d_{avg}(S_i)| >  \Omega(\frac{1}{\epsilon} \log n)] \leq 1/n^4$
\end{lemma}
\begin{proof}
Clearly $\E[\hat d_i] = d_{avg}(S_i)$ by the symmetry of the added noise.  We also claim that the symmetric geometric distribution with parameter $e^{\epsilon}$ (i.e., the noise added to each degree) is sub-exponential with sub-exponential norm $\Theta(1/\epsilon)$.  This can easily be seen by the fact that the sub-exponential norm of a random variable $X$ is equivalent, up to a universal constant, to the value of $K$ such that $\Pr[|X| \geq t] \leq 2 \cdot \exp(-t/K)$ (see~\cite[Proposition 2.7.1]{vershynin}. 
So if we set $t = \frac{1}{\epsilon} \log \frac{1}{\sigma}$, Lemma~\ref{lem:geo-mechanism} implies that $\Pr[|\Geom(e^{\epsilon})| \geq t] \leq \sigma = e^{-\epsilon t}$.  Thus the sub-exponential norm of $\Geom(e^{\epsilon})$ is $\Theta(1/\epsilon)$.   
 So Bernstein's inequality implies that
\begin{align*}
    \Pr\left[ |\hat d_i - d_{avg}(S_i)| \geq  t \right] \leq 2 \cdot \exp\left( -c |S_i| \min\left( \epsilon^2 t^2, \epsilon t\right)\right).
\end{align*}

So if we set $t = \Theta(\frac{1}{\epsilon} \log n)$, we get that $\Pr[ |\hat d_i - d_{avg}(S_i)| \geq \Omega(\frac{1}{\epsilon} \log n)] \leq 1/n^4$ as desired.
\end{proof}

We now show that in each iteration, nodes with large degree are likely to survive to the next iteration.

\begin{lemma} \label{lem:LEDP-low-degree} 
Consider some iteration $i$.  Then there is some constant $c > 0$ such that for every $v \in S_i$, if
\[
d_{S_i}(v) > (1+\eta) d_{avg}(S_i) + c \frac{1+\eta}{\epsilon} \log n
\]
then $\Pr[v \in L_i] \leq 1/n^3$.
\end{lemma}
\begin{proof}
Let $v \in S_i$ be a node with $d_{S_i}(v) > (1+\eta) d_{avg}(S_i) + \Omega\left( \frac{1+\eta}{\epsilon} \log n\right)$.  Lemma~\ref{lem:noisy-avg-concentration} implies that $\hat d_i \leq d_{avg}(S_i) + O(\frac{1}{\epsilon} \log n)$ with probability at least $1-1/n^4$.  So we may assume that this event occurs by adding $1/n^4$ to our failure probability.  

By definition, $v \in L_i$ only if
\begin{align*}
    d_{S_i}(v) + \Geom(e^{\epsilon}) &\leq T_i = (1+\eta) \cdot \frac{1}{|S_i|} \sum_{v \in S_i} D_i(v) = (1+\eta) \hat d_i \\
    &\leq (1+\eta) \cdot \left(d_{avg}(S_i) + O\left(\frac{1}{\epsilon} \log n\right)\right).
\end{align*}

By our assumption on $d_{S_i}(v)$, this occurs only if
\begin{align*}
    \Geom(e^{\epsilon}) & \leq (1+\eta) \cdot \left(d_{avg}(S_i) + O\left(\frac{1}{\epsilon} \log n\right)\right) - d_{S_i}(v) \\
    &\leq (1+\eta) \cdot O\left(\frac{1}{\epsilon} \log n \right) - c\frac{1+\eta}{\epsilon} \log n.
\end{align*}

By setting $c$ to be a large enough constant (say, a large constant factor larger than the constant hidden in the $O(\cdot)$ notation in the above inequality), we get that $v \in L_i$ only if $\Geom(e^{\epsilon_1}) \leq -c' \frac{1+\eta}{\epsilon} \log n$ for any constant $c'$ that we want.  Lemma~\ref{lem:geo-mechanism} now implies that this happens with probability at most $1/n^4$.  Hence our total probability of $v$ being in $L_i$ is at most $1/n^4 + 1/n^4 \leq 1/n^3$, as claimed.
\end{proof}

This lemma immediately gives the following corollary.
\begin{corollary} \label{cor:degree}
With probability at least $1-1/n$, every $L_i$ consists entirely of nodes with $d_{S_i}(v) \leq (1+\eta)d_{avg}(S_i) + O\left(\frac{1+\eta}{\epsilon} \log n \right)$.
\end{corollary}
\begin{proof}
Use Lemma~\ref{lem:LEDP-low-degree} on every vertex in every iteration, using Lemma~\ref{lem:iterations} to bound the number of iterations, and take a union bound.
\end{proof}

We can now prove the main approximation bound.

\begin{theorem} \label{thm:simple-approx}
With probability at least $1-2/n$, the density of the subset returned by the algorithm is at least 
\[
\frac{OPT}{2(1+\eta)} - O\left(\frac{1}{\epsilon} \log n \right).
\]
\end{theorem}
\begin{proof}
We first argue that, like in the non-DP case~\cite{BKV12}, the best of the iterations is a good approximation.  Then we argue that we return a solution that is essentially as good.

Let $\Lambda = O\left(\frac{1+\eta}{\epsilon} \log n \right)$ be the additive loss from Corollary~\ref{cor:degree}.  We know from Corollary~\ref{cor:degree} that with probability at least $1-1/n$ when every node is removed it has degree at most $d_{S_i}(v) \leq (1+\eta)d_{avg}(S_i) + \Lambda$ (where $i$ is the iteration in which $v$ is removed).  Now orient every edge towards whichever of the endpoints is removed earlier, making an arbitrary choice if both endpoints are removed in the same iteration.  So the in-degree in this orientation of every node removed in iteration $i$ is at most $(1+\eta)d_{avg}(S_i) + \Lambda$.  Lemma 3 of~\cite{Charikar00} implies that the density of $G$ is upper bounded by the maximum in-degree in any orientation (in particular the above orientation), i.e., $\density(G) \leq \max_{v \in V} d_{in}(v)$.  Hence we have that
\begin{align} \label{eq:simple-density}
    \density(G) &\leq \max_{v \in V} d_{in}(v) \leq \max_i ((1+\eta)d_{avg}(S_i) + \Lambda) \leq 2(1+\eta) \max_i \density(S_i) + \Lambda.
\end{align}

So we know the best of the $k = O(\frac{1}{\eta} \log n)$ subgraphs is a good approximation (where we are using Lemma~\ref{lem:iterations} to bound the number of iterations).  However, we do not return the $S_i$ with maximum $\density(S_i)$, but rather the $S_i$ with maximum estimated density $\hat \density(S_i)$.  But Lemma~\ref{lem:noisy-avg-concentration} implies that $\Pr[|\hat \density (S_i) - \density(S_i)| \geq \Omega\left(\frac{1}{\epsilon} \log n\right)] \leq 1/n^3$ for each iteration $i$.  So a trivial union bound implies that $|\hat \density (S_i) - \density(S_i)| \leq O\left(\frac{1}{\epsilon} \log n\right)$ for all $i$ with probability at least $1-1/n^2$.  If this happens, then combined with~\eqref{eq:simple-density} we get that
\begin{align*}
    \density(G) \leq 2(1+\eta) \max_i \density(S_i) + \Lambda \leq 2(1+\eta) \max_i \hat \density(S_i) + 2\Lambda.
\end{align*}
This clearly implies the theorem.
\end{proof}

\paragraph{Putting it all together.} We can now combine all of this  into one easy-to-use corollary.
\begin{corollary} \label{cor:simple-main}
Let $\epsilon, \eta > 0$.  There is an $\epsilon$-LEDP algorithm for Densest Subgraph which runs in $O\left(\frac{1}{\eta} \log n\right)$ rounds and returns a set $S \subseteq V$ such that 
\[
\density(G) \leq 2(1+\eta) \density(S) + O\left( \frac{1}{\epsilon \eta} \log^2 n \right)
\]
with probability at least $1-2/n$.
\end{corollary}
\begin{proof}
Use our algorithm but use DP parameter setting $\epsilon' = \Theta(\epsilon \eta / \log n)$.  Then the number of rounds is implied by Lemma~\ref{lem:iterations}, Theorem~\ref{thm:privacy-simple} implies that the algorithm is $\epsilon$-LEDP, and Theorem~\ref{thm:simple-approx} implies the density bound.
\end{proof}

%\input{improved-approx}

%!TEX root=./main.tex

\section{Private Density Approximation}
In this section we show that if we want to return the \emph{density} of the densest subgraph $\density(G)$, rather than the set of nodes itself, then it is straightforward to have accuracy $O(\sqrt{1/\epsilon})$ in expectation or $O\left(\sqrt{\frac{\log n}{\epsilon}}\right)$ with high probability in the centralized edge-DP model.  Notably, the lower bound of~\cite{NV21} implies that if we want to output a \emph{set} $S \subseteq V$ of maximum density, then our expected additive loss must be at least $\Omega(\sqrt{\log n / \epsilon})$.  Hence our upper bound implies that the expected additive loss must be strictly larger for outputting the set than for outputting its density. 

To do this, we can use a variant of the propose-test-release~\cite{DL09} mechanism which takes advantage of the fact that a graph can only have large density if its sensitivity is low.  This allows us to actually simplify the framework by just using a function with smaller global sensitivity to approximate the density.

More formally, for $x \in \mathbb{R}^+$, let $\density_x(G) = \max(\density(G), x)$ (where recall that $\density(S) = |E(S)|/|S|$ and $\density(G) = \max_{S \subset V} \density(S)$).  We will compute a differentially private approximation to $\density_x(G)$ for appropriate $x$, and then claim that this is a good approximation to $\density(G)$.  Let us first analyze the sensitivity of the $\density_x$ function.

\begin{lemma} \label{lem:global-sensitivity}
    The sensitivity of $\density_x$ is at most $\frac{1}{2x-1}$.
\end{lemma}
\begin{proof}
    Let $G$ and $G'$ be two graphs that differ in exactly one edge $e = \{u,v\}$, which without loss of generality is contained in $G'$ and not contained in $G$.  We break into two cases depending on $\density(G)$.

    First, suppose that $\density(G) \leq x-1$.  Then clearly $\density(G') \leq x$, and hence $\density_x(G) = \density_x(G') = x$.  Thus the sensitivity is $0$.

    For the more interesting case, suppose that $\density(G) > x-1$.  Let $S \subseteq V$ be the densest subgraph of $G'$.  If $\{u,v\} \not\subseteq S$ then $S$ has the same density in both $G$ and $G'$ and no other set has larger density in $G$ than in $G'$.  Hence $\density(G) = \density(G')$ and so $\density_x(G) = \density_x(G')$ and we are done.  So assume without loss of generality that $u,v \in S$.  Then $\density(G) \geq \frac{E_{G'}(S) - 1}{|S|} = \density(G') - \frac{1}{|S|}$.  Moreover, since $|E(S)| \leq \binom{|S|}{2}$ for any $S \subseteq V$, we have that $\density(G') = \density(S) \leq \binom{|S|}{2} / |S| = \frac{|S|-1}{2}$, and so $|S| \geq 2 \density(G') + 1 \geq 2 \density(G) + 1 > 2x-1$.  Hence 
    \begin{equation*}
        \density(G) \geq \density(G') - \frac{1}{2x-1},
    \end{equation*}
    and thus $\density_x(G') - \density_x(G) \leq \frac{1}{2x-1}$ as claimed.
\end{proof}

Now we can define our algorithm: we use the Laplace mechanism on $\density_x(G)$ (see Lemma~\ref{lem:laplace} and~\cite{DworkMNS06}).  Slightly more formally, we compute $\density_x(G)$, draw noise $N \sim \text{Lap}\left( \frac{1}{(2x-1)\epsilon}\right)$, and return $\widehat \density_x(G) = \density_x(G) + N$.

\begin{lemma}
    This algorithm is $\epsilon$-edge DP.
\end{lemma}
\begin{proof}
    This follows directly from Lemma~\ref{lem:global-sensitivity} and the standard analysis of the Laplace mechanism (Lemma~\ref{lem:laplace}).  
\end{proof}

\begin{lemma}
    The expectation of $|\widehat \density_x(G) - \density(G)|$ is at most $O\left( \frac{1}{(2x-1)\epsilon}\right) + x$.  And with high probability, $|\widehat \density_x(G) - \density(G)| \leq O\left( \frac{\log n}{(2x-1)\epsilon}\right) + x$. 
\end{lemma}
\begin{proof}
    Clearly $\E[|\widehat \density_x(G) - \density(G)|]\leq x + \E[|N|] = O\left( \frac{1}{(2x-1)\epsilon}\right) + x$ (see Lemma~\ref{lem:laplace}).
    Moreover, standard tail bounds for the Laplace distribution imply that $|N| \leq O\left( \frac{\log n}{(2x-1)\epsilon}\right)$ with high probability.  If this event occurs, then we have that 
    \begin{align*}
        |\widehat \density_x(G) - \density(G)| &\leq N + \density_x(G) - \density(G) \leq  O\left( \frac{\log n}{(2x-1)\epsilon}\right) + x
    \end{align*}
    as claimed.
\end{proof}

It immediately follows that if we set $x = \Theta\left( \sqrt{\frac{\log n}{\epsilon}}\right)$, then with high probability $|\widehat \density_x(G) - \density(G)| \leq O\left(  \sqrt{\frac{\log n}{\epsilon}} \right)$, so we have the following corollary:
\begin{corollary}
    There is an $\epsilon$-edge DP algorithm which  outputs a value $\widehat \density$ such that $|\widehat \density - \density(G)| \leq O\left(  \sqrt{\frac{\log n}{\epsilon}} \right)$ with high probability.
\end{corollary}

On the other hand, if we set $x = \Theta(\sqrt{1/\epsilon})$, we get the following corollary.
\begin{corollary}
    There is an $\epsilon$-edge DP algorithm which  outputs a value $\widehat \density$ such that $\E[|\widehat \density - \density(G)|] \leq O\left(  \sqrt{\frac{1}{\epsilon}} \right)$.

\end{corollary}

%\input{private-hedge}

% \section*{Acknowledgments}

% The authors wish to thank Abhradeep Guha Thakurta and Kunal Talwar for helpful discussions.

\bibliographystyle{abbrv}
\bibliography{refs}

\newpage
\appendix
%!TEX root=./main.tex
\section{Differential Privacy Basics} \label{app:DP}

In this appendix we provide details about the basic DP definitions and primitives that we use.  We begin with a formal discussion of the LEDP model.

\subsection{Local Edge Differential Privacy}
We use the formalization of \emph{local} edge-differential privacy from~\cite{dhulipala2022differential}, suitably modified to handle $(\epsilon, \delta)$-DP rather than just pure $\epsilon$-DP.  In this formalization, LEDP algorithms are described in terms of an (untrusted) curator, who does not have access to the graph’s edges, and individual nodes. During each round, the curator first queries a set of nodes for information.   Individual nodes, which have access only to their own (private) adjacency lists (and whatever information was sent by the curator), then release information via local randomizers, defined next.

\begin{definition}[Local Randomizer~\cite{dhulipala2022differential}]
    An \emph{$(\epsilon, \delta)$-local randomizer} $R : a \rightarrow \mathcal Y$ for node $v$ is an $(\epsilon,\delta)$-edge DP algorithm that takes as input the set of its neighbors $N(v)$, represented by an adjacency list $a = (b_1, \dots, b_{|N(v)|})$. 
\end{definition}

In other words, a local randomizer is just an edge-DP algorithm where the private input is the adjacency list of $v$.  Note that such an algorithm can use other (public) information, such as whatever has been broadcast by the curator.  The information released via local randomizers is public to all nodes and the curator. The curator performs some computation on the released information and makes the result public. The overall computation is formalized via the notion of the transcript.

\begin{definition}[LEDP \cite{dhulipala2022differential}] \label{def:LEDP}
    A \emph{transcript} $\pi$ is a vector consisting of $4$-tuples $(S_U^t, S_R^t, S_{\epsilon}^t, S_Y^t)$ encoding the set of parties chosen, set of randomizers assigned, set of randomizer privacy parameters, and set of randomized outputs produced for each round $t$.  Let $S_{\pi}$ be the collection of all transcripts, and $S_R$ be the collection of all randomizers.  Let $\bot$ denote a special character indicating that the computation halts. A \emph{protocol} is an algorithm $\mathcal A : S_{\pi} \rightarrow (2^{[n]} \times 2^{S_R} \times 2^{\mathbb{R}^{\geq 0}} \times 2^{\mathbb{R}^{\geq 0}}) \cup \{\bot\}$ mapping transcripts to sets of parties, randomizers, and randomizer privacy parameters.  The length of the transcript, as indexed by $t$, is its round complexity.

    Given $\epsilon, \delta \geq 0$, a randomized protocol $\mathcal A$ on a (distributed) graph $G$ is $(\epsilon, \delta)$-locally edge differentially private ($(\epsilon, \delta)$-LEDP) if the algorithm that outputs the entire transcript generated by $\mathcal A$ is $(\epsilon, \delta)$-edge differentially private on graph $G$.
\end{definition}

As in~\cite{dhulipala2022differential}, we assume each user
can see the public information for each round on a public ``bulletin board''.  If an algorithm is $(\epsilon, 0)$-LEDP then we say that it is $\epsilon$-LEDP.  

This definition is somewhat unwieldy, and moreover does not neatly align with the intuitive explanation given in terms of an untrusted curator (the curator does not appear at all in the above definition of a protocol).  Fortunately, it is not hard to see the correspondence.  This was implicit in~\cite{dhulipala2022differential}, but we make it explicit here.

Suppose that we have an untrusted curator which initially knows only $V$, and each node initially knows only $V$ and its incident edges.  Suppose that every node runs an $(\epsilon, \delta)$-edge DP algorithm in every round, and broadcasts the output of this algorithm.  Then in each round, the algorithm run by a node is aware of all of the public information and the (private) local edge information.  If we think of this as an algorithm which has the public information hard-coded in and which takes the local edge information as input, then this is an $(\epsilon, \delta)$-local randomizer.  So the local randomizer in every round is exactly a function of the previous public information, as required.  And the curator's choice of who to query and with what question is a function from the previous transcript to a set of parties and randomizers, and hence this satisfies the definition of a ``protocol'' from Definition~\ref{def:LEDP}.  Hence we will give our algorithms in terms of a multi-round algorithm with an untrusted curator and $(\epsilon, \delta)$-edge DP algorithms at each node, rather than through local randomizers and protocols.

\subsection{Useful DP Primitives}
We first give the simple and well-known differential private mechanisms that form the building blocks of our algorithms, and then the various composition theorems that we will use to combine them.

\subsubsection{Basic Mechanisms} 
For a function $f$ from datasets to $\mathbb{R}$, let $\Delta_f$ denote the global sensitivity of $f$, i.e., the maximum over neighboring datasets $X, X'$ of $f(X) - f(X')$.  All of the basic mechanisms will depend on $\Delta_f$.

The first mechanism is the standard Laplace mechanism.

\begin{lemma}[Laplace Mechanism~\cite{DworkMNS06}] \label{lem:laplace}
Adding noise drawn from the Laplace distribution with parameter $\Delta_f / \epsilon$ satisfies $\epsilon$-differential privacy.  Moreover, the expected additive loss of the Laplace mechanism is $O(\Delta_f / \epsilon)$, and with high probability the loss is at most $O((\Delta_f / \epsilon) \log n)$.
\end{lemma}

The next basic mechanism is essentially the discrete analog of the standard Laplace mechanism.

\begin{definition}[Symmetric geometric distribution~\cite{GRS09,BV18}] Let $\gamma > 1$. The symmetric geometric distribution $\Geom(\gamma)$ takes integer values such that the probability mass function at $k$ is $\frac{\gamma-1}{\gamma+1} \cdot \gamma^{-|k|}$.
\end{definition}

\begin{lemma}[Geometric Mechanism (see~\cite{FHS22,BV18,GRS09})] \label{lem:geo-mechanism}
The Geometric Mechanism for query $f$ is the function $f'(x) + \Geom(\exp(\epsilon/\Delta_f))$.  The geometric mechanism satisfies $\epsilon$-DP.  Moreover, for every input $x$ and $\sigma \in (0,1)$, the error $|f(x) - f'(x)|$ of the geometric mechanism is at most $O(\frac{\Delta_f}{\epsilon} \cdot \log\frac{1}{\sigma})$ with probability at least $1-\sigma$.
\end{lemma}

%The next mechanism is the basic allows us to add less noise, but at the price of achieving $(\epsilon, \delta)$-DP rather than $\epsilon$-DP.

% \begin{lemma}[Gaussian Mechanism~\protect{\cite[Theorem A.1]{DR14}}] \label{lem:gaussian}
% For a function $f$, let $\Delta_f$ denote the global sensitivity of $f$. Let $\epsilon \in (0,1)$ be arbitrary.  For $c^2 > 2 \log(1.25 / \delta)$, adding noise drawn from the Gaussian distribution with standard deviation at least $c \Delta_f / \epsilon$ is $(\epsilon, \delta)$-differentially private.
% \end{lemma}

The following mechanism allows us to add noise to achieve zCDP rather than pure DP.

\begin{lemma}[Gaussian Mechanism for zCDP~\cite{bun2016concentrated}] \label{lem:gaussian}
Let $\rho \in (0,1)$ be arbitrary.  Adding noise drawn from $N(0, \sigma^2)$ is $(\Delta_f^2 / 2\sigma^2)$-zCDP.
\end{lemma}

The following standard concentration bound for useful will be useful when analyzing the Gaussian Mechanism.

\begin{lemma} \label{lem:gaussian-concentration}
    Let $X \sim N(0, \sigma^2)$.  Then $\Pr[|X - \E[X]| \geq t] \leq 2\exp(-t^2 / \sigma^2)$
\end{lemma}

\subsubsection{Composition Theorems}  
We will need a few different composition theorems.  All of them are standard and well-known, and can be found in standard textbooks~\cite{DR14,Vad17}.  

\begin{theorem}[Adaptive Sequential Composition~\cite{AdSeqComp,mcsherry09,bun2016concentrated}] \label{thm:adseqcomp}
Suppose $M = (M_1, \dots, M_k)$ is a sequence of $(\epsilon, \delta)$-differentially private algorithms, potentially chosen sequentially and adaptively. Then $M$ is $(k\epsilon, k\delta)$-differentially private.  If each $M_i$ is $\rho$-zCDP, then $M$ is $k\rho$-zCDP.
\end{theorem}

\begin{theorem}[Advanced Composition~\protect{\cite[Corollary 3.21]{DR14}}] \label{thm:adv-composition}
Given target privacy parameters $0 < \epsilon' < 1$ and $\delta' > 0$, to ensure $(\epsilon', k\delta + \delta')$ cumulative privacy loss over $k$ mechanisms, it suffices that each mechanism is $\left(\frac{\epsilon'}{2\sqrt{2k  \log(1/\delta')}}, \delta \right)$-differentially private.
\end{theorem}

\begin{theorem}[Parallel Composition~\cite{mcsherry09}] \label{thm:parallel-comp}
Let $M_i$ each provide $(\epsilon, \delta)$-DP.  Let $D_i$ be arbitrary disjoint subsets of the input domain $D$, and let $X \subseteq D$ be a database.  Then the sequence of $M_i(X \cap D_i)$ is $(\epsilon,\delta)$-differentially private.  If each $M_i$ is $\rho$-zCDP, then the sequence of $M_i(X \cap D_i)$ is $\rho$-zCDP.
\end{theorem}

In a few places we will use a slight extension of parallel composition: if every element of $D$ is in at most \emph{two} of the $D_i$ sets (rather than at most $1$, in the disjoint setting), then the sequence of $(M_i(X \cap D_i))$ is $(2\epsilon, 2\delta)$-DP.  This is a straightforward consequence of Theorem~\ref{thm:adseqcomp} and Theorem~\ref{thm:parallel-comp}.%\mdnote{Is this enough?}

\begin{theorem}[Post-Processing] \label{thm:post-processing}
Let $M$ be an $(\epsilon, \delta)$-differentially private mechanism, and let $f$ be a randomized or deterministic function whose domain is the range of $M$.  Then the randomized function $g(x) =f(M(x))$ is $(\epsilon, \delta)$-differentially private.  If $M$ is $\rho$-zCDP, then $g$ is $\rho$-zCDP.
\end{theorem}

\subsubsection{Connection to LEDP}  
All of the previous composition theorems are phrased in terms of centralized differential privacy.  It is not hard, however, to see that they continue to hold in the LEDP model.%, and in fact that combining the basic DP mechanisms with these composition theorems gives a much simpler and cleaner way of proving $(\epsilon, \delta)$-LEDP than directly reasoning about transcripts.  

To see this, consider an algorithm in the LEDP model: an untrusted curator knowing only public information ($V$, and whatever is output by the nodes and the curator itself along the way), and an algorithm at each node that has access to the public information and its incident edges.  Suppose that each node runs an $(\epsilon, \delta)$-edge DP algorithm on its incident edges.  Then when the computation for this round is viewed as a whole, every edge is in the private information of exactly two nodes, and hence parallel composition and sequential composition imply that the combined output of a round (the collection of all outputs from all nodes) is $(2\epsilon, 2\delta)$-edge DP.  The computation done at the curator is simply post-processing (since it does not access any private information directly), so by Theorem~\ref{thm:post-processing} we can view each round as a $(2\epsilon, 2\delta)$-edge DP algorithm in the traditional centralized DP model, and then sequential composition and advanced composition apply as usual.  Since we will not care about constant factors in the privacy parameters, we will simply treat each round as being $(\epsilon, \delta)$-edge DP (we can always go back and set our true privacy parameters to half of whatever we would set them to).

\section{Proofs from Section~\ref{sec:DP-hedge}} \label{app:DP-hedge}

\begin{proof}[Proof of Theorem~\ref{thm:MWU-main}]
As in the standard analysis of the Multiplicative Weights method, we use $\Phi^{(t)} = \log(\sum_{i=1}^n w^{(t)}_i)$ as a potential function and track its evolution over time. We have
\begin{align*}
\Phi^{(t+1)} = \log \left(\langle \exp(-\eta \hat{m}^{(t)}), w^{(t)}\rangle\right) = \Phi^{(t)} + \log\left( \langle \exp(-\eta \hat{m}^{(t)}), p^{(t)}\rangle\right) 
% \\
% &= \Phi^{(t)} \cdot \left(\langle \exp(-\eta \hat{m}^{(t)}) - \vec{1} - \eta \hat{m}^{(t)}, p^{(t)}\rangle + 1 + \eta \langle \hat{m}^{(t)}, p^{(t)}\rangle\right) 
% \\
% &\leq \Phi^{(t)} \exp\left(\exp(-\eta \hat{m}^{(t)}) - \vec{1} - \eta \hat{m}^{(t)}, p^{(t)}\rangle + \eta \langle \hat{m}^{(t)}, p^{(t)}\rangle\right).
\end{align*}
In the following, let $\vec{1}$ denote the all 1's vector, and for a vector $v$, let $v^2$ denote the vector obtained by squaring $v$ coordinate-wise. Let ${\E}^{(t)}[\cdot]$ denote the expectation conditioned on all the randomness up to and including round $t$. Then we have
\begin{align*}
{\E}^{(t)}[\Phi^{(t+1)}] &= {\E}^{(t)}\left[\Phi^{(t)} + \log\left(\langle \exp(-\eta \hat{m}^{(t)}), p^{(t)}\rangle \right)\right]\\
&\leq \Phi^{(t)} + \log {\E}^{(t)}\left[\left(\langle \exp(-\eta \hat{m}^{(t)}), p^{(t)}\rangle \right)\right] \qquad \text{(Jensen's inequality: $\log \E[X] \geq \E[\log X]$)}\\
&= \Phi^{(t)} + \log \left(\left\langle {\E}^{(t)}\left[\exp(-\eta \hat{m}^{(t)})\right], p^{(t)}\right\rangle \right) \qquad \text{($p^{(t)}$ is deterministic conditioned on the past)}\\
&= \Phi^{(t)} + \log\left(\langle \exp\left(-\eta m^{(t)} + \tfrac{\eta^2\nu^2}{2}\vec{1}\right), p^{(t)}\rangle\right) \qquad \text{(using the formula for the mgf of a normal distribution)} \\
&\leq \Phi^{(t)} + \log\left(\left\langle \vec{1} -\eta m^{(t)} + \tfrac{\eta^2\nu^2}{2}\vec{1} + (-\eta m^{(t)} + \tfrac{\eta^2\nu^2}{2}\vec{1})^2, p^{(t)}\right\rangle \right) \\
& \qquad \text{(since $|m^{(t)}_i|, \nu, \eta \in [0, 1]$, and using the fact that $\exp(x) \leq 1 + x + x^2$ for $|x| \leq 1.5$)} \\
&\leq \Phi^{(t)} + \log(1 - \eta \langle m^{(t)}, p^{(t)}\rangle + 3\eta^2) \qquad \text{(since $|m^{(t)}_i|, \nu, \eta \in [0, 1]$)} \\
&\leq \Phi^{(t)} - \eta \langle m^{(t)}, p^{(t)}\rangle + 3\eta^2, 
\end{align*}
where the final inequality follows from the fact that $\log(1+x) \leq x$ for all $x > -1$. Taking expectations on both sides of the above inequality over the randomness up to and including round $t$, we get
\[\E[\Phi^{(t+1)}] \leq \E[\Phi^{(t)}] - \eta \E[\langle m^{(t)}, p^{(t)}\rangle] + 3\eta^2.\]
Thus, by induction, using the fact that $\Phi^{(1)} = \log(n)$, we have
\[\E[\Phi^{(T+1)}] \leq \log(n) - \eta \sum_{t=1}^T \E[\langle m^{(t)}, p^{(t)}\rangle] + 3\eta^2 T.\]
On the other hand, for any given index $i$, we have
\[\Phi^{(T+1)} = \log \left(\textstyle\sum_{i'} \exp\left(\textstyle\sum_{t=1}^T -\eta \hat{m}^{(t)}_{i'}\right)\right) \geq -\eta \sum_{t=1}^T \hat{m}^{(t)}_i,\]
and hence
\[\E[\Phi^{(T+1)}] \geq -\eta \E[\textstyle \sum_{t=1}^T \hat{m}^{(t)}_i] = -\eta \E[\textstyle \sum_{t=1}^T m^{(t)}_i].\]
Putting the above inequalities together, and simplifying, we get 
\[
\E\left[\sum_{t=1}^T \langle m^{(t)}, p^{(t)} \rangle\right] \leq \E\left[\sum_{t=1}^T m_i^{(t)}\right] + 3\eta T + \frac{\log n}{\eta}.
\]
Using the value $\eta = \sqrt{\frac{\log n}{T}}$, we get the stated regret bound.
\end{proof}

\end{document}